\documentclass{article}

\usepackage[top=1in, bottom=1in, left=1in, right=1in]{geometry}
\usepackage[utf8]{inputenc}

\usepackage{times}
\usepackage{comment}
\usepackage{enumitem}
\usepackage{amsmath, amsthm, amssymb, natbib, graphicx, url, algorithm2e}
\usepackage{bbm}
\usepackage{xcolor}
\usepackage{braket}
\usepackage{hyperref} 
\usepackage{parskip}
\usepackage{url}

\usepackage[nameinlink,capitalize]{cleveref}
\definecolor{dgreen}{rgb}{0,0.5,0}
\hypersetup{
  colorlinks=true,
  linkcolor=blue,
  filecolor=blue,
  citecolor = dgreen,      
  urlcolor=cyan,
}

\theoremstyle{plain}
\newtheorem{theorem}{Theorem}
\newtheorem*{theorem*}{Theorem}

\newtheorem{lemma}[theorem]{Lemma}
\newtheorem{corollary}[theorem]{Corollary}

\newtheorem{proposition}[theorem]{Proposition}

\theoremstyle{definition}
\newtheorem{definition}[theorem]{Definition}

\newtheorem{remark}[theorem]{Remark}

\newcommand{\R}{\mathbb{R}}
\newcommand{\Z}{\mathbb{Z}}
\newcommand{\N}{\mathbb{N}}
\newcommand{\C}{\mathbb{C}}

\newcommand{\mS}{\mathbb{S}}

\newcommand{\poly}{\text{poly}}

\DeclareMathOperator{\sech}{sech}

\newcommand{\E}{\mathbb{E}}
\newcommand{\Tr}{\operatorname{Tr}}

\newcommand{\Cov}{\operatorname{Cov}}

\newcommand{\op}{\mathrm{op}}

\newcommand{\1}{\mathbf{1}}

\newcommand{\tilt}{\mathcal{T}}

\newcommand{\lineartilt}{T}

\newcommand{\NP}{\textnormal{NP}}
\newcommand{\BQP}{\textnormal{BQP}}
\newcommand{\RP}{\textnormal{RP}}
\newcommand{\norm}[1]{\| #1\|}

\begin{document}
\title{A computational phase diagram for the transverse field Ising model}
\author{Thuy-Duong Vuong\thanks{Department of Computer Science and Engineering, UC San Diego. \texttt{thvuong@ucsd.edu}}}
\maketitle

\begin{abstract}
We study the transverse field Ising model, defined by the Hamiltonian $H =\frac{1}{2}\sum_{i, j\in [n]} J_{ij} Z_i Z_j +\sum_{i=1}^n h_i^z Z_i + \eta\sum_{i} X_i$ where $J $ is the symmetric interaction matrix, and $\eta$ is the transverse field strength. Let $\Delta(J)=\lambda_{\max}(J)-\lambda_{\min}(J)$ be the spectral width of $J.$
When the inverse temperature $\beta\geq0$ satisfies $\Delta(J)\cdot\frac{\tanh(\beta\eta)}{\eta}\leq1$, we give a randomized classical algorithm that approximates the partition function $Z(\beta)=\Tr(e^{-\beta H})$ to a given relative error $\epsilon\in(0,1)$ in time polynomial in $n$, $\beta$, the model parameters, and $\epsilon^{-1}$. When
$ \Delta(J) \cdot \frac{\tanh(\beta \eta)}{\eta} > 1 ,$ we show that approximating $ Z(\beta)$ within an $\exp(o(n))$-multiplicative factor is $\textsc{NP}$-hard, and thus unlikely to admit an efficient classical or quantum algorithms under standard complexity theoretic assumptions. Furthermore, in the regime $\Delta (J)\cdot \frac{\tanh(\beta \eta)}{\eta}\leq 1,$ we provide an efficient randomized classical algorithm that approximates Pauli string observables of the Gibbs state $ \rho_\beta = \frac{e^{-\beta H}}{\Tr(e^{-\beta H})}$ within an arbitrarily small additive error. In the special case when the observable is also diagonal in the $X$-basis, i.e. $P \in \{I, X\}^{\otimes n}$, the algorithm further achieves arbitrarily small relative error.
\end{abstract}
\section{Introduction}

Developing efficient algorithms to compute the partition function and observables of quantum Gibbs states is a central problem in quantum statistical mechanics, quantum information, and complexity theory. In this work, we study the problem of approximating the partition function and Gibbs observables for the transverse-field Ising model (TFIM). The TFIM has been extensively studied in quantum simulation and quantum annealing~\cite{Bravyi2014MonteCS,king2018observation,tindall2026dynamics,dutta2015quantum,sachdev1999quantum,Boixo2013EvidenceFQ,ronnow2014defining,shin2014quantum}. It is the simplest and most studied stoquastic Hamiltonian~\cite{Bravyi2006TheCO} and characterizes \text{StoqMA}~\cite{bravyi2017complexity}, one of the four complexity classes in Cubitt and Montanaro~\cite{cubitt2016complexity}'s classification of $2$-local quantum Hamiltonians.

The TFIM Hamiltonian is defined by
\begin{equation}\label{eq:TFIM}
    H = \frac{1}{2}\sum_{i, j\in [n]} J_{ij} Z_i Z_j + \sum_{i=1}^n h_i^z Z_i + \eta\sum_{i=1}^n X_i,\footnote{Here for $P\in \set{X,Y,Z,I}$ and $i\in [n],$
$P_i := I^{\otimes (i-1)} \otimes P \otimes I^{\otimes (n - i)}.$ $X, Y, Z, I$ are the $2\times 2$ Pauli matrices:
\begin{align*}
X &:= 
\begin{bmatrix}
0 & 1 \\
1 & 0
\end{bmatrix}, \quad
Y := 
\begin{bmatrix}
0 & -i \\
i & 0
\end{bmatrix}, \quad
Z := 
\begin{bmatrix}
1 & 0 \\
0 & -1
\end{bmatrix}, \quad I:=\begin{bmatrix}
1 & 0 \\
0 & 1
\end{bmatrix}
\end{align*}.}
\end{equation}
where $J_{ij}\in \R$ is the interaction term, $h_i^z\in \R$ is the longitudinal field, $\eta\in \R$ is the transverse field. 


Without loss of generality, we may assume that the \emph{interaction matrix} $J=(J_{ij})_{ij} \in\R^{n\times n}$ is symmetric, as replacing $J$ with $ \frac{J + J^\intercal}{2}$ leaves the Hamiltonian unchanged. Since $J$ is symmetric, its eigenvalues are real.  We write $\lambda_{\max}(J)$ and $\lambda_{\min}(J)$ for its largest and smallest eigenvalues, respectively, and $\Delta(J):=\lambda_{\max}(J)-\lambda_{\min}(J)$ for its spectral width.


We consider the problem of computing the partition function $Z(\beta)=\Tr(e^{-\beta H})$ at inverse temperature $\beta\geq 0$. Prior works either require the system to be ferromagnetic \cite{Bravyi2014MonteCS} (i.e., $-J_{ij} \geq 0$ and $h_i^z=0$) or at high temperature \cite{Crosson2020ClassicalSO,SV26}. Specifically, Crosson and Slezak \cite{Crosson2020ClassicalSO} obtained an efficient classical approximation algorithm for $Z(\beta)$ when $\beta$ satisfies the Dobrushin-type condition
\[
\beta < \frac{1}{2(\Delta+2)B},\quad 
\text{where} \quad
B = \max_{i \ne j} |J_{ij}| \quad \text{and} \quad
\Delta = \max_i |\{j \ne i : J_{ij} \ne 0\}|.
\] 
Very recently, \cite{SV26} extended this result to a broader regime defined by the spectral condition $\beta \Delta(J) \leq 1$.\footnote{In particular, \cite{SV26} applies when $ \beta \leq \frac{1}{2 \Delta B},$ a condition that strictly improves upon \cite{Crosson2020ClassicalSO}'s condition that $\beta < \frac{1}{2(\Delta+2)B}$. See \cite{SV26} for further discussion.} They also proved that this condition is tight when the transverse field strength $|\eta|$ is sufficiently small relative to $\Delta(J)$. However, characterizing the tractability region under strong transverse fields remained an open problem (see \cite[Remark 3]{SV26}). We remark that understanding the TFIM in the strong transverse field regime and the qualitative changes induced by varying transverse field are of central importance to quantum annealing and condensed matter physics \cite{harris2018phase,bernaschi2024quantum,narasimhan2024simulating}.

Our first result gives a computational phase diagram for the TFIM partition function for all transverse field strengths. In \cref{thm:easy regime informal}, Part 1,  we give a randomized classical algorithm that approximates $Z(\beta)$ to a relative error $\epsilon\in(0,1)$ in time polynomial in $n$, $\beta$, the model parameters, and $\epsilon^{-1}$ when $\Delta(J)\cdot\frac{\tanh(\beta\eta)}{\eta}\leq1$,
\footnote{Throughout, we adopt the convention $ \frac{\tanh x}{x} = 1$ when $x =0.$ In particular, $\frac{\tanh(\beta \eta)}{\eta} =\beta$ when $ \eta =0.$} a weaker condition than $\beta \Delta(J)\leq 1$ from \cite{SV26}. In \cref{thm:hard regime informal}, we show that when 
$\Delta(J) \cdot \frac{\tanh(\beta \eta)}{\eta} > 1$, approximating the partition function within an $\exp(o(n))$-multiplicative factor is \NP-hard (\cref{thm:hard regime informal}) and thus is unlikely to admit an efficient classical or quantum algorithm under the standard complexity theoretic assumptions that $ \NP \not\subseteq \BQP$ and $ \NP \neq \RP.$
When the transverse field dominates the interaction strength, specifically $\Delta(J)\leq |\eta|$, our algorithm applies at \emph{any} inverse temperature $\beta \geq 0$, and consequently enables approximating the \emph{ground state energy} $E_0 = \lambda_{\min}(H)$ to arbitrary additive precision (\cref{remark:ground state energy approximation}).

As a further application of our techniques, we provide classical algorithms for estimating observables of the Gibbs state $\rho_\beta = \frac{e^{-\beta H}}{Z(\beta)}$ in the tractable regime $\Delta(J) \frac{\tanh(\beta \eta)}{\eta} \leq 1$. Specifically, we obtain an efficient classical algorithm to approximate $ \Tr(P\rho_\beta)$ for any Pauli string $P \in \{I, X, Y, Z\}^{\otimes n}$ to arbitrary additive error (see \cref{thm:easy regime informal}, Part 2). For Pauli strings that are diagonal in the $X$-basis, our algorithm achieves the stronger guarantee of arbitrarily small relative error. Because Pauli strings form a basis for observables, our technique extends to any observable with an explicitly given polynomial-size Pauli expansion $O=\sum_{P\in\{I,X,Y,Z\}^{\otimes n}}c_PP$, yielding an additive estimator whose runtime has an additional $\poly(\sum_P|c_P|)$ dependency.
Finally, combining our techniques with the state preparation framework of Wong~\cite[Section 3.2]{wong2026lee}, we also obtain an efficient \emph{quantum} algorithm to prepare the Gibbs state (see \cref{thm:Gibbs state preparation}).

\begin{theorem}[Informal, see formal version at \cref{thm:main theorem tractable}]\label{thm:easy regime informal}
    Fix $\beta\geq 0$ such that $0\leq \Delta (J) \cdot \frac{\tanh(\beta \eta)}{\eta} \leq 1.$ 
    
    Fix $ \epsilon,\delta \in (0,1).$ There exists
    an efficient randomized classical algorithm that runs in time \[\poly(n, \beta, \max_{i,j} |J_{ij}| , \max|h_i^z|,  |\eta|, \epsilon^{-1}) \cdot \log \delta^{-1}\]
    and:
    \begin{enumerate}
\item Computes $ \widehat{Z}$ s.t. 
$\Pr[(1-\epsilon) Z(\beta) \leq \widehat{Z} \leq (1+\epsilon) Z(\beta)]\geq 1-\delta.$
\item Given a Pauli string $O=\prod_{i=1}^n P_i,$ computes $ \widehat{\Tr(O\rho_\beta)}$ s.t.
$\Pr\left[\left|\widehat{\Tr(O \rho_\beta)} -\Tr(O \rho_\beta)\right|\leq \epsilon \right] \geq 1-\delta.$

In the special case when $ O$ is also diagonal in the $X$-basis, i.e., $P_i\in \set{X_i, I_i}\forall i\in [n],$ we have the stronger guarantee $ \Pr\left[\left|\widehat{\Tr(O \rho_\beta)} -\Tr(O \rho_\beta)\right|\leq \epsilon|\Tr(O \rho_\beta)| \right] \geq 1-\delta.$
\end{enumerate}
    
\end{theorem}

Our algorithm extends more generally to systems with site-dependent transverse fields (see \cref{thm:main theorem tractable}).

\begin{theorem}[Informal, see formal version at \cref{thm:main hardness of approximation}]
\label{thm:hard regime informal}
Fix $ \eta ,\beta > 0.$ For any $\varepsilon >0,$ there exists a family of symmetric matrices $J\in \R^{n\times n}$ with $ \Delta(J) \cdot \frac{\tanh (\beta \eta) }{\eta} \in (1, 1+\varepsilon),$ such that approximating $\Tr\left(e^{-\beta\left( \frac{1}{2}\sum_{ij} J_{ij} Z_i Z_j + \eta \sum_i X_i \right)}\right)$
    within an $\exp(c n)$-multiplicative factor is $\NP$-hard, and thus unlikely to admit an efficient classical or quantum algorithm assuming $ \NP \not\subseteq \BQP$ and $ \NP \neq \RP$.  
    Here $c$ is a constant dependent only on $\varepsilon, \eta,\beta.$ Moreover, the matrices in this family can be chosen to have zero diagonal and rational entries of $O_{\eta,\beta,\varepsilon}(1)$ bit length (see \cref{remark:hardness finite precision and zero diagonal}).
\end{theorem}


\begin{remark}[Approximating the ground state energy when $ \Delta(J)\leq |\eta|$]\label{remark:ground state energy approximation}
For $\beta > 0$, let $F(\beta) := -\frac{1}{\beta}\log Z(\beta)$ be the free energy. $F(\beta)$ approximates the ground state energy $E_0 = \lambda_{\min}(H)$:
$F(\beta) \le E_0 \le F(\beta) + \frac{n\log 2}{\beta}.$ Set $\beta = \frac{3 n\log 2}{\epsilon},$ and use \cref{thm:easy regime informal} to compute $ \hat{Z}$ in $\poly(n, \max_{i,j} |J_{ij}| , \max|h_i^z|,  |\eta|, \epsilon^{-1}) \cdot \log \delta^{-1}$ time, where
\[\Pr [\frac{1}{2} Z(\beta)\leq  \hat{Z}\leq 2 Z(\beta)] \geq 1-\delta \] 
and output  $\hat{E}_0= -\frac{1}{\beta}\log \hat{Z}.$ Then $ \Pr[|\hat{E}_0 - E_0| \leq \epsilon]\geq 1-\delta.$ 
\end{remark}
\subsection{Other related works}

Most algorithms for computing partition functions and Gibbs observables of quantum spin systems apply only in the high-temperature regime \cite{harrow2020classical, HKT22, BLMT24, Bakshi2026RapidMF, Putterman2026WhenQT, mann2026fastalgorithmsstoquasticspin, yin2023polynomial}. When specialized to the TFIM, these general results are weaker than those of Crosson and Slezak~\cite{Crosson2020ClassicalSO} and Sinclair and Vuong~\cite{SV26}.
Helmuth and Mann~\cite{Helmuth_2023} consider low-temperature quantum systems that are sufficiently small perturbations of classical spin systems; for the TFIM, their result requires a sufficiently weak transverse field. Mann and Waite~\cite{mann2026efficient} consider a weak-interaction condition that, for the TFIM, implies the Dobrushin condition of \cite{Crosson2020ClassicalSO}.

\subsection{Technical overview}
For ease of exposition, we absorb $-\beta$ into the Hamiltonian
coefficients and assume $\eta >0$. 
The goal is to approximate $\Tr(e^H)$ and $ \Tr(O \rho)$ with $ \rho =\frac{e^H}{\Tr(e^H)}$ in the tractable region $ \Delta(J)\cdot \frac{\tanh \eta}{\eta} \leq 1$ and show hardness of approximating $\Tr(e^H)$ beyond this region.

Fix an accuracy parameter $\epsilon > 0,$ and a Trotter number
$M= \poly(n, \max_{i,j} |J_{ij}| , \max|h_i^z|,  \eta, \epsilon^{-1}) ,$ 
and consider the Lie-Trotter approximation $Q_M$ of $ e^H$, i.e.,
\[ Q_M =  (e^{H_z/M} e^{H_x/M})^M \text{ with } H_z=\frac{1}{2}\sum_{i,j} J_{ij} Z_i Z_j+\sum_i h_i^z Z_i \text{ and } H_x=\eta\sum_i X_i.\]
For a Pauli string observable $O = \prod_i P_i,$ we show that $\Tr(Oe^H)$ is well approximated by $ \Tr(O Q_M)$, and write:
\begin{equation}\label{eq:partition function informal}
  \Tr(O Q_M) \propto \E_{\boldsymbol{\sigma} \sim \pi} \left[ e^{W(\boldsymbol{\sigma})} f (\boldsymbol{\sigma}) \right] , \text{ } W (\boldsymbol{\sigma}) = \frac{1}{2 M} \sum_{i,j} J_{ij} \langle \sigma_i, \sigma_j \rangle,\text{ } 
f (\boldsymbol{\sigma})\in  [-1,1] \forall \boldsymbol{\sigma}, \text{ } \pi=\bigotimes_{i=1}^n \pi^{(i)} 
\end{equation}
Here, each $\pi^{(i)}$ is an Ising model on a cycle graph on $M$ vertices (a.k.a. a 1D Ising model with closed boundary conditions), with couplings $\pm K_M(\eta)$ where $K_M (\eta) =\frac{1}{2}\log \coth(\eta/M)$. The couplings take both positive and negative values when $P_i\in \set{X_i, Y_i}.$ 
For $\lambda \in[0,1]$, let $ \mu_\lambda$ be the tilt of $ \pi$ by $\lambda W,$ i.e., $\mu_\lambda(\boldsymbol{\sigma})  = \frac{e^{\lambda W(\boldsymbol{\sigma}) } \pi(\boldsymbol{\sigma})}{\E_\pi[e^{\lambda W }]}.$

Using the telescoping identity, we obtain:
\begin{align*}
   \E_{ \pi} [e^W f] 
   &=\E_{\pi} [e^{W/L}] \E_{\mu_{1/L}}[e^{W/L}] \cdots \E_{\mu_{(L-2)/L}}[e^{W/L}]  \cdot \E_{\mu_{(L-1)/L}}[e^{W/L} f] 
\end{align*}
Choosing $L\geq n^2 \max_{i,j} |J_{ij}|$ ensures that $\max \{ |e^{W/L} f|  , |e^{W/L}|\} = O(1).$ We can then estimate $ \E_{ \pi} [e^W f]$ using samples from $ \mu_\lambda$ with $\lambda\in \set{ \frac{i-1}{L}\mid i\in [L]}.$

To sample from $\mu_{\lambda}$, we use the Glauber dynamics: each step of this Markov chain chooses \(i\in[n]\) uniformly and resamples $\sigma_i\in \set{\pm 1}^M$ conditioned on $\boldsymbol{\sigma}_{-i} = (\sigma_j)_{j\neq i};$ resampling $\sigma_i$ is equivalent to sampling from a 1D Ising model, and thus can be efficiently implemented. 

Approximating the TFIM partition functions and Gibbs observables in the tractable region ($ \Delta(J)\cdot \frac{\tanh \eta}{\eta} \leq 1$) thus reduces to proving rapid mixing of the Glauber dynamics on $\mu_{\lambda}.$ By the trickle down method of \cite{AKV24}, this
reduces to showing that the covariance matrices of all linear tilts of $\pi^{(i)}$ have operator norm at most
\[
\frac{\tanh\eta}{\tanh(\eta/M)}
\approx M\cdot\frac{\tanh\eta}{\eta}.
\]
In \cref{subsec:covariance bound for 1D}, we prove the more general result that this covariance bound holds for any Ising model on a cycle with coupling terms in $[-K_M(\eta), K_M(\eta)]$ and arbitrary external fields. The key idea is to bound each entry of the covariance matrix by relating the Ising model on the cycle to Ising models on paths.

We sketch the proof of hardness of approximation. 
We can relate the quantum partition function $\Tr(e^H)$ to the partition function $Z^{\text{classical}}_{J}$ of a classical distribution $ \tilde{\mu}^{\text{classical}}_{ J}$.
The hardness of approximation proof relies on the fact that the intractable regime is precisely where the Curie-Weiss TFIM is bimodal:  
 the average $\bar{\sigma} = \frac{1}{n}\sum_i \sigma_i ,$  where $(\sigma_1, \cdots, \sigma_n)$ is sampled from the corresponding classical distribution,  is either close to $ m_* \mathbf{1}$ or $-m_* \mathbf{1}$ with overwhelming probability, for some $m_*\in (0,1)$. We prove this key property in \cref{subsec:quantum curies-weiss} by establishing a large deviation principle for $ \bar{\sigma}$.
 Inspired by gadget construction techniques from \cite{sly2010computational,SS12,galanis2024sampling}, we use the bimodality of the Curie-Weiss TFIM to construct a hard TFIM instance from a hard \textsc{MAXCUT} instance $G.$ We build a base gadget graph $\hat{\mathcal{G}}$ by replacing each vertex $v$ in $G$ with a Curie-Weiss gadget $K_v$, and another gadget graph $\mathcal{G}$  by adding edges between these Curie-Weiss gadgets to represent the edges of $G.$ We let $ \mathcal{V}$ denote the common vertex set of $\hat{\mathcal{G}}$ and $\mathcal{G},$ and $\hat{J}$ and $\tilde{J}$ be the adjacency matrices  of $\hat{\mathcal{G}} $ and $\mathcal{G}$ respectively. 
 For each Curie-Weiss gadget $K_v,$ we use a subset of vertices $R_v\subseteq V(K_v)$ to encode the spin value of $v$:  we set $Y(\sigma_{R_v})=+1$ if $\bar{\sigma}_{R_v}=|R_v|^{-1}\sum_{i\in R_v}\sigma_i$ is strictly closer to $m_* \mathbf{1}$ than to $-m_* \mathbf{1}$, and set it to $-1$ otherwise.
 Each configuration $\boldsymbol{\sigma} = (\sigma_i)_{i\in \mathcal{V}}$ thus maps to a configuration $\mathcal{Y}(\boldsymbol{\sigma} ) = (Y(\sigma_{R_v}))_{v\in V_G} \in \set{\pm 1}^{V_G}.$ 
For $ J \in \set{\hat{J}, \tilde{J}}$ and $ \mathbf{y}\in \set{\pm 1}^{V_G},$ let $Z_{J}(\mathbf{y})$ be the contribution of configurations $\boldsymbol{\sigma} = (\sigma_i)_{i\in \mathcal{V}}$ with $\mathcal{Y}(\boldsymbol{\sigma} ) = \mathbf{y}$ to $Z^{\text{classical}}_{ J} ;$ using the bimodality of the Curie-Weiss gadget, we can 
relate $ \log \frac{ Z_{\tilde{J}}(\mathbf{y})}{ Z_{\hat{J} }(\mathbf{y}) }$ to the size of the cut defined by $\mathbf{y}.$
By adding sufficiently many antiferromagnetic edges between gadgets, we can make the contribution from the maximum cut(s) of $G$ dominate, so that $ \log \frac{Z^{\text{classical}}_{ \tilde{J}}  }{Z^{\text{classical}}_{ \hat{J}}  }$ encodes $\textsc{MAXCUT}(G)$, the maximum cut value of $G.$ 
By a suitable choice of the gadget construction's parameters, we can ensure that both matrices $\tilde{J}$ and $\hat{J}$ lie just beyond the tractable region, i.e., $ \Delta(J) \cdot \frac{\tanh ( \eta) }{\eta} \in (1, 1+\varepsilon)$  for $ J \in \set{\tilde{J}, \hat{J}}$, while sufficiently accurate approximations to both $ Z^{\text{classical}}_{ \tilde{J}}$ and $ Z^{\text{classical}}_{ \hat{J}} $ would give an approximation of $\textsc{MAXCUT}(G)$ to a sufficiently small relative error.
The \textsc{APX}-hardness of \textsc{MAXCUT} \cite{alimonti1997hardness} then implies the hardness of approximation result in \cref{thm:hard regime informal}. 
\subsection{AI methodology}
The author developed the high-level strategies for the approximation algorithms and the hardness of approximation proof. 
AI tools (GPT and Gemini) were used to develop the proofs of \cref{thm:covariance bound for each block}, \cref{thm:concentration on two wells}, \cref{prop:multiplicative approx of trotterized X observable}, and \cref{prop:addition apporoximation of Pauli observables}. The author also noted that elements of the AI-generated proof for \cref{thm:concentration on two wells} previously appeared in \cite{CCIV08}. AI tools assisted with proofreading and typesetting. The author retains full responsibility for the paper's content.

\subsection{Acknowledgement}
We thank Benjamin Wong 
for insightful discussions.
\subsection{Organization}
In \cref{sec:expression for observables}, we relate the partition function and Gibbs observables of the TFIM to the partition functions and expectations of classical distributions. 
We prove \cref{thm:easy regime informal} in \cref{sec:tractable}, and \cref{thm:hard regime informal} in \cref{sec:lower bound}. In \cref{sec:Gibbs state preparation}, we analyze the Gibbs state preparation algorithm.

\section{Preliminaries}
For a measure $\nu$ and function $V$ on an inner product space $\mathbb{V}$, the \textbf{tilt} of $\nu$ by $V$ is the measure $\mathcal{T}_V \nu$ on $\mathbb{V}$ defined by
\[
\mathcal{T}_V \nu (d\sigma) = \frac{e^{V(\sigma)} \nu(d\sigma)}{\mathbb{E}_{\sigma \sim \nu} [e^{V(\sigma)}]}
\]
For $w \in \mathbb{V}$,
we let $\lineartilt_w \nu \equiv \tilt_{\sigma\mapsto \langle w,\sigma \rangle}\nu $ denote the \textbf{linear tilt} of $\nu$ by $w.$

For a matrix $A\in \R^{n\times n},$ let $\norm{A}_{\infty} = \max_{i} \sum_{j} |A_{ij}|$ denote its row norm, $\norm{A}_{\operatorname{op}}$ denote its operator norm, and $ \norm{A}_{1}$ denote its Schatten-$1$ norm. 

For symmetric $A$, let $\lambda_{\max}(A)$ and $\lambda_{\min}(A)$ denote its (real) maximum and minimum eigenvalues, and $\Delta(A) = \lambda_{\max}(A) - \lambda_{\min}(A)$ denote its \emph{spectral width}.

For a vector $ v\in \R^n,$ let $ \norm{v}_\infty =\max_i |v_i|,$ $\norm{v}_2 = \sqrt{\sum_{i \in [n]} v_i^2 } $ and $ \norm{v}_1= \sum_{i \in [n]} |v_i|.$ Unless otherwise noted, we use $ \norm{v}$ to denote the $L_2$-norm $\norm{v}_2.$

Let $ \mS= \R/\Z$ be the circle of unit length; we identify $\mathbb{S}$ with \([0,1)\) with $0$ and $1$ identified.
Let $L^2(\mS)$ denote
the collection of equivalence classes of square-integrable measurable functions $f:\mS \to \R$\footnote{where $\mS$ is equipped with the Lebesgue measure} where two functions are equivalent if they agree almost everywhere. 
We view $L^2(\mS)$ as a vector space, so that for any function $f\in L^2(\mS)$ and scalar $ c\in \R,$ $c f$ is the function that maps $t$ to $ c f(t)$ for $t\in \mS.$ We let $\textbf{1}: \mS\to \R$ be the function that maps $t$  to $1$ for all $t\in \mS. $
This vector space is equipped with the inner product $ \langle f,g \rangle = \int_{\mS} f(t)g(t) dt$ and $L_2$-norm $\norm{f}^2_2 =\langle f , f\rangle.$
We also define the $L_1$ norm, $ \norm{f}_1 = \int_{\mS} |f(t)| dt$ and note that $ \norm{f}_1 \leq \norm{f}_2$ by Holder's inequality. Unless otherwise noted, we use $ \norm{f}$ to denote the $L_2$-norm $\norm{f}_2.$

Throughout, \(A=\exp(\pm r)B\) means that \(A=\theta B\) for some \(\theta\in[e^{-r},e^r]\).

\textbf{Computational model.} 
Unless otherwise noted, we analyze our algorithms in the classical Real RAM model of computation \cite{borodin1975computational,preparata2012computational,bringmann2017efficient}, where
the following operations take constant time: (1)
accessing the content of any memory cell, (2) generating a uniformly distributed
real number in $[0, 1]$, and (3) performing basic arithmetical operations
involving real numbers like addition, multiplication, division, comparison, truncation, exponential, and logarithm. 
In \cref{sec:finite precision}, we discuss how to implement our algorithms with finite-precision arithmetic and analyze their costs in the Word RAM model.

\section{Relating Quantum Observables to Classical Expectations}\label{sec:expression for observables}


In this section, we rewrite observables of the TFIM as expectations over classical distributions. To this end, we introduce the following definitions.
\begin{definition}\label{def:general path definition}

For $\eta\geq 0$ and $p\in\{0,1\}$, define
\begin{equation}\label{eq:normalization factor for each index}
Z_{\eta, (p)}= \frac{1}{2}\Tr(X^p e^{\eta X}) = \sum_{\substack{d\geq 0\\d\equiv p \mod 2}}  \frac{\eta^d }{d!} =  \begin{cases} \cosh\eta, &\text{ if } p = 0\\ \sinh \eta, &\text{ if } p = 1\end{cases}.   
\end{equation}

Construct the probability measure $\nu_{\eta,(p)}$ over functions $ \sigma:\mS \to \set{\pm 1}$ as follows:
\begin{enumerate}
    \item Sample $D\sim \mathrm{Pois}(\eta, p),$  where $\mathrm{Pois}(\eta, p)$ is the Poisson distribution $\mathrm{Pois}(\eta)$ conditioned on $D\equiv p \mod 2$. 
    
    For $\eta > 0$, $\mathrm{Pois}(\eta, p) $ is defined by
    \[
    \Pr_{D\sim \mathrm{Pois}(\eta, p) } [D=d]=\frac{\eta^d}{d!\,Z_{\eta, (p)}},
    \qquad d\geq0,\quad d\equiv p\mod2.
    \]
    For $ \eta =0,$ taking the limit as $\eta\to 0$, we have $\Pr_{D\sim \mathrm{Pois}(0, p) } [D=d] = \mathbbm{1}[d = p];$ thus   $\mathrm{Pois}(0, p) $ places all probability mass on $p.$
    \item Conditional on $D=d$, sample $d$ independent uniform
    points $t^{(1)},\ldots,t^{(d)}$ in $(0,1)$ and, independently,
    a uniform sign $u\in\{-1,+1\}$. Define $ \sigma: \mS \to \set{\pm 1}$\footnote{Recall that $ \mS= \R/\Z\equiv [0,1)$ with $0$ and $1$ identified.} by: for $t\in [0,1),$ set
    $
    \sigma(t)=u(-1)^{\#\{j\in [d]:t^{(j)}\leq t\}}
    .
    $  Let
    $
    \mathcal J(\sigma)
    :
    =\{t^{(1)},\ldots,t^{(d)}\}.
    $

    Let $\nu_{\eta,(p)}$ be the law of $\sigma$.

    

   
\end{enumerate}

\end{definition}

\begin{definition}\label{def:1D Ising model CBC}
For interaction $\mathbf{K} = (K^{(k)})_{k=0}^{M-1}\in \mathbb{R}^M$ and external field $\mathbf{w} = (w^{(k)})_{k=0}^{M-1} \in \mathbb{R}^M$, the 1D Ising model with closed boundary conditions (CBC) $\pi_{\mathbf{K}, \mathbf{w}}^{\mathrm{1D Ising, CBC}} $ is the probability distribution  over $\set{\pm 1}^M$ defined by: for any configuration $\mathbf{s} = (s^{(0)}, \dots, s^{(M-1)}) \in \set{\pm 1}^M$,
\begin{equation}\label{def:1D Ising with closed boundary}
\pi_{\mathbf{K}, \mathbf{w}}^{\mathrm{1D Ising, CBC}}(\mathbf{s}) 
= \frac{1}{Z_{\mathbf{K}, \mathbf{w}}^{\mathrm{1D Ising, CBC}}}
\exp\left( \sum_{k=0}^{M-1} K^{(k)} s^{(k)} s^{(k+1)} + \sum_{k=0}^{M-1} w^{(k)} s^{(k)} \right), \quad \text{where } s^{(M)} = s^{(0)}.
\end{equation}
This definition extends to allow infinite fields, i.e., $ w^{(k)} \in \set{\pm \infty}$ by taking limits. Infinite fields are interpreted as
conditioning on $s^{(k)} =\operatorname{sign}(w^{(k)})$ whenever
$w^{(k)}\in \set{\pm \infty}.$
\end{definition}
Fix $M\in \N_{\geq 3}.$ We let $\hat{\nu}_{\eta, (p), M} $ be the push-forward distribution of $ \nu_{\eta, (p)}$ by the map \[\sigma \mapsto \mathrm{Val}_M(\sigma) = (\sigma(0), \sigma(\frac{1}{M}), \cdots, \sigma(\frac{M-1}{M}) ) \in \set{\pm 1}^M.\]
We will show that $\hat{\nu}_{\eta, (p),M}$ is a 1D Ising model with closed boundary conditions.
\begin{proposition}\label{prop:approximation of path functional}
     $\hat{\nu}_{\eta, (p), M} $, the push-forward distribution  of $ \nu_{\eta, (p)}$ by the map $\sigma \mapsto \mathrm{Val}_M(\sigma)$, is exactly $\pi_{\mathbf{K}_{\eta, (p)}, \mathbf{0}}^{\mathrm{1D Ising, CBC}}$, where $\mathbf{K}_{\eta, (p)}  = (K^{(k)})_{k=0}^{M-1}$ s.t. $K\equiv K_M (\eta) =\frac{1}{2}\log \coth(\eta/M),$ 
     \[K^{(0)} = \cdots = K^{(M-2)} = K_M (\eta), \quad K^{(M-1)} = K_M (\eta) (-1)^p. \]
     At $\eta=0$, the Ising representation is understood as the limit when
$\eta\downarrow0$.

\end{proposition}

In what follows, fix a symmetric matrix
$J=(J_{ij})_{i,j=1}^n\in\mathbb R^{n\times n}$, a longitudinal field
$\boldsymbol h^z=(h_i^z)_{i=1}^n\in\mathbb R^n$ and $\boldsymbol\eta=(\eta_i)_i\in \R_{\geq 0}^n.$ 
Define the Ising potential function $v \equiv v_{J,\boldsymbol h^z}: \set{\pm 1}^n\to \R$ by:
\begin{equation}\label{eq:function Ising}
    v(s)
:=
\frac12\sum_{i,j=1}^nJ_{ij}s_is_j
+\sum_{i=1}^nh_i^zs_i
\end{equation}

Let
\begin{equation}\label{eq:Hz and Hx}
   H_z =\frac{1}{2}\sum_{i, j} J_{ij} Z_i Z_j +\sum_{i=1}^n h_i^z Z_i, \qquad H_x = \sum_{i} \eta_i X_i, \qquad  H = H_z + H_x. 
\end{equation}

Consider an observable $O$  that is a product of Pauli matrices, i.e., $O=\prod_{i=1}^n P_i.$

Let $S = \set{i: P_i\in \set{X_i, Y_i}},$ $X_S=\prod_{i\in S}X_i$.
For $i\in [n]$, let $p_i = \mathbbm{1}_{i\in S}$, $ \nu^{(i)} = \nu_{\eta_i,(p_i)}$ as in \cref{def:general path definition}. 
Fix an integer $M\geq 3.$ Let $\hat{\nu}^{(i)} = \hat{\nu}_{\eta_i,(p_i), M} = \pi_{\mathbf{K}_{\eta_i, (p_i)}, \mathbf{0}}^{\mathrm{1D Ising, CBC}}$ as in \cref{prop:approximation of path functional} and \[\tilde{\nu}^{(i)} = \lineartilt_{(h_i^z/M) \cdot \mathbf{1}}\hat{\nu}^{(i)} =\pi_{\mathbf{K}_{\eta_i, (p_i)}, (h_i^z/M) \cdot \mathbf{1}}^{\mathrm{1D Ising, CBC}} . \]

Let $\nu = \bigotimes_{i=1}^n\nu^{(i)}$, $\hat{\nu} = \bigotimes_{i=1}^n\hat{\nu}^{(i)}$  and $\tilde{\nu} =\bigotimes_{i=1}^n\tilde{\nu}^{(i)}. $

For $\boldsymbol{\sigma} =  (\sigma_i)_{i\in [n]}\sim \nu,$ let $ \boldsymbol{\hat{\sigma}}=\mathrm{Val}_M(\boldsymbol{\sigma}) = (\hat{\sigma}_i)_{i=1}^n$ where $\hat{\sigma}_i = \mathrm{Val}_M(\sigma_i).$ Then $ \boldsymbol{\hat{\sigma}}\sim \hat{\nu} .$

Write $\boldsymbol\sigma (t) =(\sigma_1(t), \cdots, \sigma_n(t))$ and let
\begin{equation}\label{eq: VM and W}
    \begin{split}
    V^{(M)}(\boldsymbol{\sigma})  \equiv V^{(M)} (\boldsymbol{\hat{\sigma}})
:&=  \frac{1}{M }\sum_{k=0}^{M-1} v(\boldsymbol\sigma(k/M)) = W(\boldsymbol{\hat{\sigma}}) +\sum_i \frac{h_i^z}{M}  \langle\mathbf{1}, \hat{\sigma}_i\rangle,\qquad
        W(\boldsymbol{\hat{\sigma}}) = \frac{1}{2M}\sum_{i,j=1}^n J_{ij}\langle\hat{\sigma}_i, \hat{\sigma}_j\rangle\\
        V(\boldsymbol{\sigma})
:&= \int_{\mS} v(\boldsymbol \sigma(t))dt, \qquad f (\boldsymbol \sigma)\equiv f(\boldsymbol{\hat{\sigma}})=\prod_{i: P_i \in \set{Y_i,Z_i} } \sigma_i(0)
    \end{split}
\end{equation}

The following proposition follows from standard applications of the Lie--Trotter product formula (cf.\ \cite{AB20,LRRS21,Bravyi2014MonteCS,Crosson2020ClassicalSO,SV26}). In what follows, let $Q_M = (e^{H_z/M} e^{H_x/M} )^M .$
\begin{proposition}\label{prop:derive expression for observables}
We have, for $ C=\Tr(X_S e^{H_x} ) \geq 0$
 \begin{equation}\label{eq:derivation of observables continuous}
 \begin{split}
\Tr(O \exp (H) ) &=  \mathrm{i}^{|\set{i: P_i= Y_i}|} C\, \E_{\boldsymbol\sigma \sim \nu} [e^{V(\boldsymbol \sigma)} f (\boldsymbol \sigma) ] 
\end{split}
\end{equation}
Let $D = \Tr (X_S ( e^{\frac{1}{M}\sum_i h_i^z Z_i } e^{H_x/M})^M )\geq 0.$ We have:
\begin{equation}\label{eq:derivation of observables discretized}
\Tr(O Q_M) = \mathrm{i}^{|\set{i: P_i = Y_i}|} C\, \E_{\boldsymbol{\sigma} \sim \nu} \left[ e^{V^{(M)}(\boldsymbol{\sigma})} f(\boldsymbol{\sigma}) \right] = \mathrm{i}^{|\set{i: P_i = Y_i}|}  D \,\E_{\boldsymbol{\hat{\sigma}} \sim \tilde{\nu}} \left[ e^{W(\boldsymbol{\hat{\sigma}})} f(\boldsymbol{\hat{\sigma}}) \right] 
\end{equation}
\end{proposition}
The following propositions show that, for a suitable choice of $M$, $\Tr(O Q_M)$ is a good approximation of $\Tr(O e^H).$ 
\begin{proposition}\label{prop:multiplicative approx of trotterized X observable}
Let $c= 2 \max_i (\sum_{j\neq i} |J_{ij}| + |h_i^z|).$
Fix $\epsilon\in (0,1).$ Suppose
$M\ge
 \frac{c((c+1) n +2 \sum_i |\eta_i|)}{\epsilon} 
.$ Then  
\[\Tr(X_S e^{H} )e^{-\epsilon}\leq  \Tr(X_S Q_M) \leq\Tr(X_S e^{H} )e^\epsilon .\]

\end{proposition}
\begin{proposition}\label{prop:addition apporoximation of Pauli observables}

Fix $\epsilon \in (0,1).$  
Let $c = \norm{J}_{\infty} + \norm{\mathbf{h}^z}_{\infty} + \norm{\mathbf{\eta}}_{\infty}.$ 
Suppose $M\geq C \left(\max\left\{ \frac{nc }{\epsilon}, \sqrt{\frac{(nc)^3}{\epsilon}}\right\}\right)$ where $C = 15.$
Let $\rho =e^H/\Tr(e^H).$
We have:
\[\left| \frac{\Tr(O Q_M)}{\Tr(Q_M)} - \Tr(O \rho)\right| \leq\epsilon.  \]
\end{proposition}
We defer the proofs to \cref{subsec:defer from observables}.
\section{Tractable regime}\label{sec:tractable}

In this section, we will prove the following theorem, which immediately implies \cref{thm:easy regime informal}:

\begin{theorem} \label{thm:main theorem tractable}
Consider a symmetric matrix $J = (J_{ij})_{i,j\in [n]}\in \R^{n\times n}, \boldsymbol{h^z}=(h_i^z)_{i=1}^n\in\mathbb R^n, \boldsymbol{\eta} =(\eta_i)_i\in \R^n.$  

Let $ \eta_{\min} =\min_{i} |\eta_i|.$ Let $\alpha(\eta) = \frac{\tanh \eta}{\eta}$\footnote{Here, we follow the convention $\frac{\tanh(x)}{x}= 1$ when $ x=0.$}. 
Suppose $\Delta(J) \cdot \alpha(\eta_{\min}) \leq 1.$

Let  $H =\frac{1}{2}\sum_{i, j} J_{ij} Z_i Z_j +\sum_{i=1}^n h_i^z Z_i + \sum_{i} \eta_i X_i$ and $\rho = e^H/\Tr(e^H).$

Fix $\epsilon,\delta\in (0,1).$ There exists a randomized classical algorithm that runs in $\poly(n, \max_{i,j} |J_{ij}| , \max|h_i^z|,  \max |\eta_i|, \epsilon^{-1}) \cdot \log \delta^{-1}$ time and:
\begin{enumerate}
\item Compute $ \widehat{\Tr( e^H)}$ s.t. 
$\Pr[\widehat{\Tr(e^H)} = e^{\pm \epsilon} \Tr(e^H)]\geq 1-\delta.$
    \item Given $ S\subseteq [n]$ and $X_S=\prod_{i\in S} X_i$, compute $ \widehat{\Tr(X_S \rho)}$ s.t. 
$\Pr[\widehat{\Tr(X_S \rho)} = e^{\pm \epsilon} \Tr(X_S \rho)]\geq 1-\delta.$
\item Given a Pauli string $O=\prod_{i=1}^n P_i,$ compute $ \widehat{\Tr(O\rho)}$ s.t.
$\Pr\left[\left|\widehat{\Tr(O \rho)} -\Tr(O \rho)\right|\leq \epsilon \right] \geq 1-\delta.$
\end{enumerate}
\end{theorem}

We need the following result:
\begin{theorem}\label{thm:main mixing theorem}
Fix an integer $M\geq 3.$ Let $\Omega = \set{\pm 1}^M .$ 

For $i \in [n],$ consider  $K_i = (K_i^{(k)})_{k=0}^{M-1}, w_i= (w_i^{(k)})_{k=0}^{M-1},$ and let $ \pi_i =\pi_{K_i,w_i}^{\mathrm{1D Ising, CBC}}. $ 

Consider a symmetric matrix $J = (J_{ij})_{i,j\in [n]}\in \R^{n\times n},$ and let $W$ be as in \eqref{eq: VM and W}, i.e.,
$ W (\mathbf{s}) = \frac{1}{2M} \sum_{i,j} J_{ij} \langle \mathbf{s} _i, \mathbf{s}_j \rangle.$ 

Let $ \pi = \bigotimes_{i=1}^n \pi_i$ and for $\lambda\in [0,1],$ let $\mu_\lambda = \tilt_{\lambda W} \pi.$ We view $\mu_\lambda$ as a distribution over $\Omega^n.$ We let  $\mu \equiv \mu_1.$  

     Suppose $\forall i, k: |K_i^{(k)}|\leq K$ for $K \in (0,+\infty].$ Let  $\eta \geq 0$ be s.t. $ K = \frac{1}{2} \ln \coth (\eta/M)$, and let $\alpha_M(\eta) = \frac{\tanh \eta}{M\tanh(\eta/M)}.$\footnote{Here, we follow the convention $\frac{\tanh(x)}{x}= 1$ when $ x=0,$ so that $ \alpha_M(\eta) = 1 $ when $\eta =0.$} 
     
     In particular, this condition holds when $ K_i \in \set{\mathbf{K}_{\eta_i, (p)} \mid p\in \set{0,1} }$ with $\eta_i \geq \eta,$ where $ \mathbf{K}_{\eta_i, (p)}$ is as defined in \cref{prop:approximation of path functional}. 

   Let $ L = n \lceil \norm{J}_\infty \rceil +1.$ 
  The following holds:
  \begin{enumerate}
      \item 
Fix $\lambda \in [0,1].$  
Suppose $ \lambda \Delta(J) \alpha_M (\eta) < 1.$   
Let $\rho^{(t)}$ be the distribution of the Glauber dynamics\footnote{Each step of the Glauber dynamics chooses \(i\in[n]\) uniformly and resamples $\mathbf{s}_i\in \set{\pm 1}^M$ conditioned on $\mathbf{s}_{-i} = (\mathbf{s}_j)_{j\neq i}.$} with respect to $\mu_\lambda$ initialized at $\rho^{(0)}\ll \mu_\lambda$ after $t$ steps. Each step of the Glauber dynamics contracts KL-divergence by $\left(1-\frac1{Cn}\right)$, i.e.,
\[
D_{\operatorname{KL}}(\rho^{(t+1)}\|\mu_\lambda)
\leq \left(1-\frac1{Cn}\right)
D_{\operatorname{KL}}(\rho^{(t)}\|\mu_\lambda),
\quad C = \frac{1}{1 - \lambda \Delta(J) \alpha_M (\eta) }.\]  
Let $\rho^{(0)}=\pi.$    
Fix $\epsilon_{TV}\in(0,1).$
If 
$t\geq 2 (C n \log (L/\epsilon_{TV})),$ 
then
$ d_{TV}(\rho^{(t)} , \mu_\lambda) \leq \epsilon_{TV} .$ Sampling from $\pi$ and
   implementing  each step of the Glauber dynamics take $O(n M)$ (classical) time. Hence, there exists a classical algorithm that samples from $ \hat{\mu}_\lambda$ with $ d_{TV}(\hat{\mu}_\lambda, \mu_\lambda)\leq \epsilon_{TV}$ in time $O\left(\frac{n \log (L/\epsilon_{TV})}{1 - \lambda \Delta(J) \alpha_M (\eta) } \cdot nM \right).$
     \item  Suppose $ (1-1/L) \Delta(J) \alpha_M (\eta) < 1.$   
    For $ \epsilon,\delta\in (0,1),$ there exists a classical algorithm that runs in time \[T = O\left(\frac{n \log (L/\epsilon)}{1 - (1-1/L) \Delta(J) \alpha_M (\eta) } \cdot nM \cdot \frac{L^2}{\epsilon^2}\cdot \log\frac{1}{\delta}\right) \] and outputs an approximation $\hat{Z} $ of  $ Z =\E_{\pi}[e^W] $  so that $\Pr[ \hat{Z} = \exp(\pm \epsilon) Z ] \geq 1-\delta.$
  \end{enumerate}

\end{theorem}

To prove 
\cref{thm:main mixing theorem}, we need the following variant of \cite[Theorem 9]{AKV24}.

\begin{theorem}[Variant of {\cite[Theorem 9]{AKV24}}]\label{thm:semilogconcave}
Suppose that $\mu = \bigotimes_{i = 1}^n \mu^{(i)}$ is a product measure where each $\mu^{(i)}$ is a probability measure on $(\mathbb S^{N-1})\subseteq \R^N$. Let $J \succeq 0$, and define the measure $\nu$ by its density
$\frac{d\nu}{d\mu}(x) \propto \exp\left(\frac{1}{2} \sum_{i,j} J_{ij} \langle x_i, x_j \rangle \right). $
Suppose that for all $i \in [n]$, the measure $ \mu^{(i)}$ is $ \rho$ semi-log concave, 
meaning that for all $w ,$ $ \norm{\Cov(\lineartilt_w \mu^{(i)})}_{\operatorname{op}}\leq \rho ,$ with $ \norm{J}_{\operatorname{op}} \rho <1.$
Then approximate tensorization of entropy (ATE) for $\nu$ holds with constant at most $ C = \frac{1}{1 - \norm{J}_{\operatorname{op}} \rho},$ and thus each step of the Glauber dynamics on $ \nu$ contracts KL-divergence by $\left(1-\frac1{Cn}\right).$ \footnote{For background on ATE and Glauber dynamics, see \cite[Sections 2.2 and 2.3]{AKV24}.}
\end{theorem}

To apply \cref{thm:semilogconcave}, we need the following covariance bound:
\begin{lemma}
\label{thm:covariance bound for each block}
Let $\tilde{\pi}_{\mathbf K,\mathbf w}$ be the push-forward distribution of $\pi_{\mathbf K,\mathbf w}^{\mathrm{1D Ising, CBC}}$ in \cref{def:1D Ising model CBC} by the scaling map $\mathbf{s} \mapsto \frac{1}{\sqrt{M}}\mathbf{s}.$
Let $K > 0$ be s.t.
$
|K^{(k)}|\le K \forall k, $
and $\eta> 0 $ be s.t. \(\tanh K = e^{-2\eta/M}\), i.e., $ K = \frac{1}{2} \ln \coth (\eta/M)$.
 Then $\norm{\Cov(\tilde{\pi}_{\mathbf K,\mathbf w})}_{\mathrm{op}} \leq\alpha_M(\eta).$

\end{lemma}
We also need the following propositions: 
\begin{proposition}\label{prop:1D Ising sampling time}
     Given $\mathbf K,\mathbf w,$ there exists a classical algorithm that samples from $\pi_{\mathbf K,\mathbf w}^{\mathrm{1D Ising, CBC}}:\set{\pm 1}^M \to \R_{\geq 0} $ in $O(M)$ time.
     In particular, this applies when $\mathbf K\in \set{ \mathbf{K}_{\eta, (p)}\mid p\in \set{0,1}}$ where $\eta \geq 0$ and $ \mathbf{K}_{\eta, (p)}$ is as in \cref{prop:approximation of path functional}, provided that the distribution is well-defined.
\end{proposition}

\begin{proposition}\label{lem:modify counting to sampling}
    Consider a normalized density function $\nu$ over $\Omega$ and a function $V: \Omega \to [-L, L]$ with $L \in \N_{\geq 1}$. For $\lambda \in [0,1]$, let $\mu_\lambda = \tilt_{\lambda V} \nu$, i.e.,
    $
        \mu_\lambda(\sigma) = \frac{\exp(\lambda V(\sigma))\nu(\sigma)}{Z(\lambda)}$, where $\quad Z(\lambda) = \E_{\nu}[\exp(\lambda V(\sigma))].
    $
    For any $\epsilon,\delta \in (0,1)$, there exists a classical algorithm that outputs a strictly positive estimator $\hat{Z}$ such that
    \[
        \Pr\left[(1-\epsilon) Z(1) \leq \hat{Z} \leq (1+\epsilon) Z(1)\right] \geq 1- \delta.
    \]
    Let $K = \Theta\left((L/\epsilon)^2\right) $ and $ K' = \Theta(K \log(1/\delta)). $ 
    The algorithm uses $K'$ independent samples from $\hat{\mu}_\lambda$ (where $d_{TV}(\hat{\mu}_\lambda, \mu_\lambda) \leq \frac{1}{8K}$ and $\lambda \in \{i/L \mid i \in [0,L-1] \cap \N\}$) and $K'$ evaluations of $V$, and uses $O(K')$ additional arithmetic operations.
\end{proposition}

We will prove \cref{thm:covariance bound for each block} in \cref{subsec:covariance bound for 1D}. 
We defer 
the proofs of \cref{prop:1D Ising sampling time,lem:modify counting to sampling,thm:semilogconcave} to \cref{subsec:defer from tractable}.

\begin{proof}[Proof of \cref{thm:main mixing theorem}]
Part 1 is a corollary of \cref{thm:semilogconcave} and \cref{thm:covariance bound for each block}.
    We replace $ J $ with $\tilde{J}= J-\lambda_{\min}(J) I\succeq 0;$ this does not change the distribution $\mu.$
    Let  $\tilde{\pi}_i$ and $\tilde{\mu}_\lambda$ be the pushforward distributions of $\pi_i$ and $\mu_\lambda$ by $ \mathbf{s} \mapsto \frac{1}{\sqrt{M}} \mathbf{s}$ respectively. 
    
    Clearly, $ \tilde{\mu}_\lambda (\zeta_1, \cdots, \zeta_n)\propto\exp(\frac{\lambda}{2} \sum_{ij} \tilde{J}_{ij} \langle \zeta_i,\zeta_j\rangle ) \prod_{i=1}^n \tilde{\pi}_i(\zeta_i).$
    
   We will show $ \tilde{\pi}_i$ is $\alpha_M(\eta) $-semi log concave for each $ i\in [n],$ i.e., show that $\norm{\Cov(\lineartilt_w \tilde{\pi}_i)}_{\mathrm{op}}\leq  \alpha_M(\eta)$ for all $w.$
    
    

This is trivial for $ \eta =0,$ since $\lineartilt_w\tilde{\pi}_i $ is supported on the ball of radius $1,$ so $\norm{\Cov(\lineartilt_w \tilde{\pi}_i)}_{\mathrm{op}}\leq 1 = \alpha_M(0).$ 
Assume $\eta > 0.$
Since  $ \lineartilt_w \tilde{\pi}_i$ is the pushforward of $ \lineartilt_{w'} \pi_i\equiv \pi_{K_i,w_i+w'}^{\mathrm{1D Ising, CBC}}$ by $ \mathbf{s} \mapsto \frac{\mathbf{s}}{\sqrt{M}} $, with $w' = w/\sqrt{M} $, \cref{thm:covariance bound for each block} gives $ \norm{\Cov(\lineartilt_w \tilde{\pi}_i)}_{\mathrm{op}}\leq \alpha_M(\eta) $ as required.  

By \cref{thm:semilogconcave}, the Glauber dynamics on
$\tilde{\mu}_\lambda$ (and thus $\mu_\lambda$) satisfies KL-divergence contraction
\[
D_{\operatorname{KL}}(\rho^{(t+1)}\|\mu_\lambda)
\leq \left(1-\frac1{Cn}\right)
D_{\operatorname{KL}}(\rho^{(t)}\|\mu_\lambda)\]
with
\[ C = \frac{1}{ 1- \norm{\lambda \tilde{J}}_{\mathrm{op} } \alpha_M(\eta) } = \frac{1}{1-\lambda \alpha_M(\eta) \Delta(J) }.\]
Hence, using   Pinsker's inequality, $ D_{\operatorname{KL}}(\rho^{(t)}\|\mu_\lambda)
\leq \left(1-\frac1{Cn}\right)^t
D_{\operatorname{KL}}(\rho^{(0)}\|\mu_\lambda),$ and the fact that $
\sup_{\mathbf s\in\operatorname{supp}\pi}
\frac{\pi(\mathbf s)}{\mu_\lambda(\mathbf s)} 
\le e^{2L}
$ since $\sup_{\mathbf{s}\in \Omega^n} |W(\mathbf{s}) |\leq L,$
the next statement follows. 
Each step of the Glauber dynamics samples from a 1D Ising model with closed boundary conditions, whose parameters can be computed in $O(nM)$ time, and thus can be done in $O(nM)$ time by \cref{prop:1D Ising sampling time}. Similarly, we can sample from $\pi$ in $O(n M)$ time by using \cref{prop:1D Ising sampling time} to sample from each $\pi_i$ in $O(M)$ time.

We prove part 2. To compute $Z$, we apply \cref{lem:modify counting to sampling} with
relative accuracy $\epsilon/2$, which reduces estimating
$\E_{\pi}[e^W]$ to approximately sampling from
$ \mu_{\lambda} = \tilt_{\lambda W}\pi$ for
$\lambda\in\{0,1/L,\ldots,1-1/L\}.$ 
 Let $K = \Theta\left((L/\epsilon)^2\right) $ and $ K' = \Theta(K \log(1/\delta)). $ 
For each $\lambda \in [0,1-1/L],$ we sample from $ \hat{\mu}_\lambda$ where $d_{TV}(\hat{\mu}_\lambda,\mu_{\lambda})\leq \epsilon_{TV}:=\frac{1}{8K}$
using part 1, which takes time
$ T' = O \left(\frac{n^2 M \log (L/\epsilon)}{1- (1-1/L) \Delta(J) \alpha_M (\eta)} \right).$
    Hence, the total runtime is 
    \[ O\left( K' (T'  + n^2 M )\right) = O( K'T') = O\left(\frac{n^2 M \log (L/\epsilon)}{1 -(1-1/L)  \Delta(J) \alpha_M (\eta) } \cdot \frac{L^2}{\epsilon^2}\cdot \log\frac{1}{\delta}\right). \]

\end{proof}

We are ready to prove \cref{thm:main theorem tractable}.
\begin{proof}[Proof of \cref{thm:main theorem tractable}]
  Define $H_z, H_x$ as in \eqref{eq:Hz and Hx}. 

Let $ O$ be an arbitrary Pauli string observable.
Conjugating by $U = \prod_{i: \eta_i < 0} Z_i$ yields $U H U^\dag = H' := H_z + \sum_{i=1}^n |\eta_i| X_i$ and $U O U^\dag = s O$, for $s \in \set{\pm 1}$ computable in $O(n)$ time. Thus $\Tr(O e^H) = s \Tr(O e^{H'})$.  Hence, it suffices to estimate $\Tr(O e^{H'}).$ We therefore assume without loss of generality that $\eta_i \geq 0\forall i.$ 

    Let $\eta \equiv \eta_{\min}.$
   We use the same notation and setup as in \cref{prop:derive expression for observables}. 

   Throughout, let
   \begin{equation}\label{eq:set M and L}
      c  =  \norm{J}_{\infty} +\norm{\mathbf{h}^z}_\infty+ \norm{\boldsymbol{\eta}}_\infty+1,\quad L = n \lceil \norm{J}_\infty \rceil +1,\quad
    M =\lceil 1000 \cdot \max \left\{\frac{n c^2 }{\epsilon} , \left (\frac{nc}{\epsilon}\right)^{3/2}\right\}\rceil.  
   \end{equation} 

We note that for parts 1 and 2, it suffices to compute a multiplicative approximation of $ \Tr(X_S e^H) $ given $S\subseteq [n].$

By \cref{prop:multiplicative approx of trotterized X observable} and the choice of $M$, we have $ \Tr(X_S e^H) =\exp(\pm \epsilon/8) \Tr(X_S Q_M). $

    To approximate $\Tr(X_S Q_M) ,$ we recall from \cref{prop:derive expression for observables} that \[\Tr(X_S Q_M) = D\E_{\tilde{\nu}}[e^W], \text{ where } D = \Tr (X_S ( e^{\frac{1}{M}\sum_i h_i^z Z_i } e^{H_x/M})^M ) = \prod_{i\in [n]} \Tr(X^{\mathbbm{1}_{i\in S}} (e^{\frac{h_i^z}{M}  Z} e^{\frac{\eta_i}{M}  X} )^M) \geq 0. \] 
    We compute $D$ in $O(nM)$ time. Next, we use \cref{thm:main mixing theorem} with accuracy parameters $\epsilon/8$ and $\delta/2$ to approximate $\E_{\tilde{\nu}}[e^W]$ and output $\widehat{\Tr(X_S e^H) } = \hat{Z} D$. We have: 
     \[\Pr[\widehat{\Tr(X_S e^H) } =\exp(\pm \epsilon/4) \Tr(X_S e^H)]  \geq \Pr[\hat{Z}= \exp(\pm \epsilon/8) \E_{\tilde{\nu}}[e^W]] \geq 1-\delta/2\]
For part 1, we apply the above procedure to $S =\emptyset.$ For part 2, we apply the above procedure for $ S$ and $\emptyset,$ compute $ \widehat{\Tr(X_S e^H) }$ and $\widehat{\Tr(e^H) },$ and output $ \widehat{\Tr(X_S \rho)} = \frac{\widehat{\Tr(X_S e^H) }}{\widehat{\Tr(e^H) }}.$ We have:
\[\Pr[ \widehat{\Tr(X_S \rho)}  = \exp(\pm \epsilon/2) \Tr(X_S \rho)] \geq 1-\delta. \]

We show that the total runtime is 
\begin{equation}\label{eq:total runtime details}
    T = O (\epsilon^{-2} \log \delta^{-1} \cdot n^2  L^3 M \log (L/\epsilon) ) = \poly(n, \max_{i,j} |J_{ij}| , \max|h_i^z|,  \max |\eta_i|, \epsilon^{-1}) \cdot \log \delta^{-1} 
\end{equation}
by arguing that
\begin{equation}\label{eq:gap condition}
(1-1/L) \Delta(J) \alpha_M(\eta)\leq 1-\frac{1}{2L}.
\end{equation}
\eqref{eq:gap condition} is trivially true when $ \eta = 0.$ When $\eta>0,$
    note that
    the choice of $M$ implies $ M \geq\frac{10 L \eta }{\epsilon},$ thus $\tanh(\eta/M)\geq \frac{\eta}{M} \left(1 -\frac{\eta^2}{M^2}\right) $ and 
    $(1-1/L) \Delta(J) \alpha_M(\eta)  \leq \frac{1 -1/L}{ 1- \frac{\epsilon^2}{100 L^2}} \cdot \Delta(J) \alpha(\eta) \leq 1-\frac{1}{2L} .$



   Next, we prove part 3. Let $\epsilon' = \frac{\epsilon}{100}.$
  Applying \cref{prop:addition apporoximation of Pauli observables}, we have $ \left|\Tr (O \rho) - \frac{\Tr(O Q_M) }{\Tr(Q_M) }\right| \leq \epsilon', $ where we recall that by \cref{prop:multiplicative approx of trotterized X observable}, $ \Tr(Q_M)> 0.$  Fix $ \lambda= 1-1/L.$ Using \cref{prop:derive expression for observables}, we can rewrite
   \[\frac{\Tr(O Q_M) }{\Tr(Q_M) } = \mathrm{i}^{|\set{i: P_i = Y_i}|} R , \quad R:= A B,\quad B:=\frac{ D\E_{\tilde{\nu}}[e^{\lambda W}] }{\Tr(Q_M) },\quad A := \frac{\E_{\tilde{\nu}}[e^W f]}{\E_{\tilde{\nu}}[e^{\lambda W}]} =  \E_{\mu_\lambda}[e^{W/L} f] \]
   where $ \mu_\lambda= \tilt_{\lambda W} \tilde{\nu},$ and $f$ is the real-valued function defined in \eqref{eq: VM and W}. 
   
   By our choice of $M$, and the fact that $ |W| \leq L,$ we have $ e^{\lambda W}\leq e^W  e^{|W|/L}\leq e^{W} \cdot e,$ thus
\[
B\leq e \cdot \frac{D  \E_{\tilde{\nu}}[e^{W}] }{\Tr(Q_M) } = e\cdot \frac{\Tr(X_SQ_M)}{\Tr(Q_M)}
\le e\bigl(\Tr(X_S\rho)+\epsilon'\bigr)
\le e(1+\epsilon')
\]

   By the above argument,  we can compute $\widehat{\E_{\tilde{\nu}}[e^{\lambda W}]} $ and $\widehat{\Tr(Q_M) } $ such that:  \[\Pr[\widehat{\E_{\tilde{\nu}}[e^{\lambda W}]} = \exp(\pm \epsilon')  \E_{\tilde{\nu}}[e^{\lambda W}] \text{ and } \widehat{\Tr(Q_M) } = \exp(\pm\epsilon') \Tr(Q_M)] \geq 1-\delta/2.\] Let $ \hat{B} = \frac{D \widehat{\E_{\tilde{\nu}}[e^{\lambda W}]}}{\widehat{\Tr(Q_M) } }.$ Then, with probability $ \geq 1-\delta/2,$ we have: \[0\leq e^{-2\epsilon'}  B \leq  \hat{B} \leq e^{2\epsilon'}  B .\]
   
   We approximate $A:=\E_{\mu_\lambda}[e^{W/L} f] .$ Note that $ |e^{W/L} f|\leq e,$ thus $|A| \leq e.$  
   Let $\hat{A} = \widehat{\E_{\mu_\lambda}[e^{W/L} f]} $ be the empirical average of $ e^{W /L} f$ over $\Theta(\frac{\log(1/\delta)}{\epsilon^2})$ independent samples from $ \hat{\mu}_\lambda$ where $d_{TV}(\hat{\mu}_\lambda,\mu_{\lambda})\leq \epsilon'/(4e).$ By Hoeffding's inequality, with probability $ \geq 1-\delta/2,$ we have:
   $ |A-\hat{A}|\leq \epsilon'.$
 The algorithm outputs  $\widehat{\Tr (O \rho)} = \mathrm{i}^{|\set{i: P_i = Y_i}|} \hat{R}$ with $\hat{R} =\hat{A}\hat{B}.$  With probability $ \geq 1-\delta,$ we have: 
   \[ |\widehat{\Tr (O \rho)} -\Tr (O \rho) | \leq \epsilon'+ |R-\hat{R}|\leq \epsilon' + \hat{B}\cdot |A -\hat{A}| + |A| \cdot |B-\hat{B}| \leq 21 \epsilon' \leq \epsilon. \]
   To sample from $ \hat{\mu}_\lambda,$ we use \cref{thm:main mixing theorem}. By \eqref{eq:gap condition}, the total runtime is $O(T)$ with $T$ as in \eqref{eq:total runtime details}. 
   \end{proof}

\cite[Section 3.2]{wong2026lee} reduces Gibbs state preparation for stoquastic Hamiltonians to approximately computing conditional marginals of the corresponding classical distribution (see \cref{thm:marginal to state reduction} from \cref{sec:Gibbs state preparation} for details). Using \cref{thm:main mixing theorem}, we can approximate these conditional marginals, and thus obtain an efficient quantum algorithm that prepares $\rho =\frac{e^H}{\Tr(e^H)}.$ We defer the proof to \cref{sec:Gibbs state preparation}.
\begin{theorem}\label{thm:Gibbs state preparation}
 Fix $\epsilon_{TV} \in (0,1).$   In the same setting as \cref{thm:main theorem tractable}, there exists a quantum algorithm that runs in $\poly(n, \max_{i,j} |J_{ij}| , \max|h_i^z|,  \max |\eta_i|, \epsilon_{TV}^{-1}) $ time and outputs a quantum state $\hat{\rho}$ where $\norm{\hat{\rho}-\rho}_1 \leq \epsilon_{TV}.$ 
\end{theorem}

\subsection{Bounding the covariance of the 1D Ising model}\label{subsec:covariance bound for 1D}



In this subsection, we prove \cref{thm:covariance bound for each block}. 
We recall the definition of $\pi_{\mathbf{K}, \mathbf{w}}^{\mathrm{1D Ising, CBC}}. $ 
\[ \forall \mathbf{s}= (s^{(k)})_{k=1}^M\in \set{\pm 1}^M: \pi_{\mathbf{K}, \mathbf{w}}^{\mathrm{1D Ising, CBC}}(\mathbf{s}) 
= \frac{\exp\left( \sum_{k=1}^{M} K^{(k)} s^{(k)} s^{(k+1)} + \sum_{k= 1}^{M} w^{(k)} s^{(k)} \right)}{Z_{\mathbf{K}, \mathbf{w}}^{\mathrm{1D Ising, CBC}}}
\]
where $s^{(M+1)} = s^{(1)}.$ 
It suffices to prove \cref{thm:covariance bound for each block} for finite fields, i.e., $\mathbf{w} \in\R^M;$ 
the infinite field cases follow by taking limits.

In what follows, we assume $\mathbf{w} \in\R^M.$
\cref{thm:covariance bound for each block} directly follows from the following bound on the entries of $\Cov(\pi_{\mathbf{K}, \mathbf{w}}^{\mathrm{1D Ising, CBC}}).$
\begin{lemma}\label{lem:bound pairwise variance for 1D Ising CBC} 
Let $K>0$ be s.t.
$
|K^{(k)}|\le K \forall k$ and $ a= \tanh K.$
   Then, $\forall i, j\in [M], i\leq j,$ let \[ A=|\prod_{k=i}^{j-1} \tanh(K^{(k)})|,\quad B= |\prod_{k\in ([j,M]\cup[1,i-1])\cap\Z}  \tanh(K^{(k)})|.\]  Then  \[|\Cov(\pi_{\mathbf{K}, \mathbf{w}}^{\mathrm{1D Ising, CBC}})_{i,j}| \leq \frac{A+B}{1+|\prod_{k=1}^{M} \tanh(K^{(k)})|}\leq \frac{a^{j-i} + a^{M- (j-i)}}{1 + a^M}.\]
\end{lemma}

\begin{proof}[Proof of \cref{thm:covariance bound for each block}]
Let $ \tilde{\pi}\equiv \tilde{\pi}_{\mathbf{K}, \mathbf{w}}$, and $\pi \equiv \pi_{\mathbf{K}, \mathbf{w}}^{\mathrm{1D Ising, CBC}}.$
\begin{align*}
    \norm{\Cov(\tilde{\pi}) }_{\op} = \frac{1}{M}     \norm{\Cov(\pi)}_{\op} \leq \frac{1}{M}\norm{\Cov(\pi) }_{\infty} =\frac{1}{M}\max_{i} \sum_j  |\Cov(\pi)_{i,j}| \leq \frac{\tanh(\eta)}{M}\coth(\frac{\eta}{M})
\end{align*}
where the last inequality follows from \Cref{lem:bound pairwise variance for 1D Ising CBC}. Indeed, by relabeling, we only need to bound $\sum_j  |\Cov(\pi)_{i,j}|  $ for $i=1.$ Let $a = \tanh K = e^{-2\eta/M} \in (0,1).$ By \Cref{lem:bound pairwise variance for 1D Ising CBC}, $ |\Cov(\pi)_{i,j}| \leq \frac{a^{j-1} + a^{M+1-j}}{1 + a^M}$ for $i =1$ and $ j\in [M],$ thus
\begin{align*}
    \sum_j  |\Cov(\pi)_{i,j}| \leq \sum_{d=0}^{M-1}
\frac{a^d+a^{M-d}}{1+a^M} = \frac{(1+a) \sum_{d=0}^{M-1} a^d   }{1 + a^M} = \frac{1+a}{1-a}\cdot \frac{1-a^M }{1+a^M} = \coth(\frac{\eta}{M}) \tanh(\eta).
\end{align*}
\end{proof}
To prove \cref{lem:bound pairwise variance for 1D Ising CBC}, we relate the entries of $\Cov(\pi_{\mathbf{K}, \mathbf{w}}^{\mathrm{1D Ising, CBC}})$ to those of certain Ising models on a path (a.k.a. 1D Ising with open boundary condition) using \cref{prop:covariance bound 2 spin}, \cref{eq:main helper for pairwise CBC covariance bound,eq:application of 2spin bound}, then bound the latter using \cref{prop:pairwise covariance bound Ising on path}.

\begin{proposition}\label{prop:covariance bound 2 spin}
    Consider $\rho:\set{\pm 1}^2\to \R_{>0}$. We can write $\rho$ uniquely as the (unnormalized) density function of an Ising model as follows. Let 
    \begin{align*}
        J &=\frac{1}{4}\ln \frac{\rho(+1,+1)\rho(-1,-1)}{\rho(+1,-1)\rho(-1,+1)},\quad v_1 =\frac{1}{4}\ln \frac{\rho(+1,+1)\rho(+1,-1)}{\rho(-1,+1)\rho(-1,-1)} , \\
        v_2 &= \frac{1}{4}\ln \frac{\rho(+1,+1)\rho(-1,+1)}{\rho(+1,-1)\rho(-1,-1)}, \quad C=\left(\prod_{x,y\in \set{\pm 1}} \rho(x,y)\right)^{1/4}.
    \end{align*}
    Then
    \[ \forall \mathbf{s}=(s^{(1)},s^{(2)}) \in \set{\pm 1}^2: \rho(\mathbf{s}) =C \exp(J s^{(1)} s^{(2)} + v_1 s^{(1)} + v_2 s^{(2)} ) .\]
     We have $|\tanh(J)|= \sup_{w\in \R^2} |\Cov(\lineartilt_w \rho)_{1,2}|.$
\end{proposition}
\begin{proof}
    The first statement is by a direct calculation. For the second statement, we first prove that \[ \sup_{\mathbf{w}\in \R^2} |\Cov(\lineartilt_w \rho)_{1,2}| \leq |\tanh(J)|.\]
    Fix $ \mathbf{w}\in \R^2,$ and let $\rho'= \lineartilt_w \rho.$
    Let $u_1 = v_1 + w_1, u_2= v_2+ w_2.$ We have  \[ \rho' (\mathbf{s}) =\frac{1}{Z}\exp(J s^{(1)} s^{(2)} + u_1 s^{(1)} + u_2 s^{(2)} ), \]
where
    $Z = 2 (e^J \cosh(u_1+u_2) + e^{-J} \cosh(u_1-u_2))\geq 2 (e^J + e^{-J}).$ Hence
    \[ \Cov(\rho')_{1,2} = 4 ( \rho'(+1,+1) \rho'(-1,-1) -\rho'(+1,-1)\rho'(-1,+1)) = \frac{4 (e^{2J} - e^{-2J})}{Z^2}\]
    thus
    \[|\Cov(\lineartilt_w \rho)_{1,2}| = \frac{4 |e^{2J} - e^{-2J}|}{Z^2} \leq  \frac{4 |e^{2J} - e^{-2J}|}{4 (e^J + e^{-J})^2}\leq |\tanh(J)| .\]
Equality is attained when $ w_1 = -v_1, w_2 = -v_2.$ Thus $|\tanh(J)|= \sup_{w\in \R^2} |\Cov(\lineartilt_w \rho)_{1,2}|.$
\end{proof}

For $\mathbf K=(K^{(k)})_{k=1}^{M-1}\in \R^{M-1}$ and $\mathbf w = (w^{(k)})_{k=1}^{M}\in \R^M,$  
 the 1D Ising model with open boundary condition (OBC) $ \pi_{\mathbf K,\mathbf w}^{\mathrm{1D Ising, OBC}}:\set{\pm 1}^M \to \R_{\geq 0}$ is defined by:
\begin{equation}\label{def:1D Ising with open boundary}
  \forall \mathbf{s}= (s^{(k)})_{k=1}^{M}\in \set{\pm 1}^M: \pi_{\mathbf K,\mathbf w}^{\mathrm{1D Ising, OBC}}(\mathbf{s})
=\frac{1}{Z_{\mathbf K,\mathbf w}^{\mathrm{1D Ising, OBC}}}
\exp\left\{
\sum_{k=1}^{M-1} K^{(k)} s^{(k)} s^{(k+1)}
+
\sum_{k=1}^{M} w^{(k)} s^{(k)}
\right\}
\end{equation}

\begin{proposition} \label{prop:pairwise covariance bound Ising on path}
Let $\mu\equiv  \pi_{\mathbf K,\mathbf w}^{\mathrm{1D Ising, OBC}}. $ Then
    $\forall i, j\in [M], i\leq j:|\Cov(\mu)_{i,j}| \leq |\prod_{k=i}^{j-1}\tanh (K^{(k)})|.$
\end{proposition}

\begin{proof}[Proof of \cref{lem:bound pairwise variance for 1D Ising CBC}]
The case $i=j$ is trivial; below we assume $i\neq j.$ 
By relabeling, we can assume $i = 1$ w.l.o.g. Let $\mu_0$ be the unnormalized density function corresponding to $\pi_{\mathbf{K}, \mathbf{w}}^{\mathrm{1D Ising, CBC}},$ defined by:
\begin{align*}
    \mu_0(s^{(1)},\cdots, s^{(M)})& = \exp\left(
\sum_{k=1}^{M} K^{(k)} s^{(k)} s^{(k+1)}
+
\sum_{k=1}^{M} w^{(k)} s^{(k)}
\right) \text{ where } s^{(M+1)} = s^{(1)}
\end{align*}
Let  $ \mu_1 ,\mu_2$ be the unnormalized density functions of 1D Ising models with open boundary conditions, defined by:
\begin{align*}
    \mu_1(s^{(1)},\cdots, s^{(j)})& = \exp\left(
\sum_{k=1}^{j-1} K^{(k)} s^{(k)} s^{(k+1)}
+
\sum_{k=1}^{j-1} w^{(k)} s^{(k)}
\right)\\
\mu_2(s^{(j)},\cdots, s^{(M)},s^{(1)}) &= \exp\left(
\sum_{k=j}^{M} K^{(k)} s^{(k)} s^{(k+1)}
+
\sum_{k=j}^{M} w^{(k)} s^{(k)}
\right) \text{ where } s^{(M+1)} = s^{(1)}.
\end{align*}
We observe that $ \mu_{0}(s^{(1)}, \cdots, s^{(M)}) = \mu_1(s^{(1)},\cdots, s^{(j)})\mu_2(s^{(j)},\cdots, s^{(M)}, s^{(1)}) .$ For $ t \in \set{0,1,2}$, let $\rho_t:\set{\pm 1}^2 \to \R_{>0}$ be the unnormalized density induced on $(s^{(1)},s^{(j)})$ by $\mu_t.$ Explicitly, for $(x,y)\in\set{\pm 1}^2:$
\begin{align*}
   \rho_0(x,y) &= \sum_{\mathbf{s}\in \set{\pm 1}^M: s^{(1)} =x, s^{(j)} = y} \mu_0 (\mathbf{s})\\
   \rho_1(x,y) &=
\sum_{s^{(2)},\ldots,s^{(j-1)}}
\mu_1(x,s^{(2)},\ldots,s^{(j-1)},y)\\
\rho_2(x,y)&=
\sum_{s^{(j+1)},\ldots,s^{(M)}}
\mu_2(y,s^{(j+1)},\ldots,s^{(M)},x).
\end{align*}
We have:
\[ \quad  \forall (x,y)\in\set{\pm 1}^2:\rho_0(x,y) = \rho_1(x,y) \rho_2(x,y).\]
We abuse notation slightly and identify $ \rho_t$ and $\mu_t$ with the corresponding probability distributions, and note that:
 \[ \Cov(\rho_t)_{1,2} = \Cov(\mu_t)_{1, j}\quad \] 
We apply \cref{prop:covariance bound 2 spin} to $\rho_t,$ and define $J_t = \frac{1}{4}\ln \frac{\rho_t(+1,+1)\rho_t(-1,-1)}{\rho_t(+1,-1)\rho_t(-1,+1)}$ accordingly. We have 
$J_0 = J_1 + J_2$ and
\begin{equation}\label{eq:main helper for pairwise CBC covariance bound}
    |\Cov(\mu_0)_{1, j}| = |\Cov(\rho_0)_{1,2}| \leq |\tanh(J_0)|\leq \tanh(|J_1|+|J_2|)
\end{equation}
where in the last inequality, we use the monotonicity of $\tanh.$
Next,  by \cref{prop:covariance bound 2 spin}, we have:
\begin{equation}\label{eq:application of 2spin bound}
  \tanh(|J_1|) =|\tanh(J_1)| =\sup_{w\in \R^2}|\Cov(\lineartilt_w \rho_1)_{1,2}| \leq \sup_{w\in \R^j} |\Cov(\lineartilt_w\mu_1)_{1,j}| \leq \prod_{k=1}^{j-1} |\tanh(K^{(k)})|=A   
\end{equation}
where the last inequality is by \cref{prop:pairwise covariance bound Ising on path}.
Similarly, $\tanh(|J_2|) \leq  \prod_{k=j}^M |\tanh(K^{(k)})| = B.$ Thus, by monotonicity of $ \tanh,$ and \cref{eq:main helper for pairwise CBC covariance bound}
\begin{align*}
    |\Cov(\mu_0)_{1, j}| \leq \tanh(|J_1|+|J_2|) &\leq \tanh(\tanh^{-1}( A)+ \tanh^{-1}(B) ) \\
    &= \frac{A+ B}{1 + \prod_{k=1}^{M} |\tanh(K^{(k)})|}
\end{align*}
where in the last equality, we use the identities $\tanh(x+y)=\frac{\tanh x + \tanh y}{1+\tanh (x) \tanh(y)}$ and $AB = \prod_{k=1}^{M} |\tanh(K^{(k)})|.$ 
When $|K^{(k)}|\leq K \forall k$, using monotonicity of $\tanh$, we have $A \leq a^{j-i}, B\leq a^{M- (j-i)}$ and \[   |\Cov(\mu_0)_{1, j}| \leq  \tanh(\tanh^{-1}( A)+ \tanh^{-1}(B) )= \frac{A+ B}{1 + \prod_{k=1}^{M} |\tanh(K^{(k)})|} \leq  \frac{a^{j-i} + a^{M-(j-i)}}{1 + a^M}.\]
\end{proof}

\begin{proof}[Proof of \cref{prop:pairwise covariance bound Ising on path}]
    For any $i,j\in [M]$, let $\Psi(i\to j) = \frac{1}{2}(\E_{s\sim \mu}[s^{(j)} |s^{(i)}=1] - \E_{\mu}[s^{(j)}|s^{(i)} = -1]).$ 
    For $i\leq j,$ a standard calculation (see e.g. \cite[Eq. (21)]{BKM05}) gives:
    \[ \Cov(\mu)_{i,j} = \Cov(\mu)_{i,i} \Psi(i\to j) = \Cov(\mu)_{i,i}  \prod_{k =i }^{j-1} \Psi(k\to (k+1))\]
    The statement follows from $|\Cov(\mu)_{i,i}|\leq 1$ and $| \Psi(k\to (k+1))| \leq |\tanh(K^{(k)})|.$ We prove $| \Psi(k\to (k+1))| \leq |\tanh(K^{(k)})|.$

    Deleting the edge between the spins $k$ and $(k+1)$ separates the path into two
components, so the marginal on these two spins has the form
\[
p(x,y)\propto
\exp(K^{(k)}xy+ux+vy)
\]
for some $u,v\in\mathbb R$. Hence
\[
|\Psi(k\to k+1)|
=\frac12|\tanh(v+K^{(k)})-\tanh(v-K^{(k)})|
=\frac{|\sinh(2K^{(k)})|}
       {\cosh(2v)+\cosh(2K^{(k)})}
\le |\tanh(K^{(k)})|.
\]
\end{proof}

\section{Lower bound}\label{sec:lower bound}

In this section, we prove \cref{thm:hard regime informal}. 

We note that the case $\eta =0$ corresponds to the classical Ising model case, and \cref{thm:hard regime informal} follows from \cite{galanis2024sampling}, so we will focus on the case $ \eta \neq 0.$
Fix $\eta\neq 0$ and $b'>b_{\min}:=\eta/\tanh\eta$. Following \cite{galanis2024sampling},
to formally state the result, we define the following computational problem.

\noindent\fbox{\begin{minipage}{\dimexpr\linewidth-2\fboxsep-2\fboxrule\relax}
\textit{Problem:} \textsc{SpectralTFIM}$(\eta,b')$\\
\textit{Input:} A symmetric matrix $J\in\R^{n\times n}$ with
$\Delta(J)\in(b_{\min},b').$ 

\textit{Output:} The partition function
$Z= \Tr\!\left(\exp\!\left(
\frac12\sum_{i,j}J_{ij}Z_iZ_j+\eta\sum_iX_i\right)\right)$.
\end{minipage}}
\vspace{0.1in}
\begin{theorem}\label{thm:main hardness of approximation}
    For every fixed $\eta\neq 0$ and $b'>\eta/\tanh\eta$, there is a constant
$c>0$ dependent only on $\eta$ and $b'$ such that approximating
\textsc{SpectralTFIM}$(\eta,b')$ within a factor $e^{cn}$ is \NP-hard.
A polynomial-time randomized (respectively quantum) algorithm for this task
would imply $\textsc{NP}=\textsc{RP}$ (respectively
$\textsc{NP}\subseteq\textsc{BQP}$).
\end{theorem}
Our hard instance involves gluing together multiple copies of Curie-Weiss gadgets. In \cref{subsec:quantum curies-weiss}, we prove properties of the transverse field Curie-Weiss model that would be used in our analysis (see \cref{thm:concentration on two wells} and \cref{thm:conditioning equal product}). \cref{subsec:hard instance} describes the hard instance(s), and \cref{subsec:analyzing hard instance} gives the proof of \cref{thm:main hardness of approximation}.

As explained in the proof of \cref{thm:main theorem tractable}, for $H_z= \frac{1}{2}\sum_{ij} J_{ij} Z_i Z_j $, $ H_z + \eta \sum_i X_i$ and $ H_z+|\eta|\sum_i X_i$ are unitarily equivalent, so we can assume wlog that $\eta>0.$
For the rest of this section, fix $\eta>0$ and let $\nu_{\eta} \equiv \nu_{\eta, (0)}$ be as defined in \cref{def:general path definition}.  
Recall that $ \nu_\eta$ is a distribution over functions from $\mS =\R/\Z$ to $\set{\pm 1}.$ It will be convenient to view $ \nu_\eta$ as a distribution over the separable Hilbert vector space $L^2(\mS)$
 by identifying a function $\sigma:\mS \to \set{\pm 1}$ with its equivalence class in $L^2(\mS).$
\subsection{The transverse-field Curie-Weiss model}\label{subsec:quantum curies-weiss}

Fix $\eta>0$ and $b>0$ such that
\begin{equation}\label{eq:assumption on b and eta}
  \frac{b}{\eta} \tanh (\eta)-1> 0.
\end{equation}


For $ n\in \N_{\geq 1},$ we let $ \mu_n \equiv \mu^{\text{Curie-Weiss}}_{n, b, \eta} $ denote the classical distribution corresponding to the Curie-Weiss TFIM model with Hamiltonian $ H = \frac{b}{2n} (\sum_{i=1}^n Z_i)^2 +\eta \sum_{i=1}^n X_i.$  
\begin{equation}\label{eq:def of curies weiss quantum modified}
    \mu_n = \tilt_{U}\nu, \qquad \nu = \otimes_{i=1}^n \nu_{\eta},\qquad U(\sigma_1,\cdots, \sigma_n) = \frac{b n}{2} \cdot\norm{\bar{\sigma}}^2 \text{ where } \bar{\sigma} = \frac{1}{n} \sum_{i=1}^n \sigma_i
\end{equation}
We prove the following:
\begin{theorem}\label{thm:concentration on two wells}

    There exist $ m_* \equiv m_*(b, \eta) \in (0,1)$ and $ c  \equiv c(b, \eta) > 0$ such that the following holds.
    For any $\delta >0,$ there exists $n_0\equiv n_0(\delta)>0$ s.t. 
    \[ \forall n\geq n_0(\delta): \Pr_{(\sigma_i)_{i\in [n]} \sim \mu_n, \bar{\sigma}  = \frac{1}{n}\sum_{i\in [n]}\sigma_i}\left[ \bar{\sigma}\not \in B(m_* \mathbf{1}, \delta) \cup B(-m_* \mathbf{1}, \delta)   \right ] \leq \exp(-n c \delta^2).\]
    Here for $ s\in \set{\pm 1},$ $B(s m_* \mathbf{1}, \delta) = \set{\sigma \in L^2 (\mS): \norm{\sigma - s m_* \mathbf{1}}< \delta}.$
\end{theorem}

The following corollary of \cref{thm:concentration on two wells} is crucial for the analysis of our hard instance construction.
\begin{corollary}\label{thm:conditioning equal product}
    Let $n , t\in \N_{\geq 1}$ with $n> t,$ and let $R = \set{t+1, \cdots, n}, S =\set{1,\cdots, t}, r := n-t.$

Let $ m_* \in (0,1)$ be as in \cref{thm:concentration on two wells}. Fix $\delta \in (0,m_*).$ 
    For $s\in \set{\pm 1}$, let \[B_s \equiv B(s m_* \mathbf{1}, \delta) = \set{\sigma \in L^2 (\mS): \norm{\sigma - s m_* \mathbf{1}}< \delta}.\]
    
    For measurable functions $ \tau_j: \mS\to \set{\pm 1}$ for $j\in S,$ and $ \sigma_i:\mS\to \set{\pm 1}$ for $i\in R,$ write
    \[\boldsymbol{\sigma} \equiv \boldsymbol{\sigma}_R  = (\sigma_i)_{i\in R}, \qquad
    \bar{\sigma} \equiv \bar{\sigma}_R= \frac{1}{r} \sum_{i \in R} \sigma_i, \qquad \boldsymbol{\tau}\equiv \boldsymbol{\tau}_S = (\tau_j)_{j\in S}, \qquad
    \bar{\tau} = \frac{1}{t} \sum_{j \in S} \tau_j.
\]
and let
\begin{equation}\label{eq:indicator Y}
  Y(\boldsymbol{\sigma}_R) = \begin{cases} 1 &\text{ if } \norm{\bar{\sigma}  -  m_* \mathbf{1}}  < \norm{\bar{\sigma}  +  m_* \mathbf{1}} 
\\   -1 &\text{ else} \end{cases}   
\end{equation}
 \begin{equation}\label{eq:rewrite U in term of average}
    \begin{split}
 U(\boldsymbol{\tau}, \boldsymbol{\sigma}) :&= \frac{b}{2r} \cdot\norm{\sum_{j\in S}\tau_j + \sum_{i\in R} \sigma_i}^2 
= \frac{br}{2} \cdot \norm{\bar{\sigma}}^2 +bt\langle \bar{\tau}, \bar{\sigma}\rangle + \frac{b t^2}{2 r} \cdot \norm{\bar{\tau}}^2      
\end{split}
    \end{equation}
     Let $\nu^{R} = \bigotimes_{i\in R} \nu_\eta.$ For $s\in \set{ \pm 1 },$ let
    \begin{equation}\label{eq:total weight at a realization}
       Q_s(\boldsymbol{\tau}) := \int \mathbbm{1}_{Y(\boldsymbol{\sigma}) = s} \exp( U(\boldsymbol{\tau}, \boldsymbol{\sigma})  ) \nu^{R}(d\boldsymbol{\sigma})
    \end{equation}
    
There exists $r_0\equiv r_0(\delta,t, b , \eta) >0$ so that for $ r \geq r_0$, the following holds:
    \[ Q_s(\boldsymbol{\tau})  = \exp(\pm 3 b t\delta) \cdot A_0 \cdot \exp(b s m_*\sum_{j\in S} \langle \tau_j ,  \mathbf{1}\rangle )  \]
where
$A_0 = \displaystyle\int \mathbbm{1}_{\bar{\sigma} \in B_{+1} } e^{\frac{b r}{2}\norm{\bar{\sigma}}^2} d\nu^R(\boldsymbol{\sigma}).$
\end{corollary}

We prove \cref{thm:concentration on two wells} by establishing a large deviation principle (LDP) for $\bar{\sigma} = \frac{1}{n}\sum_{i=1}^n \sigma_i$ where $ \boldsymbol{\sigma} =(\sigma_1,\cdots, \sigma_n)\sim \mu_n$ with a suitable rate function. 


All LDPs below have speed $n$ and use the norm topology
on the Hilbert space $L^2(\mS)$.

\begin{definition}\label{def:LDP}
A sequence of probability measures $\{\rho_n\}_{n\ge 1}$ satisfies the
large deviation principle with rate function $I$\footnote{Here, we implicitly assume that $I$ is a lower semicontinuous map into $[0,\infty]$.} if, for every measurable
set $\Gamma$,
\[
-\inf_{x\in \Gamma^\circ} I(x)
\le
\liminf_{n\to\infty} \frac{1}{n}\log \rho_n(\Gamma)
\le
\limsup_{n\to\infty} \frac{1}{n}\log \rho_n(\Gamma)
\le
-\inf_{x\in \overline{\Gamma}} I(x),
\]
where $\Gamma^\circ$ and $\overline{\Gamma}$ denote  the interior and closure of $\Gamma$, respectively.
\end{definition}

\begin{proposition}\label{lem:application of LDP}
 Let $\sigma_1,\dots,\sigma_n$ be i.i.d.\ samples from $\nu_\eta$, and let $\tilde{\rho}_n$ be the distribution of their average $ \bar{\sigma} = \frac{1}{n} \sum_{i=1}^n \sigma_i.$  
 
 Let $ F(\sigma) = \frac{b}{2} \norm{\sigma}^2 $
and $\rho_n = \tilt_{n F} \tilde{\rho}_n $, i.e., $ d\rho_n(\sigma) =\frac{\exp(n F (\sigma)) d\tilde{\rho}_n( \sigma) }{\E_{\tilde{\rho}_n} [\exp(n F(\sigma)) ]}   .$ Then  
$\rho_n$ is exactly the law of the average $\bar{\sigma}' = \frac{1}{n} \sum_{i=1}^n \sigma'_i$ where $(\sigma'_1,\dots, \sigma'_n)\sim \mu_n.$
 
 The sequence $(\tilde{\rho}_n)_{n\in\N}$ satisfies a large deviation principle (LDP) with rate \begin{equation}\label{eq:log moment generating functional} \tilde{I} (x) = \sup_{\lambda \in L^2(\mS)} (\langle \lambda, x\rangle - \Lambda_{\nu_\eta}(\lambda) )\quad \text{where}\quad 
 \Lambda_{\nu_\eta}(\lambda) = \log \E_{\sigma \sim \nu_\eta}[\exp(\langle \lambda, \sigma\rangle)].    
 \end{equation}

The sequence $ (\rho_n)_n$ satisfies a large deviation principle (LDP) with rate 
\begin{equation} \label{eq:LDP rate of tilt}
  J (x) = I(x) -\inf I \quad \text{where} \quad I=
\tilde{I} -F
\end{equation}
Moreover, 
\begin{equation}\label{eq:compare LDP rate with Phi}
    I(x)  \geq -\Phi_{b,\eta}(b x) \quad \text{where} \quad \Phi_{b,\eta}: L^2(\mS)\to \R,\quad \Phi_{b,\eta}(x) =  \Lambda_{\nu_\eta}(x) - \frac{1}{2b} \norm{x}^2.
\end{equation}
 \end{proposition}
 \begin{proof}[Proof of \cref{lem:application of LDP}]
     That the sequence $(\tilde{\rho}_n)_{n\in\N}$ satisfies an LDP with rate $ \tilde{I}$ follows from Cramer's theorem for Banach spaces \cite[Theorem 2.3]{PE99}: $ L^2(\mS)$ is a separable Banach space (see \cite[Section 4.3, Theorem 4.13]{brezis2011functional} and \cite[Section 7B, Theorem 7.24]{axler2020measure}), and for any $ t >0$, $\E_{x\sim \nu_\eta}[\exp(t \norm{x})] <\infty$  since $\nu_\eta$ is supported on the set of functions with $L_2$ norm bounded by $1$, i.e., $\mathcal{B} = \set{ f\in L^2(\mS): \norm{f}_2\leq 1}.$ 
     Since $\mathcal{B}$ is convex, the measures $ (\tilde{\rho}_n)_n$ are also supported on $\mathcal{B}.$ Since $F$ is continuous and bounded on $\mathcal{B},$   
the first statement and Varadhan's theorem (see \cite[Chapter 3, Exercise 3.9]{RS15}) imply that the sequence of measures $(\rho_n)_n$ satisfies an LDP with rate $J (x) = I(x)-\inf I$, where $I = \tilde{I}-F.$
Next, substituting $\lambda = bx$ in the definition of $ \tilde{I},$ we obtain:
\[ \tilde{I}(x)\geq b \norm{x}^2 - \Lambda_{\nu_\eta}( b x)  \]
thus
\[ I(x) = \tilde{I}(x)- F(x) \geq F(x) -  \Lambda_{\nu_\eta}( b x)  = -\Phi_{b,\eta} (bx) \]
\end{proof}
Next, we analyze the functional $\Phi_{b,\eta}$ from \eqref{eq:compare LDP rate with Phi}.
We need the following proposition, which follows from similar arguments as in \cite{CCIV08}. We defer the proof to \cref{subsec:defer from lower bound}.
\begin{proposition}\label{prop:constant field deviation}
Let $\phi(r) =\Phi_{b,\eta}(r\mathbf{1})$, $ R(r) =\sqrt{\eta^2+r^2}. $ We have:
\begin{equation}\label{eq:derivative of Phi at constant field}
       \frac{d}{d r} \Phi_{b,\eta}(r\mathbf{1})  =  \phi'(r) = r(\frac{\tanh R(r)}{R(r)} - \frac{1}{b}) 
   \end{equation}
    There exists a unique $m_*\in (0,1)$, such that, writing $ R_*=R(bm_*),$ we have:
    \begin{equation}\label{eq:def of critical magnetization}
        \frac{d}{d r} \Phi_{b,\eta}(r\mathbf{1})|_{r = b m_*} =\phi'(b m_*) =0 \Leftrightarrow  \frac{\tanh R_* }{R_*} = \frac{1}{b}.
    \end{equation}
     Moreover, for $c_1 =\frac{\eta(\frac{b}{\eta} \tanh (\eta)-1) }{8}>0$,
\begin{equation}\label{eq:bound for constant field}
 \forall m\geq 0: \phi( bm_*) - \phi(b m)  \geq c_1 (m-m_*)^2  
\end{equation}
\end{proposition}
We also need the following lemma:
\begin{lemma}\label{lem:properties of Phi}
Let $m_*\in (0,1)$ be from \cref{prop:constant field deviation}. There exists a constant $ c\equiv c(b,\eta) >0$ dependent only on $b, \eta $ such that
\begin{equation}\label{eq:distance to global well}
    \forall q\in L^2 (\mS): \Phi_{b,\eta} (b m_* \mathbf{1}) - \Phi_{b,\eta} (b q) \geq c \min_{s\in\set{\pm 1}} \norm{q - m_* s \mathbf{1}}^2 
\end{equation}
\end{lemma}

The following proposition characterizes the rate function $J$ from \cref{lem:application of LDP}.
\begin{proposition}\label{prop:inf of tilde I}
    For $I$ from \eqref{eq:LDP rate of tilt} and $m_*$ from \cref{prop:constant field deviation}, 
$\inf I = I ( m_* \mathbf{1}) = -\Phi_{b,\eta}(b m_* \mathbf{1}).$ Hence, the sequence $ (\rho_n)_n$ from \cref{lem:application of LDP} satisfies a large deviation principle (LDP) with rate \[J(x) \geq \Phi_{b,\eta} (b m_* \mathbf{1}) -\Phi_{b,\eta}  (b x) \geq c \min_{s\in\set{\pm 1}} \norm{x - m_* s \mathbf{1}}^2 ,\]
where $c$ is as in \cref{lem:properties of Phi}.
\end{proposition}

\begin{remark}\label{remark:fix gap in CCIV08}
\cite{CCIV08}
asserts the same LDP rate as in \cref{prop:inf of tilde I}, but their proof has a gap: it relies on the false statement  that $\norm{|q|- m_*\mathbf{1}}^2    \geq  c' \min_{s\in \set{\pm 1}} \norm{q - s m_* \mathbf{1}}^2$ for some constant $c'> 0$ (see \cite[Eq (7.4)]{CCIV08}). 
To see why this statement is false,
take  $q (t) =\begin{cases} m_* & t\in [0,1/2) \\- m_* & t\in [1/2,1)\end{cases}.$ Then    $\norm{|q|- m_*\mathbf{1}}^2  = 0$ but $\min_{s\in \set{\pm 1}} \norm{q - s m_* \mathbf{1}}^2 = 2 m_*^2 >0.$ \cref{prop:sign change} in the proof of \cref{lem:properties of Phi} addresses this gap.
\end{remark}
\begin{proof}[Proof of \cref{prop:inf of tilde I}]
By \cite[Lemma 4.1]{LRRS21}, $ \Lambda_{\nu_\eta}$ is convex, and its G\^ateaux derivative is \[\nabla \Lambda_{\nu_\eta}(f) =\frac{\E_{\sigma \sim \nu_\eta}[\sigma \exp(\langle f, \sigma\rangle) ] }{e^{\Lambda_{\nu_\eta} (f)} }.\] Convexity implies that $\nabla \Lambda_{\nu_\eta}(f)$ is a subgradient of $\Lambda_{\nu_\eta}$ at $f.$ In particular, for $ f = b m_* \mathbf{1}$ where $m_*$ is as defined in \cref{prop:constant field deviation}, a calculation gives $\nabla \Lambda_{\nu_\eta}(f)= m_* \mathbf{1}.$ Indeed, by the translational invariance of $\nu_\eta$, $ \nabla \Lambda_{\nu_\eta}(b m_* \mathbf{1}) =\tilde{c} \mathbf{1},$  with $R(r) =\sqrt{\eta^2 +r^2}$ and
    \[   \tilde{c} = \frac{d}{dr} \Lambda_{\nu_\eta}(r\mathbf{1})|_{r = bm_*} = \frac{r \tanh(R(r))}{R(r)} \mid_{r = bm_*} = \frac{b m_*\tanh(R(b m_*))}{R(b m_*)} = m_*\]
    where we use \eqref{eq:derivative of Phi at constant field} and \eqref{eq:def of critical magnetization}.
Since $\tilde{I}$ is the Legendre transform of $ \Lambda_{\nu_\eta}$, the Fenchel-Young identity (see  e.g. \cite[Proposition 16.10]{BC17}) at $ x =  m_* \mathbf{1}$ and $f =  b m_* \mathbf{1}$ gives \[\tilde{I}(x) = \langle x , f  \rangle - \Lambda_{\nu_\eta}(f) = b \norm{x}^2 -\Lambda_{\nu_\eta} (bx) \]
hence
\[I(m_* \mathbf{1}) =\tilde{I}(x) - F(x) =F(x) -\Lambda_{\nu_\eta} (bx) = -\Phi_{b,\eta} (b m_* \mathbf{1})\]
On the other hand, \eqref{eq:compare LDP rate with Phi} and \cref{lem:properties of Phi} together imply
\[ \inf I \geq -\sup \Phi_{b,\eta} = -\Phi_{b,\eta} (b m_* \mathbf{1})\]
hence $ \inf I  = I (m_* \mathbf{1}) = -\Phi_{b,\eta} (b m_* \mathbf{1})$ as desired.
\end{proof} 

\begin{proof}[Proof of \cref{lem:properties of Phi}]
Let $\Phi \equiv \Phi_{b,
\eta}$ and $\Lambda \equiv \Lambda_{\nu_\eta}.$

Observe that $ \norm{q - m_* \mathbf{1}}\leq \norm{q + m_* \mathbf{1}} $ iff $\int q (t) dt \geq 0$ thus
\begin{equation}\label{eq:rewriting rhs}
   \min_{s\in \set{\pm 1}} \norm{q - s m_* \mathbf{1}}^2= \int q(t)^2 dt - 2 m_* \left|\int q(t) dt\right| + m_*^2 = \norm{ |q| - m_*\mathbf{1}}^2 + 2m_* \left(\int_{\mS} |q(t)|dt -\left|\int q(t) dt\right|\right)  
\end{equation}
Let $|q|\in L^2 (\mS) $ be defined by $ |q|(t) = |q(t)|.$
We relate $ \Phi(b|q| ) $ and $ \Phi(b q)$ using the following proposition, whose proof is deferred to \cref{subsec:defer from lower bound}.
\begin{proposition}\label{prop:sign change} Consider $ f\in L^2(\mS),$ $ f_{\min}, f_{\max} \in \R_{\geq 0}$ s.t. $f_{\min} \leq \norm{f}_1 \leq f_{\max} .$ We have:
\begin{align*}
\Phi(|f|) - \Phi(f)= \Lambda (|f|) - \Lambda(f) &\geq \frac{\sech(\eta)}{2}\cdot (1-e^{\left|\int_{\mS} f dt\right| - \norm{f}_1  } )( 1- e^{-\left|\int_{\mS} f dt\right| -\norm{f}_1} )\\
  &\geq \frac{\sech(\eta) e^{-f_{\max}} (1-e^{-f_{\min}})   }{2}\cdot \left (\norm{f}_1 -\left|\int_{\mS} f dt\right|\right)   \geq 0   
\end{align*}
\end{proposition}

\cite[Lemma 5.2, Eq. (5.7)]{CCIV08} gives $\Lambda (f) \leq \int_{\mS} \Lambda(f(t) \mathbf{1}) dt$ thus $ \Phi(f)\leq  \int_{\mS} \phi(f(t)) dt. $ Applying \cite[Lemma 5.2, Eq. (5.7)]{CCIV08} to $f=b|q|$, using $\Phi(|f|)\geq \Phi(f) $ from \cref{prop:sign change}, and applying \eqref{eq:bound for constant field} pointwise, we obtain:
\begin{equation}\label{eq:first intermediate bound for potential function difference}
   \Phi (b m_* \mathbf{1}) -\Phi (bq)\geq  \Phi (b m_* \mathbf{1}) -\Phi (b|q|)  \geq \int_{\mS}\left[ \phi (b m_*) -\phi (b |q(t)|)\right] dt \geq c_1 \norm{|q|- m_*\mathbf{1}}^2   
\end{equation}

By the triangle inequality and $\norm{f}_1 \leq \norm{f}_2$ for $ f =|q|- m_*\mathbf{1}, $ we have:
\begin{equation}\label{eq:control L1 norm}
     | \norm{q}_1 - m_*| \leq\norm{|q|- m_*\mathbf{1}}_1 \leq \norm{|q|- m_*\mathbf{1}}_2.
\end{equation}
We consider two cases: 
\begin{itemize}
    \item If $ \norm{|q|- m_*\mathbf{1}}_2 \geq \frac{m_*}{2}$: By \cref{eq:control L1 norm}, we have 
\[ \left(\int_{\mS} |q(t)|dt -\left|\int q(t) dt\right|\right) \leq \norm{q}_1 \leq m_* + \norm{|q|- m_*\mathbf{1}}_1 \leq  3 \norm{|q|- m_*\mathbf{1}}_2 \]
thus by \eqref{eq:rewriting rhs} and \eqref{eq:first intermediate bound for potential function difference}, for $c_2 = \frac{c_1}{13},$ we have:
\begin{align*}
   c_2 \min_{s\in \set{\pm 1}} \norm{q - s m_* \mathbf{1}}_2^2 \leq 13 c_2\norm{|q|- m_*\mathbf{1}}_2^2  \leq  \Phi (b m_* \mathbf{1}) -\Phi(bq)
\end{align*}
\item If $ \norm{|q|- m_*\mathbf{1}}_2 < \frac{m_*}{2}:$  By \cref{eq:control L1 norm}, we have:
\[  | \norm{q}_1 - m_*| \leq\norm{|q|- m_*\mathbf{1}}_1 \leq \norm{|q|- m_*\mathbf{1}}_2 \leq\frac{m_*}{2} \Rightarrow \frac{m_*}{2} \leq \norm{q}_1 \leq \frac{3 m_*}{2} \]
This combined with \cref{prop:sign change} gives $ \Phi (b|q|) -\Phi(bq) \geq 2 m_* c_3  \left(\int_{\mS} |q(t)|dt -\left|\int q(t) dt\right|\right)$

where $ c_3 = \frac{b \sech (\eta) e^{-\frac{3 b m_*}{2}}(1-e^{-\frac{b m_*}{2}})}{4 m_*}>0$. By \eqref{eq:rewriting rhs} and \eqref{eq:first intermediate bound for potential function difference}, for $c_3' = \min \set{c_1,c_3}>0$
\[  \Phi (b m_* \mathbf{1}) -\Phi (bq) =  \Phi (b m_* \mathbf{1}) -\Phi (b|q|)+ \Phi (b|q|) -\Phi (bq) \geq c'_3    \min_{s\in \set{\pm 1}} \norm{q - s m_* \mathbf{1}}^2 \]
\end{itemize}
Thus, for $c = \min \set{c_2, c_3'}=\min \set{\frac{c_1}{13}, c_3}>0,$ $\Phi (b m_* \mathbf{1}) -\Phi (bq) \geq c \min_{s\in \set{\pm 1}} \norm{q - s m_* \mathbf{1}}^2 $ as desired.

\end{proof}

    

 We are finally ready to prove \cref{thm:concentration on two wells} and \cref{thm:conditioning equal product}.

\begin{proof}[Proof of \cref{thm:concentration on two wells}]
    The law of $ \bar{\sigma}$ is exactly $\rho_n$ from \cref{lem:application of LDP}. Let $ m_*$ be as in \cref{prop:constant field deviation}.
    Let $ \Gamma : = L^2 (\mS)\setminus (B(m_* \mathbf{1}, \delta) \cup B(-m_* \mathbf{1}, \delta)) .$ By \cref{prop:inf of tilde I}, \cref{def:LDP}, where $c'$ is the constant $c$ from \cref{lem:properties of Phi}, we have: 
    \[\limsup_{n\to\infty} \frac{1}{n} \log \Pr[\bar{\sigma} \in\Gamma ] \leq -\inf_{\overline{\Gamma}} J(\bar{\sigma}) \leq -c' \inf_{\bar{\sigma}\in \overline{\Gamma}} \min_{s\in \set{\pm 1}}\norm{\bar{\sigma} - m_* s\mathbf{1}}^2 \leq - c'\delta^2. \]
  Hence, there exists $n_0\equiv n_0(\delta)>0$ s.t. for $c =c'/2:$
    \[ \forall n\geq n_0(\delta): \Pr[\bar{\sigma} \in\Gamma ] \leq\exp(-c n\delta^2) .  \qedhere\]
\end{proof}
\begin{proof}[Proof of \cref{thm:conditioning equal product} given \cref{thm:concentration on two wells}]
For $s \in \{ \pm 1 \}$, let $\mathbf{v}_s = s m_* \mathbf{1}$,
\[
    \Omega_s = \left\{ x\in L^2(\mS): \left\|x - \mathbf{v}_s \right\| <\left\| x + \mathbf{v}_s \right\| \right\},
\]
and let $\overline{\Omega}_s = \left\{  x\in L^2(\mS): \left\|x - \mathbf{v}_s \right\| \le \left\| x + \mathbf{v}_s \right\| \right\}$ be its closure. 
By our choice of $\delta,$ $ B_s\subseteq \Omega_s$ for $s\in \set{\pm 1}.$
  

 For a measurable set $\Gamma\subseteq L_2(\mS)$, let
 \[\varphi(\Gamma)=\int \mathbbm{1}_{\bar{\sigma} \in \Gamma } e^{U(\boldsymbol{\tau}, \boldsymbol{\sigma}) } d\nu^R (\boldsymbol\sigma),\qquad \hat{\varphi} (\Gamma)  = \int \mathbbm{1}_{\bar{\sigma} \in \Gamma } e^{\frac{b r}{2}\norm{\bar{\sigma}}^2} d\nu^R (\boldsymbol\sigma)  , \qquad\tilde{\varphi}_s(\Gamma) =  e^{b t \langle \bar{\tau},  \mathbf{v}_s\rangle}\hat{\varphi} (\Gamma)  \]
Note that $Q_{+1}(\boldsymbol{\tau})   = \varphi(\Omega_{+1})$ and $  Q_{-1}(\boldsymbol{\tau})   = \varphi(\overline{\Omega}_{-1}),$ and $ A_0 = \hat{\varphi}(B_{+1}) = \hat{\varphi}(B_{-1}) $ by symmetry. Hence, it suffices to show that for large $r$:
\[\forall s \in \set{\pm 1}: e^{-2bt \delta} \tilde{\varphi}_s(B_s)\leq \varphi(\Omega_s) \leq \varphi(\overline{\Omega}_s) \leq e^{3 bt \delta} \tilde{\varphi}_s(B_s)  \]
  
Note that $ \hat{\varphi} (\Gamma) = Z \cdot \mu^{\text{Curie-Weiss}}_{r, b, \eta}\{\bar{\sigma}\in \Gamma\} $ where $Z = \E_{\nu^R} [e^{\frac{b r}{2}\norm{\bar{\sigma}}^2}].$ 
Let $D = L^2(\mS) \setminus (B_{+1} \cup B_{-1}).$
Let $c_0$ be the constant $c$ from \cref{thm:concentration on two wells} and $n_0(\delta)$ be as in
\cref{thm:concentration on two wells}.
Set $c=c_0/2$ and increase $n_0(\delta)$ so that
$e^{-cr\delta^2}\leq 1/2$ whenever $r\geq n_0(\delta)$.
Since $Z=2A_0+\hat{\varphi}(D)$, \cref{thm:concentration on two wells} gives
\[
\hat{\varphi}(D)
\leq
\frac{2e^{-2cr\delta^2}}{1-e^{-2cr\delta^2}}\cdot A_0
\leq 2e^{-cr\delta^2}A_0 =2e^{-c r \delta^2 }  \hat{\varphi} (B_{+1}) = 2e^{-c r \delta^2 }  \hat{\varphi} (B_{-1}) .
\]
This combined with
  $U(\boldsymbol{\tau}, \boldsymbol{\sigma}) - \frac{br}{2}\norm{\bar{\sigma}}^2 \leq b (\frac{t^2}{2r} + t)$ and $\left|b t \langle \bar{\tau}, \mathbf{v}_s\rangle\right| \leq  bt$ gives:
  \begin{equation}\label{eq:compare D with Bs}
     \varphi(D) \leq e^{b (\frac{t^2}{2r} + t)} \hat{\varphi}(D) \leq 2e^{ b (\frac{t^2}{2r} + t)- c r \delta^2} \hat{\varphi}(B_{s}) \leq 2e^{ b (\frac{t^2}{2r} + 2 t)- c r \delta^2} \tilde{\varphi}_s(B_{s}) \forall s\in \set{ \pm 1}. 
  \end{equation}
 Assume $ r \geq t/\delta.$ Then:
 \begin{equation}\label{eq:compare Bs ball} \exp(-2 b t\delta) \tilde{\varphi}_s(B_s) \leq \varphi(B_s) \leq \exp(2 b t\delta) \tilde{\varphi}_s(B_s)\end{equation} 
 since 
  for $ \boldsymbol{\sigma}$ s.t. $ \bar{\sigma}\in B_s$,
   \[  \left|U(\boldsymbol{\tau}, \boldsymbol{\sigma}) - \frac{br}{2}\norm{\bar{\sigma}}^2 - b t \langle \bar{\tau},  \mathbf{v}_s\rangle \right|  =\left|\frac{b t^2}{2 r} \norm{\bar{\tau}}^2 +  bt  \langle \bar{\tau}, \bar{\sigma}-\mathbf{v}_s  \rangle \right|\leq \frac{b t^2}{2 r}  +  bt \delta \leq 2 bt \delta  \]
   
  Using $ B_s \subseteq \Omega_s \subseteq \overline{\Omega}_s \subseteq D \cup B_s,$ \eqref{eq:compare D with Bs} and \eqref{eq:compare Bs ball}, we have:
   \[e^{-2bt \delta} \tilde{\varphi}_s(B_s)\leq \varphi(B_s) \leq \varphi(\Omega_s) \leq \varphi(\overline{\Omega}_s) \leq \varphi(B_s) + \varphi(D) \leq  (e^{2bt \delta} + 2e^{ b (\frac{t^2}{2r} + 2t)- c r \delta^2} ) \tilde{\varphi}_s(B_s)\leq e^{3 bt \delta} \tilde{\varphi}_s(B_s)  \]
   for $r\ge r_0(\delta, t, b, \eta),$ where \[ r_0(\delta, t, b, \eta)=\max\left\{n_0(\delta), \frac{t}{\delta}, \frac{\max(0,R')}{c\delta^2} + t \sqrt{\frac{b}{2c\delta^2}} \right\}>0 \text{ with } R' = 2bt(1-\delta)+ \ln \frac{2}{e^{bt\delta} -1}.\qedhere\]
\end{proof} 

\subsection{The hard instance}\label{subsec:hard instance}
As usual, we identify an Ising interaction matrix $J$ with a weighted graph, where the weight between vertices $i,j $ stands for the entry $J_{ij}$ of $J.$

Consider a $3$-regular graph $G = (V_G, E_G)$ with $|V_G| = n_0$ vertices, an instance of \textsc{MaxCut}.

Let $b,w\in \R$ and $r, t'\in\N$ be parameters to be chosen later.
Let $K$ be the Curie-Weiss gadget graph on $n_1:=r + 3 t'$ vertices, where the interaction between any pair of vertices of $K$, including self-loops, is $\frac{b}{r}.$ Let $ t = 3t'.$

\begin{itemize}
    \item We replace each node $v \in V_G$ with a distinct copy $K_v $ of $K$ on a set of vertices $W_v.$ Let $\hat{\mathcal{G}}$ be the disjoint union of $K_v$ for $v\in V_G,$ and $\hat{J}$ be the corresponding interaction matrix.
\item   We now describe how to encode the edges of \(G\) using connections between the gadgets.
    For each neighbor $v$ of $u,$ let $ S^v_u$ be a subset of $W_u$ with $|S^v_u| = t',$ so that $ S^v_u\cap S_u^w = \emptyset$  for $v \neq w.$  For each edge $ \set{u,v} \in E_G,$ we add a perfect matching $E_{uv}$ between $ S_u^v$ and $S_v^u,$ where each edge in the perfect matching has weight $ w. $ Let $S_u=\bigcup_{v} S^v_u, R_v = W_v \setminus S_v.$  
    
    We let $\mathcal{G}$ be the graph obtained by adding these edges to $\hat{\mathcal{G}},$ and $\tilde{J}$ be the corresponding interaction matrix.
Let $\mathcal{V} =\bigcup_{v\in V_G} W_v, \mathcal{R} =\bigcup_{v\in V_G} R_v, \mathcal{S} =\bigcup_{v\in V_G} S_v.$
\end{itemize}

\subsection{Proof of \cref{thm:main hardness of approximation}}\label{subsec:analyzing hard instance}
We use the same setup and notation as in \cref{subsec:hard instance}. Recall that $b_{\min} = \frac{\eta}{\tanh \eta }>0 .$ Let 
 $ b= \frac{b'+b_{\min}}{2}> b_{\min}$ and $w = \frac{b_{\min}- b'}{5} <0$

For a nonempty subset $M\subseteq \mathcal{V},$ let  $\sigma_M=(\sigma_i)_{i\in M}, \bar{\sigma}_M = \frac{1}{|M|}\sum_{i\in M} \sigma_i,$ and $\nu^M = \bigotimes_{i\in M} \nu_\eta.$ When $M= \mathcal{V},$ we omit the superscript and write $ \nu \equiv \nu^{\mathcal{V}}.$

As in \cref{thm:conditioning equal product},
let $Y(\sigma_M) = \begin{cases} 1 &\text{ if } \norm{\bar{\sigma}_M  -  m_* \mathbf{1}}  <\norm{\bar{\sigma}_M  +  m_* \mathbf{1}} 
\\   -1 &\text{ else} \end{cases}  $ and 
let $\mathcal{Y}(\sigma_{\mathcal{R}}) = (Y(\sigma_{R_v}))_{v\in V_G}.$ 

For $\tau_{\mathcal{S}}=(\tau_i)_{i\in \mathcal{S}}$, and $\mathbf{y}= (y_v)_{v\in V_G} \in \set{\pm 1}^{V_G},$ we let 
\[Q_{\mathbf{y}}(\tau_{\mathcal{S}}) := \int \mathbbm{1}_{\mathcal{Y}(\sigma_{\mathcal{R}})= \mathbf{y}} \exp(\hat{W}(\tau_{\mathcal{S}}, \sigma_{\mathcal{R}}) )   \nu^{\mathcal{R}} (d\sigma_{\mathcal{R}}) =  \prod_{v\in V_G}  Q_{y_v} (\tau_{S_v})   \]
where $ Q_{y_v}(\tau_{S_v}) = \displaystyle\int  \mathbbm{1}_{Y(\sigma_{R_v}) = y_v} \exp( U(\tau_{S_v}, \sigma_{R_v})) \nu^{R_v} (d\sigma_{R_v})$ with \[U (\tau_{S_v}, \sigma_{R_v}) = \frac{b}{2 r}\cdot \norm{\sum_{j \in S_v} \tau_j + \sum_{i\in R_v} \sigma_i}^2.\]

We apply \cref{thm:conditioning equal product}. Let $m_*= m_*(b,\eta) \in (0,1),$ $A_0$ and $r_0\equiv r_0(\delta,t, b , \eta) $ be as in \cref{thm:conditioning equal product}. For $ \delta \in (0, m_*)$ to be chosen later, and $ r\geq r_0$, we have: 
\[Q_{y_v}(\tau_{S_v}) = e^{\pm 3 b t\delta }\cdot A_0 \cdot e^{b y_v m_*\sum_{i\in S_v} \langle \tau_i ,  \mathbf{1}\rangle } . \]

Let
\begin{equation}\label{eq: define H for lower bound}
    H_x = \sum_{i\in \mathcal{V}} \eta X_i, \quad H_z = \frac{1}{2}\sum_{i, j \in \mathcal{V}} \tilde{J}_{ij} Z_i Z_j, \quad \hat{H}_z = \frac{1}{2}\sum_{i, j \in \mathcal{V}} \hat{J}_{ij} Z_i Z_j, \quad \hat{H} = H_x + \hat{H}_z, \quad H = H_x + H_z.
\end{equation}
By \eqref{eq:derivation of observables continuous}, 
\[ \frac{\Tr(\exp( H ) ) }{\Tr(\exp( \hat{H} ))}  = \frac{\E_{\nu }[e^W ]}{\E_{\nu } [e^{\hat{W}} ] } \]
where \[\hat{W} (\boldsymbol \sigma)= \frac{1}{2}\sum_{i,j\in \mathcal{V}} \hat{J}_{ij} \langle \sigma_i, \sigma_j\rangle , \quad W (\boldsymbol \sigma)= \frac{1}{2}\sum_{i,j\in \mathcal{V}} \tilde{J}_{ij} \langle \sigma_i, \sigma_j\rangle = \hat{W} (\boldsymbol \sigma) + w \sum_{\set{u,v}\in E_G} \sum_{\set{i,j} \in E_{uv}} \langle \sigma_i, \sigma_j\rangle. \] 

We can rewrite $ \E_{\nu } [e^{\hat{W}}] \equiv Z_{\hat{\mathcal{G}}} =\sum_{\mathbf{y} \in\set{\pm 1}^{V_G}} Z_{\hat{\mathcal{G}}} (\mathbf{y})  $ and $ \E_{\nu } [e^{W}] \equiv Z_{\mathcal{G}} =\sum_{\mathbf{y} \in\set{\pm 1}^{V_G}} Z_{\mathcal{G}} (\mathbf{y})  $ where
 \[Z_{\hat{\mathcal{G}}} (\mathbf{y}) :=  \int Q_{\mathbf{y}}(\tau_{\mathcal{S}}) d\nu^{\mathcal{S}} (\tau_{\mathcal{S}}) ,\quad  Z_{\mathcal{G}} (\mathbf{y}) :=  \int \exp\left(w \sum_{\set{u,v}\in E_G} \sum_{\set{i,j} \in E_{uv}} \langle \tau_i, \tau_j\rangle \right)Q_{\mathbf{y}}(\tau_{\mathcal{S}}) d\nu^{\mathcal{S}} (\tau_{\mathcal{S}}) .\]
We have
\begin{equation}\label{eq:rewrite partition function lower bound}
    \begin{split}
        Z_{\mathcal{G}} (\mathbf{y}) 
        &= e^{\pm 3 b t \delta n_0} A_0^{n_0}\int \exp\left(b \sum_{v \in V_G} \sum_{i\in S_v}  y_v m_*\langle \tau_i ,  \mathbf{1}\rangle + w \sum_{\set{u,v}\in E_G} \sum_{\set{i,j} \in E_{uv}}\langle \tau_i, \tau_j\rangle \right)  \prod_{i\in \mathcal{S}} d\nu_{\eta}(\tau_i)\\
        &=e^{\pm 3 b t \delta n_0} A_0^{n_0} \prod_{\set{u,v} \in E_{G}} \prod_{\set{i,j} \in E_{uv}} \left[\int \exp\left( b  m_* y_u\langle \tau_i ,  \mathbf{1}\rangle +b  m_* y_v\langle \tau_j ,  \mathbf{1}\rangle  + w \langle \tau_i, \tau_j\rangle \right) d\nu_{\eta}(\tau_i)d\nu_{\eta}(\tau_j) \right] \\
    \end{split}
\end{equation}
Let
    \begin{align*}
       \Psi_w (y_u, y_v) :&= \int \exp\left( b  m_* y_u\langle \tau_i ,  \mathbf{1}\rangle +b  m_* y_v\langle \tau_j ,  \mathbf{1}\rangle  + w \langle \tau_i, \tau_j\rangle \right) d\nu_{\eta}(\tau_i)d\nu_{\eta}(\tau_j) \\
       &= \frac{\Tr(e^{w Z_1 Z_2 + b m_* (y_u Z_1 +y_v Z_2 ) + \eta (X_1 + X_2) } ) }{\Tr(\exp(\eta (X_1 + X_2))) } \tag{by  \eqref{eq:derivation of observables continuous}}\\ 
       &= \frac{ \Tr(e^{w y_u y_v Z_1 Z_2 + b m_* (Z_1 +Z_2 ) + \eta (X_1 + X_2) })}{(2 \cosh\eta)^2 } .
    \end{align*}
    The second equality is because $w Z_1 Z_2 + b m_* (y_u Z_1 +y_v Z_2 ) + \eta (X_1 + X_2) $ and $w y_u y_v Z_1 Z_2 + b m_* (Z_1 +Z_2 ) + \eta (X_1 + X_2) $ are unitarily equivalent via conjugation by $ X_1^{(y_u-1)/2} X_2^{(y_v-1)/2}.$ 
    
    For $ w'\in \R,$ let $ F(w') =\Tr(e^{w' Z_1 Z_2 + b m_* (Z_1 +Z_2 ) + \eta (X_1 + X_2) }) .$
    In \cref{subsec:defer from lower bound}, we show that 
    \begin{equation}\label{eq:justify edge interaction}
        F(w') > F(-w') \text{ if }  w' >0.
    \end{equation}
 For $ w'\in \R,$ we let
    \begin{align*}
      \kappa_{w'} :=\frac{1}{2}\log  \frac{\Psi_{w'} (+1,-1) }{\Psi_{w'}(+1,+1)} = \frac{1}{2}\log \left(\frac{F(-w')}{F(w')}\right), \quad 
      C_{w'}:=\left(\Psi_{w'}(+1,+1) \Psi_{w'}(+1,-1)\right)^{t' |E_G|/2}  e^{-\kappa_{w'} t' |E_G|}
    \end{align*}
    We note that $\kappa_0 = 0$ and
    $\kappa_w>0$ by \eqref{eq:justify edge interaction} and the fact that $w <0,$ and that we can efficiently compute $ C_{w'}$ given $w'\in \R.$
    
    Let $\mathrm{cut}(\mathbf{y}):=|\set{\set{u,v}\in E_G : y_u y_v = -1}|.$
    We can rewrite \eqref{eq:rewrite partition function lower bound} as:
    \[ Z_{\mathcal{G}} (\mathbf{y}) = 
    e^{\pm 3 bt \delta n_0} A_0^{n_0}  C_w e^{2 \kappa_w \cdot t' \cdot \mathrm{cut}(\mathbf{y})} \]
    Applying the same argument for $w' = 0$, and using that $\kappa_0 = 0$, we have $ Z_{\hat{\mathcal{G}}} (\mathbf{y}) = e^{\pm 3 bt \delta n_0} A_0^{n_0}  C_0.$
    
    Let $ \mathrm{maxcut(G)} =\max_{\mathbf{y}} \mathrm{cut}(\mathbf{y})$ be the value of the maxcut instance $G.$ 
    We have $Z_{\hat{\mathcal{G}}} =  e^{\pm 3 bt \delta n_0} A_0^{n_0}  C_0 2^{n_0}$ and for $ \kappa'=2\kappa_w:$ 
    \begin{align*}
        e^{-3 bt \delta n_0} A_0^{n_0}  C_w e^{\kappa' t' \cdot \mathrm{maxcut(G)}} \leq &Z_{\mathcal{G}} = \sum_{\mathbf{y} \in\set{\pm 1}^{V_G}} Z_{\mathcal{G}} (\mathbf{y}) \leq e^{3 bt \delta n_0} A_0^{n_0}  C_w 2^{n_0} e^{\kappa' t' \cdot \mathrm{maxcut(G)}}
    \end{align*}
    thus, assuming $\delta \leq \frac{1}{3 t'},$ we have that for $  c' = c'(b,\eta,w ) = \frac{6 b + \log 2}{\kappa'}$:
    \begin{equation}\label{eq:ratio approximate maxcut}
      -\frac{c'}{t'} \cdot n_0\leq \frac{1}{\kappa' t' } \left(\log \frac{Z_{\mathcal{G}}}{Z_{\hat{\mathcal{G}}}} - \log\frac{C_w}{C_0} \right) - \mathrm{maxcut(G)}  \leq \frac{c'}{t'} \cdot n_0   
    \end{equation}

    Note that $ \hat{J}$ is a block matrix consisting of blocks $\mathbf{J}_v = \frac{b}{r}\cdot \mathbf{1} \mathbf{1}^\intercal$ for $v\in V_G,$ and thus 
    $ \Delta(\hat{J}) = b \cdot \frac{n_1}{r}.$ Let $J' = \tilde{J}-\hat{J}$ and note that the added edges form a matching, so $J'$ is a direct sum of \(2\times2\) blocks of the form $ \begin{bmatrix} 0 & w' \\ w' & 0 \end{bmatrix}$ with $w'\in \set{0,w},$ thus 
    $ \norm{J'
    }_{\operatorname{op}}  = |w|.$  By Weyl's eigenvalue perturbation bound (see \cite[Chapter 3]{bhatia2007perturbation}), we have:  
    \[ |\Delta(\tilde{J})- \Delta(\hat{J})| \leq 2 \norm{J'
    }_{\operatorname{op}} = 2 |w|.  \]
    Clearly, $ \max\set{\norm{\tilde{J}}_\infty, \norm{\hat{J}}_\infty}\leq  b \cdot \frac{n_1}{r}  +  |w|. $
    Recall that $ n_1 = 3 t'+r.$ Choose $r$ s.t. $ r\geq  16 t'\cdot \frac{b'+b_{\min}}{b'- b_{\min}}. $ Then for $J \in\set{\tilde{J},\hat{J}},$ we have 
    $ \Delta(J) \in (b_{\min},b')$ and $ \norm{J}_\infty\leq b'.$

    Since \textsc{MaxCut} is APX-hard even for 3-regular graphs \cite{alimonti1997hardness,galanis2024sampling}, we can choose a constant $\zeta \in (0,1)$ s.t. there exists a family of $3$-regular graphs where given a graph $G$ from this family, it is NP-hard to compute $\widehat{\mathrm{maxcut(G)}}$ s.t. $ |\widehat{\mathrm{maxcut(G)}} -\mathrm{maxcut(G)}| \leq \zeta\mathrm{maxcut(G)}. $



Let $c>0$ be chosen later, and suppose an algorithm approximates
\textsc{SpectralTFIM}$(\eta,b')$ within a factor $e^{c|\mathcal V|}$.
We run this algorithm separately on the two valid inputs $\tilde J$
and $\hat J$. When both calls succeed, their outputs satisfy
$\widehat{\Tr(e^H)}=e^{\pm c|\mathcal V|}\Tr(e^H)$ and
$\widehat{\Tr(e^{\hat H})}=e^{\pm c|\mathcal V|}\Tr(e^{\hat H})$.

We can efficiently compute the estimator \[ \widehat{\mathrm{maxcut(G)}} := \frac{1}{\kappa' t' } \left(\log\frac{\widehat{\Tr(e^{H}) }}{\widehat{\Tr(e^{\hat{H}}) } } - \log\frac{C_w}{C_0}\right)\] of $\mathrm{maxcut(G)}.$ 
    By \eqref{eq:ratio approximate maxcut}, and the fact that \[\mathrm{maxcut(G)}  \geq \E_{\mathbf{y} \sim \text{Uniform}(\set{\pm 1}^{V_G})} [\mathrm{cut}(\mathbf{y})] = \frac{3 n_0}{4},\] the estimation error is bounded by:
    \[ \left|  \widehat{\mathrm{maxcut(G)}} -\mathrm{maxcut(G)}   \right|\leq n_0\left(\frac{2c n_1}{\kappa' t'} + \frac{c'}{t'}\right)  \leq\frac{4}{3} \cdot \mathrm{maxcut(G)} \left(\frac{2c (3t' + r) }{\kappa' t'} + \frac{c'}{t'}\right) \leq\mathrm{maxcut(G)} \zeta \]
    when $ t' \geq \frac{8 c'}{3\zeta}$ and $ c \leq \frac{3}{8}\cdot \frac{\zeta }{2 }\cdot \frac{\kappa' t'}{(3t'+r)}. $ Set
    \[ t' = \lceil\frac{8 c'}{3\zeta}\rceil, \text{  } \delta = \min \set{\frac{1}{3t'}, \frac{m_*}{2}} , \text{  } r =\left\lceil \max \left\{ 16 t'\cdot \frac{b'+b_{\min}}{b'- b_{\min}} ,  r_0(\delta,3t', b , \eta)  \right\} \right\rceil,\text{  } c = \frac{3 \zeta \kappa't'}{16 (3t'+r)}. \]

Both \(\tilde J\) and \(\hat J\) have spectral widths in \((b_{\min},b')\). Thus, running an approximation algorithm for \textsc{SpectralTFIM}($\eta,b'$) separately on these two inputs would give a \(\zeta\)-relative approximation to \(\mathrm{maxcut}(G)\). If the algorithm succeeds on each input with probability greater than \(2/3\), taking the median of a constant number of independent runs makes both estimates correct with probability greater than \(9/10\). This contradicts the hardness of approximating \textsc{MaxCut} unless $\NP=\RP$. The same reduction with a quantum approximation algorithm would imply $\NP\subseteq\BQP$.

\begin{remark}\label{remark:hardness finite precision and zero diagonal}
The hardness from \cref{thm:main hardness of approximation} persists even if we restrict to matrices $J$ with zero diagonal and rational entries of bit length bounded by a constant depending only
on $\eta$ and $b'$.
 Indeed, choose rational
\(b,w\) satisfying for $ D= \frac{b'-b_{\min }}{5} >0,$
\[
 b_{\min}+\frac{9D}{20}<b<b_{\min}+\frac D2,
 \qquad \frac D{10}<-w<\frac D5,
\]
and take $r$ such that
$b(1+3t'/r)+2|w|<b',$ the same hardness reduction holds.
Since $b,w,r$ depend only on $\eta$ and $b'$, and every entry of
$\hat J$ and $\tilde J$ belongs to $\{0,b/r,w\}$, their entries
are rational numbers with $O_{\eta, b'}(1) $ bit length.

Subtracting $\frac{b}{r}\cdot I$ from $\hat{J} $ and $ \tilde{J}$ 
yields matrices $\hat{J}^o, \tilde{J}^o$ with  zero diagonal entries, while keeping the spectral widths and the ratio between the corresponding partition functions, and thus the hardness reduction, unchanged.
\end{remark}
\bibliographystyle{abbrv}
\bibliography{refs}
\appendix
\section{Deferred proofs}
\subsection{From \cref{sec:expression for observables}}\label{subsec:defer from observables}

\begin{proof}[Proof of \cref{prop:approximation of path functional}]

We prove the proposition for the case $\eta > 0;$ the case $\eta =0$ follows by taking the limit $\eta \to 0.$
Fix $\mathbf{s} = (s^{(0)}, \cdots, s^{(M-1)})\in \set{\pm 1}^M,$ and let $ s^{(M)} = (-1)^p s^{(0)}.$
Recall the definition of $\mathcal{J}(\sigma)  $ from  \cref{def:general path definition}. 

Note that $\sigma(k/M) = s^{(k)}\forall k\in \set{0,\cdots, M-1}$ iff \[\sigma(0) = s^{(0)}\quad \text{and}\quad \forall k\in [M]: (-1)^{\left|\mathcal{J}(\sigma) \cap ((k-1)/M, k/M]\right|} = s^{(k-1)}s^{(k)} .\]
For $k\in[M],$ let $p_k\in \set{0,1}$ be s.t. $ (-1)^{p_k} = s^{(k-1)}s^{(k)}.$ 

Note that $(-1)^{\sum_{k=1}^M p_k} =s^{(0)} s^{(M)} = (-1)^p$ so $\sum_k p_k\equiv p \bmod 2.$  

For $d_1,\cdots, d_M\in \N$ where $ d_k \equiv p_k \bmod 2, $ let  $d:=\sum_k d_k $ and 
\begin{align*}
   E_{d_1,\cdots, d_M} :&=  \int  \mathbbm{1}\set{\sigma(0) = s^{(0)} \land \left|\mathcal{J}(\sigma) \cap ((k-1)/M, k/M]\right| = d_k \forall k\in [M]} \nu_{\eta, (p)}(d\sigma) \\
   &= \frac{\eta^d}{Z_{\eta, (p)} d!} \times \frac{1}{2}\times  \frac{d!}{d_1! \cdots d_M!}\times \frac{1}{M^d} \\
   &=\frac{1}{2} Z_{\eta, (p)}^{-1} \prod_{k=1}^M \frac{(\eta/M)^{d_k}}{d_k!}. 
\end{align*}
then 
\begin{align*}
    \hat{\nu}_{\eta,(p), M}(\mathbf{s} )
    &=\sum_{(d_k)_{k=1}^{M}: d_k\equiv p_k\bmod 2 } E_{d_1,\cdots, d_M}
    = \frac{1}{2}Z_{\eta, (p)}^{-1} \sum_{(d_k)_{k=1}^{M}: d_k\equiv p_k\bmod 2 }\prod_{k=1}^M \frac{(\eta/M)^{d_k}}{d_k!}\\
    &= \frac{1}{2}Z_{\eta, (p)}^{-1} \prod_{k=1}^M \left(\sum_{d_k: d_k\equiv p_k\bmod 2 } \frac{(\eta/M)^{d_k}}{d_k!}\right)
    = \frac{1}{2}Z_{\eta, (p)}^{-1} \prod_{k=1}^M Z_{\eta/M, (p_k)} \\
    &= \frac{C^M}{2  Z_{\eta, (p)}}   \exp\left(K\sum_{k=0}^{M-2} s^{(k)} s^{(k+1)} + (-1)^p K s^{(M-1)} s^{(0)} \right), \quad C = \sqrt{\sinh(\eta/M)\cosh(\eta/M)},
\end{align*}
where we use \eqref{eq:normalization factor for each index}.
\end{proof}

\begin{proof}[Proof of \cref{prop:derive expression for observables}]
{\small We prove the identities assuming $\eta_i>0$ for every $i$; the cases
with $\eta_i=0$ follow by taking the corresponding limits
$\eta_i\downarrow0$ on both sides.
Recall $v$ from \eqref{eq:function Ising} and $V,V^{(M)},W,f$ from
\eqref{eq: VM and W}. 


For $Q_M=(e^{H_z/M}e^{H_x/M})^M$, inserting resolutions of the
identity in the $Z$-basis gives
\begin{equation}\label{eq:rewrite observable first}
\begin{aligned}
\Tr(OQ_M)
&=\sum_{(\mathbf s^{(k)})_{k=0}^{M}}
 \bra{\mathbf s^{(M)}}O\ket{\mathbf s^{(0)}}
 \prod_{k=0}^{M-1}
 \left(
 \bra{\mathbf s^{(k)}}e^{H_z/M}\ket{\mathbf s^{(k)}}
 \bra{\mathbf s^{(k)}}e^{H_x/M}\ket{\mathbf s^{(k+1)}}
 \right)\\
&=\sum_{(\mathbf s^{(k)})_{k=0}^{M}}
 \exp\!\left(\frac1M\sum_{k=0}^{M-1}
                  v(\mathbf s^{(k)})\right)
 \prod_{i=1}^n
 \left(
 \bra{s_i^{(M)}}P_i\ket{s_i^{(0)}}
 \prod_{k=0}^{M-1}
 \bra{s_i^{(k)}}e^{\eta_iX/M}\ket{s_i^{(k+1)}}
 \right).
\end{aligned}
\end{equation}
In the above, we write each $Z$-basis vector $ \ket{\mathbf{s}^{(k)}}$ as
$\ket{\mathbf s^{(k)}}=\bigotimes_{i=1}^n\ket{s_i^{(k)}}$ where $\ket{s^{(k)}_i}\in \set{\ket{+1}, \ket{-1}}$ and $ \ket{+1} =\begin{bmatrix} 1 \\ 0\end{bmatrix}, \ket{-1}=\begin{bmatrix} 0 \\ 1\end{bmatrix} .$

Let
\[
A_i\bigl((s_i^{(k)})_{k=0}^{M}\bigr)
:=
\bra{s_i^{(M)}}P_i\ket{s_i^{(0)}}
\prod_{k=0}^{M-1}
\bra{s_i^{(k)}}e^{\eta_iX/M}\ket{s_i^{(k+1)}}.
\]
We have:
\[
\bra{s_i^{(M)}}P_i\ket{s_i^{(0)}}
=
\mathbbm{1}_{\{s_i^{(M)}=(-1)^{p_i}s_i^{(0)}\}}\,
\mathrm{i}^{\mathbbm{1}_{\{P_i=Y_i\}}}
(s_i^{(0)})^{\mathbbm{1}_{\{P_i\in\{Y_i,Z_i\}\}}}.\footnote{Here we abuse notation slightly and use \(P_i\) to denote the corresponding $2\times 2$ matrix that acts on site $i.$}
\]
For $\eta_i'=\eta_i/M$, put
\[
C_i=\sqrt{\sinh(\eta_i')\cosh(\eta_i')},
\qquad
K_i=K_M(\eta_i)=\frac12\log\coth(\eta_i').
\]
Then
\[
\prod_{k=0}^{M-1}
\bra{s_i^{(k)}}e^{\eta_i'X}\ket{s_i^{(k+1)}}
=
C_i^M
\exp\!\left(
K_i\sum_{k=0}^{M-1}s_i^{(k)}s_i^{(k+1)}
\right).
\]
Consequently,
\begin{align*}
A_i\bigl((s_i^{(k)})_{k=0}^{M}\bigr)
&=
\mathbbm{1}_{\{s_i^{(M)}=(-1)^{p_i}s_i^{(0)}\}}\,
C_i^M
(s_i^{(0)})^{\mathbbm{1}_{\{P_i\in\{Y_i,Z_i\}\}}}
\mathrm{i}^{\mathbbm{1}_{\{P_i=Y_i\}}}
\exp\!\left(
K_i\sum_{k=0}^{M-1}s_i^{(k)}s_i^{(k+1)}
\right)\\
&=
\mathbbm{1}_{\{s_i^{(M)}=(-1)^{p_i}s_i^{(0)}\}}\,
C_i^M
(s_i^{(0)})^{\mathbbm{1}_{\{P_i\in\{Y_i,Z_i\}\}}}
\mathrm{i}^{\mathbbm{1}_{\{P_i=Y_i\}}}
Z_{\mathbf K_{\eta_i,(p_i)},\mathbf0}^{\mathrm{1D Ising,CBC}}\,
\pi_{\mathbf K_{\eta_i,(p_i)},\mathbf0}^{\mathrm{1D Ising,CBC}}
\bigl((s_i^{(k)})_{k=0}^{M-1}\bigr).
\end{align*}

Substituting into \eqref{eq:rewrite observable first}, and letting $ Q_M^{(1)} = (e^{\sum_{i=1}^n h_i^z Z_i/M} e^{H_x/M} )^M,$ we have:
\begin{align*}
\Tr\!\left(X_S e^{H_x}\right)& = \Tr\!\left(X_S (e^{H_x/M})^M\right) =  \prod_{i=1}^n \left(C_i^M  Z_{\mathbf{K}_{\eta_i, (p_i)}, \mathbf{0}}^{\mathrm{1D Ising, CBC}}\right)\geq 0\\
\Tr\!\left(X_S Q_M^{(1)}\right)&= \Tr\!\left(X_S (e^{\sum_{i=1}^n h_i^z Z_i/M} e^{H_x/M} )^M\right) =  \prod_{i=1}^n \left(C_i^M  Z_{\mathbf{K}_{\eta_i, (p_i)}, \frac{h_i^z}{M} \cdot \mathbf{1}}^{\mathrm{1D Ising, CBC}}\right)\geq 0
\end{align*}
\begin{align*}
\Tr(OQ_M)
&=
\mathrm{i}^{|\{i:P_i=Y_i\}|}C
\sum_{(\mathbf s^{(k)})_{k=0}^{M-1}}
\exp\!\left(\frac1M\sum_{k=0}^{M-1}
                  v(\mathbf s^{(k)})\right)
\left(\prod_{i:P_i\in\{Y_i,Z_i\}}s_i^{(0)}\right)
\prod_{i=1}^n
\pi_{\mathbf K_{\eta_i,(p_i)},\mathbf0}
^{\mathrm{1D Ising,CBC}}
\bigl((s_i^{(k)})_{k=0}^{M-1}\bigr)\\
&=
\mathrm{i}^{|\{i:P_i=Y_i\}|}C\,
\mathbb E_{\hat\nu}[e^{V^{(M)}}f]\\
&=
\mathrm{i}^{|\{i:P_i=Y_i\}|}D\,
\mathbb E_{\tilde\nu}[e^Wf].
\end{align*}
This proves \eqref{eq:derivation of observables discretized}. To show \eqref{eq:derivation of observables continuous}, we use \[ \Tr\!\left(O Q_M\right) = \mathrm{i}^{|\set{i: P_i= Y_i}|} \Tr\!\left(X_S e^{H_x}\right) \E_{\hat{\nu}}[e^{V^{(M)}} f] =\mathrm{i}^{|\set{i: P_i= Y_i}|} \Tr\!\left(X_S e^{H_x}\right) \E_{\nu}[e^{V^{(M)}} f]\] and take the limit $M\to +\infty.$
By the Schatten-Holder inequality and the Lie-Trotter product formula ($\norm{Q_M - \exp(H_z + H_x)}_{\operatorname{op}} \xrightarrow[]{M\to \infty} 0 $), we have:
\[ \left|\Tr\!\left(O Q_M\right) -\Tr\!\left(O e^H \right) \right| \leq \norm{O}_{1} \cdot \norm{Q_M - e^H}_{\operatorname{op}} \xrightarrow[]{M\to \infty} 0  \]
thus $ \lim_{M\to \infty}\Tr\!\left(O Q_M\right) = \Tr\!\left(O e^H \right).$
Almost every path has finitely many jumps, so $V^{(M)}\to V$.
Moreover, $v$ is bounded on $\{\pm1\}^n$ and $|f|\leq1$.
Dominated convergence therefore gives
\[\lim_{M\to \infty}\E_{\nu} \left[e^{V^{(M)}(\boldsymbol\sigma)}f (\boldsymbol \sigma) \right]= \E_{\nu} \left[e^{V(\boldsymbol\sigma)} f (\boldsymbol \sigma)\right],\] thus
\[\Tr\!\left(O e^H \right) = \mathrm{i}^{|\set{i: P_i= Y_i}|} \Tr\!\left(X_S e^{H_x}\right) \E_{\nu} \left[e^{V(\boldsymbol\sigma)}f (\boldsymbol \sigma)\right] \]
}
\end{proof}
\begin{proof}[Proof of \cref{prop:multiplicative approx of trotterized X observable}]
  We prove the identities assuming $\eta_i>0$ for every $i$; the cases with $\eta_i=0$ follow by taking the limit $\eta_i\downarrow0$ on both sides.
  
  Recall $V, V^{(M)}$ from \eqref{eq: VM and W}, $H_z, H_x$ and $\nu$ from \cref{sec:expression for observables}. Let  $ \mu = \tilt_{V}\nu .$ 
\eqref{eq:derivation of observables continuous} and \eqref{eq:derivation of observables discretized}   give $\Tr(X_S Q_M) = \Tr(X_S e^H) R_M$ where
    \[ R_M : = \frac{\E_{\nu}[e^{V^{(M)}}]}{\E_{\nu}[e^{V}]}= \E_{\mu}[\exp(V^{(M)}-V)] \]

For $\sigma$ as constructed in \cref{def:general path definition}, and $k\in [0,M-1]\cap \N,$ let 
 $N_k(\sigma)=|\mathcal{J}(\sigma) \cap(\frac{k}{M},\frac{k+1}{M}] |$ be the number of jumps in $(\frac{k}{M},\frac{k+1}{M}], $ and $N(\sigma) = |\mathcal{J}(\sigma)|=\sum_{k=0}^{M-1} N_k(\sigma).$ Since each jump changes $v$ by at most $c,$ we have:
\[\forall t \in [\frac{k}{M},\frac{k+1}{M}):
|v(\boldsymbol\sigma(t))-v(\boldsymbol\sigma(k/M))|
\le c \sum_{i=1}^n N_k(\sigma_i).
\]
This, combined with the triangle inequality, gives:
\begin{align*}
|V(\boldsymbol\sigma)-V^{(M)}(\boldsymbol\sigma)|
&\le \sum_{k=0}^{M-1}
 \int_{k/M}^{(k+1)/M}
 |v(\boldsymbol\sigma(t))-v(\boldsymbol\sigma(k/M))|\,dt
\le \frac1M\sum_{k=0}^{M-1}\sum_{i=1}^n
 c N_k(\sigma_i)
\leq \frac{c N(\boldsymbol\sigma)}{M}.
\end{align*}
where $N (\boldsymbol\sigma)=\sum_{i} N(\sigma_i).$ We thus have
\begin{equation}\label{eq:control ratio by number of flips}
    (\E_\mu[e^{cN/M}])^{-1}
\le \E_\mu[e^{-cN/M}]
\le R_M
\le \E_\mu[e^{cN/M}]
\end{equation}

Next, we bound $ \E_{\mu} [e^{c N/M}]$
 by performing a change of measure from $ \mu $ to $\nu$. Assuming $c/M\le1,$ the function
$f(x) = x^{c/M}$ is concave, thus Jensen's inequality gives \[\E_\mu[e^{c N/M}]=\E_\mu[f(e^{N})] \leq f(\E_\mu[e^{N}]) 
= \E_\mu[e^{N}]^{c/M}\]
Since $\frac{e^V}{\E_{\nu} [e^V]} \leq e^{cn},$ $ \E_\mu[e^{N}]= \E_\nu\left[\frac{e^V}{\E_{\nu} [e^V]} \cdot e^{N}\right] \leq e^{c n} E_{\nu}[e^{N}].$ Hence
\begin{equation}\label{eq:change of measure from mu to nu}
    \log \E_\mu[e^{c N/M}] \leq \frac{c}{M} (cn + \log E_{\nu}[e^{N}]) 
\end{equation}
Note that $\nu=\otimes_i \nu^{(i)}$, so $E_{\nu}[e^{N}] = \prod_{i} \E_{\nu^{(i)}}[e^{ N}].$ We show
\begin{equation}\label{eq:bounding jump function}
   \log \E_{\nu^{(i)}}[e^{N}] = \log \E_{\nu_{\eta_i, (p_i)}}[e^{N}]\leq 2 \eta_i + 1
\end{equation}
then conclude that \[\log E_{\nu}[e^{N}] = \sum_{i=1}^n\log\E_{\nu^{(i)}}[e^{ N}] \leq 2\sum_i \eta_i + n \]
We set $M\geq \frac{c}{\epsilon} ((c+1) n +2 \sum_{i} \eta_i)\geq c.$ Then \eqref{eq:control ratio by number of flips}, \eqref{eq:change of measure from mu to nu} give the desired result.

We finish the proof by proving \eqref{eq:bounding jump function}, i.e., we want to show that $\log \E_{\nu_{\eta, (p)}}[e^N]\leq 2\eta + 1.$
This is obviously true when $ \eta = 0.$ Assume $\eta > 0.$
\[\E_{\sigma \sim \nu_{\eta, (p)}}[e^{ t N(\sigma) }]  = \sum_{\substack{d\geq 0\\d\equiv p \mod 2}} e^{td} \frac{\eta^d }{d! Z_{\eta, (p)}} =\sum_{\substack{d\geq 0\\d\equiv p \mod 2}}  \frac{(e^t\eta)^d }{d! Z_{\eta, (p)}} = \frac{Z_{\eta e^t, (p)}}{Z_{\eta, (p)}}\] and \[ \frac{d}{d\eta} \log Z_{ \eta, (p)} =(\tanh \eta)^{1-2p} \leq 1 + \frac{1}{\eta} \Rightarrow \log \frac{Z_{ e\eta, (p)} }{  Z_{ \eta, (p)}}\leq \int_{\eta}^{e\eta} (1 + \frac{1}{t}) dt \leq (e-1)\eta + 1 \] 
Hence, $\log E_{\nu}[e^{N}] = \sum_{i=1}^n\log\E_{\nu^{(i)}}[e^{ N}] =\sum_{i=1}^n\log \frac{Z_{ e \eta_i, (p_i)}}{Z_{\eta_i, (p_i)}} \leq 2\sum_i \eta_i + n.$

\end{proof}

To prove \cref{prop:addition apporoximation of Pauli observables}, we need the following helper proposition.
\begin{proposition}\label{prop:compare tv distance between gibbs state}
For Hermitian matrices $K,C\in\mathbb{C}^{d\times d}$,
\[
\left\|
\frac{e^{K+C}}{\Tr(e^{K+C})}
-\frac{e^K}{\Tr(e^K)}
\right\|_1
\le \|C\|_{\mathrm{op}}.
\]
\end{proposition}

\begin{proof}[Proof of \cref{prop:compare tv distance between gibbs state}]
Set $\rho=e^{K+C}/\Tr(e^{K+C})$,
$\sigma=e^K/\Tr(e^K)$, and
$\delta=\|\rho-\sigma\|_1$. Assume wlog that $\delta>0.$
Write $D(\rho\|\sigma)=\Tr[\rho(\log\rho-\log\sigma)].$
Quantum Pinsker's inequality and Schatten--Holder give
\[
\begin{aligned}
\delta^2
&\le D(\rho\|\sigma)+D(\sigma\|\rho)
=\Tr[(\rho-\sigma)(\log\rho-\log\sigma)]
=\Tr[(\rho-\sigma)C]
\le \delta\|C\|_{\mathrm{op}},
\end{aligned}
\]
where the scalar normalization terms cancel since
$\Tr(\rho-\sigma)=0$.
Dividing by $\delta$ proves the claim.
\end{proof}
\begin{proof}[Proof of \cref{prop:addition apporoximation of Pauli observables}]
Let 
\[ 
S_M=e^{H_x/(2M)}e^{H_z/M}e^{H_x/(2M)},\quad
Q'_M=S_M^M, \quad
A=e^{H_x/(2M)}\]
Since
$Q_M=A^{-1}Q'_M A$, and $S_M=C^\dagger C\succ0$, where
$C=e^{H_z/(2M)}e^{H_x/(2M)}$, we have
$\Tr(Q_M)=\Tr(Q'_M) := Z_M>0.$ Let $\rho_M:=Q'_M/Z_M\succ 0,$ then $\|\rho_M\|_1=\Tr(\rho_M) =1$

Conjugation by $O=\prod_i P_i$ gives $O H_x O=H_x-2B$ where $B=\sum_{i\in D}\eta_i X_i$ and $D=\{i:P_i\in\{Y_i,Z_i\}\}. $
Diagonalizing in the $X$-basis gives:
\[
AO A^{-1}O  =e^{H_x/(2M)} O e^{-H_x/2M} O 
=e^{H_x/(2M)}e^{-(H_x-2B)/(2M)} 
=e^{B/M}
\]
and
\[ \Tr(O Q_M) = \Tr(O A^{-1}Q'_M A) = \Tr\left((A O A^{-1} O) (OQ'_M)\right ) = \Tr(e^{B/M} O Q'_M)\]
Schatten--Holder inequality gives
\begin{equation}\label{eq:helper compare Q with Q'}
\begin{split}
\frac{\left|\Tr(O Q_M)
             -\Tr(O Q'_M)\right|}{Z_M}
&=\left|\Tr\!\left((e^{B/M}-I)O\rho_M\right)\right|
\le \|e^{B/M}-I\|_{\mathrm{op}}\norm{O}_{\mathrm{op}} \norm{\rho_M}_1 \\
&\le e^{\|B\|_{\mathrm{op}}/M}-1
 =\exp\!\left(\frac1M\sum_{i\in D}\eta_i\right)-1\leq \frac{\epsilon}{2}
\end{split}.
\end{equation}
By \cite[Lemma 3]{Bravyi2014MonteCS}, $ \norm{\log S_M^M - (H_z + H_x)}_{\mathrm{op}}\leq \frac{12 R^3}{M^2} $  where $ R = 2(\norm{H_z}_{\mathrm{op} } +\norm{H_x}_{\mathrm{op}}) \leq 2 cn.$ 
This combined with Schatten-Holder and \cref{prop:compare tv distance between gibbs state} give:
\[\left|\Tr(O\rho)- \frac{\Tr(O Q'_M)}{Z_M}\right| = |\Tr(O\rho) - \Tr(O\rho_M)|\leq \norm{\rho -\rho_M}_1 \leq \norm{\log S_M^M - (H_z + H_x)}_{\mathrm{op}} \leq \frac{12 R^3}{M^2}\leq \frac{\epsilon}{2} \]
This together with \cref{eq:helper compare Q with Q'} yields the desired result.
\end{proof}
\subsection{From \cref{sec:tractable}}\label{subsec:defer from tractable}
\begin{proof}[Proof of \cref{prop:1D Ising sampling time}]
Let $\pi \equiv \pi_{\mathbf K,\mathbf w}^{\mathrm{1D Ising, CBC}} $ with $ \mathbf K\in \R^M$ and $ \mathbf w\in(\R \cup \set{\pm\infty})^M.$

Let $ \ket{+1} =\begin{bmatrix} 1 \\ 0\end{bmatrix}, \ket{-1}=\begin{bmatrix} 0 \\ 1\end{bmatrix} .$ 

For $w\in \R,$ let $A(w) =\frac{1}{e^w + e^{-w}}\begin{bmatrix} e^w & 0 \\ 0 & e^{-w}\end{bmatrix} .$ Taking the limit as $ w\to \pm\infty$, we define \[ A(+\infty) = \begin{bmatrix} 1 & 0 \\ 0 & 0\end{bmatrix}, \quad A(-\infty) = \begin{bmatrix} 0 & 0 \\ 0 & 1\end{bmatrix}.\]
For $K\in \R,$ let $B(K) =\frac{1}{e^K + e^{-K}}\begin{bmatrix} e^K & e^{-K} \\  e^{-K} & e^K\end{bmatrix} .$ 

    For $\mathbf{s} = (s^{(k)})_{k=0}^{M-1},$ we can rewrite:
    \[ \pi(\mathbf{s}) =\frac{1}{Z} \prod_{k=0}^{M-1}\left(\bra{s^{(k)}} A(w^{(k)}) \ket{s^{(k)}} \bra{s^{(k)}} B(K^{(k)})  \ket{s^{(k+1)}} \right)\text{ with } s^{(M)} = s^{(0)}, \]
    where $Z = \Tr\left(\prod_{k=0}^{M-1} \left(A(w^{(k)})B(K^{(k)})  \right) \right) .$ For $k\in[0,M-1]\cap \N$, let $ E_k = \prod_{\ell=k}^{M-1} \left(A(w^{(\ell)})B(K^{(\ell)})\right).$ We can compute $ E_0, \cdots, E_{M-1}$ in $O(M)$ time using the recursion 
    \[E_k = A(w^{(k)})B(K^{(k)}) E_{k+1}\, \forall k\in[0,M-1]\cap \N , \quad E_M = \begin{bmatrix} 1 & 0 \\ 0 & 1\end{bmatrix}.\] To sample $ \mathbf{s},$ we sample each coordinate $s^{(k)}$ in $O(1)$ time as follows: 
    \begin{enumerate}
        \item Sample $ s^{(0)}\in \set{\pm 1} $  with probability proportional to $ \bra{s^{(0)}} E_0 \ket{s^{(0)}}$
        \item For $k=1, \cdots, M-1,$ sample $ s^{(k)}$ with probability proportional to \[\bra{s^{(k-1)}} B(K^{(k-1)}) \ket{s^{(k)}} \bra{s^{(k)}} E_k \ket{s^{(0)}}.\]
    \end{enumerate}
    We consider the case $\mathbf K=\mathbf K_{\eta,(p)}$ with $\eta =0.$
We need to sample from $\lineartilt_{\mathbf w} \hat\nu_{0,(p),M}$ where $\hat\nu_{0,(p),M}$
is as in \cref{prop:approximation of path functional}. Explicitly,
\[
\hat\nu_{0,(0),M}
=\frac12\bigl(\delta_{\mathbf1}+\delta_{-\mathbf1}\bigr),
\qquad
\hat\nu_{0,(1),M}
=\frac{1}{2M}\sum_{u\in\{\pm1\}}\sum_{r=1}^{M}
  \delta_{\mathbf v^{u,r}},\footnote{Here $\delta_{\mathbf s} $ is the Dirac delta distribution at $ \mathbf{s}.$}
\quad 
(\mathbf{v}^{u,r})^{(k)}
=
\begin{cases}
u, & k<r,\\
-u, & k\ge r.
\end{cases}
\]
Since $ \hat\nu_{0,(0),M}$ and $ \hat\nu_{0,(1),M}$ have support size at most $2M,$ sampling $\lineartilt_{\mathbf w} \hat\nu_{0,(p),M}$ in $O(M)$ time is straightforward. 

Let $S^{+1} = \set{ k: w^{(k)}=+\infty}, S^{-1} = \set{ k: w^{(k)}=-\infty}, S^0 = \set{k: w^{(k)}\in\R},$ and $ S = S^{+1}\cup S^{-1}.$  

For $r\in [M],$ we compute
$
W_r=\sum_{k\in S^0 \cap [0, r-1]}w^{(k)},$
which takes $O(M)$ time in total.

For $p=0,$ the distribution $\lineartilt_{\mathbf w} \hat\nu_{0,(p),M} $ is well-defined iff either $ S^{+1}$ or $S^{-1}$ is empty; in this case, we sample 
$u\in\{\pm1\}$ with probability proportional
to $e^{uW_M}\mathbbm{1}_{S^{-u}=\emptyset}$ and output $u\mathbf1$.

We let
\[
a_u=\max(S^u\cup\{-1\}),\qquad
b_u=\min(S^{-u}\cup\{M\}).
\]
Computing \(a_u,b_u\) for both signs takes \(O(M)\) time, after which each compatibility test takes \(O(1)\) time.

For $p=1$, the vector $\mathbf v^{u,r}$ equals $u$ before $r$
and $-u$ afterward, so it is compatible with the infinite fields iff
\[
S^{-u}\cap[0,r-1]=\emptyset \text{ and }
S^u\cap[r,M-1]=\emptyset \Leftrightarrow a_u < r\leq b_u.
\]
The distribution $\lineartilt_{\mathbf w} \hat\nu_{0,(1),M} $ is well-defined iff at least one pair $(u,r)$
satisfies these conditions.

In this case, sample $(u,r)\in\{\pm1\}\times[M]$ with probability
proportional to
\[  \exp( \sum_{k\in S^0 \cap [0, M-1]} w^{(k)} (\mathbf v^{u,r})^{(k)}) \mathbbm{1}_{a_u < r\leq b_u} =
e^{u(2W_r-W_M)}
\mathbbm{1}_{a_u < r\leq b_u},
\]
and output $\mathbf v^{u,r}$.
The total runtime in both cases is $O(M).$
\end{proof}
\begin{proof}[Proof of \cref{lem:modify counting to sampling}]
    
    We will use the following result.
\begin{proposition}[{\cite{DF98}, restated in \cite[Theorem 2.2]{SVV09}}]
\label{prop:product of mean}
    Let $W_{1},\ldots,W_{\ell}$ be independent random variables that take positive real values, and satisfy 
$\frac{\E[W_i^{2}] }{\E[W_i]^{2}} \le B$ for $i \in [\ell]$. 
Let $\widehat{W} = W_{1} \cdots W_{\ell}$. 
Let $S_i$ be the average of $\lceil 16B\ell/\epsilon^{2}\rceil $ independent random samples 
from $W_i$ for $i \in [\ell]$. 
Let $\widehat{S} = S_{1} S_{2} \cdots S_{\ell}$. 
Then
\[
\Pr[(1-\epsilon)\E[\widehat{W}]
   \le \widehat{S} \le
   (1+\epsilon)\E[\widehat{W}] ]\ge \tfrac{3}{4}.
\]
\end{proposition}
 Let $\lambda_i = \frac{i}{L}\forall i\in [0,L]\cap \N.$ 
 For $ i\in [L],$ let $W_i$ be the random variable $e^{ V(\sigma)/ L},$ where $\sigma$ is (independently) sampled from $\mu_{\lambda_{i-1}}.$
We have  $\E[W_i]= \E_{\sigma \sim \mu_{\lambda_{i-1}}}[e^{ V(\sigma)/ L}] =\frac{Z(\lambda_{i})}{Z(\lambda_{i-1})}$ and $ \frac{\E[W_i^2]}{\E[W_i]^2} \leq B$ 
 for $B = e^2.$
 Let $\widehat{W} = W_1 W_2 \cdots W_L.$ Then
 \[\E[\widehat{W} ] = \prod_{i=1}^L \E[W_i]= \frac{Z(1)}{Z(0)} = Z(1)\]

 Let $S_i$ be the average of $\lceil 16B L/\epsilon^{2}\rceil$ independent random samples 
from $W_i$ for $i \in [L]$, and let 
$\widehat{S} = \prod_{i\in [L]} S_{i}$. The algorithm outputs $\hat{Z} = \widehat{S}$.
By \Cref{prop:product of mean},
\[ \Pr[(1-\epsilon) Z(1) \leq \hat{Z} \leq (1+\epsilon) Z(1)] =  \Pr[ (1-\epsilon)\E[\widehat{W}]
   \le \widehat{S} \le
   (1+\epsilon)\E[\widehat{W}] ] \geq \frac{3}{4}.  \]
   The total number of samples used is $K = L \lceil 16B L/\epsilon^{2}\rceil = \Theta (L^2/\epsilon^2).$
   Let $ \hat{Z}'$ be the estimate obtained by running the above algorithm with $\mu_{\lambda}$ replaced by $\hat{\mu}_{\lambda} $ where $d_{TV}(\hat{\mu}_{\lambda},\mu_{\lambda} )\leq \frac{1}{8 K}$ for $\lambda= \lambda_i, i\in \set{0,\cdots, L-1}. $ We couple the samples from $ \hat{\mu}_\lambda$ and $ \mu_{\lambda}$ so that, with probability $ \geq 1 - K \cdot \frac{1}{8K},$ all $K$ samples drawn from $ \hat{\mu}_{\lambda}$ are identical to samples from $\mu_{\lambda}$ and thus $ \hat{Z}' = \hat{Z}.$ Hence, with probability $ \geq 5/8,$  $ (1-\epsilon) Z(1) \leq \hat{Z}' \leq (1+\epsilon) Z(1).$ To boost the success probability to $ 1-\delta,$ we take the median of $\Theta(\log(1/\delta))$ independent trials; the proof is via Chernoff's bound. 
\end{proof}

\begin{proof}[Proof of \cref{thm:semilogconcave}]

The proof is a simplified version of the proofs of  \cite[Theorem 64]{AKV24} (the formal version of \cite[Theorem 9]{AKV24}), and  also uses elements from \cite[Theorem 54]{AKV24}. We elaborate in the most detail on the parts which differ.

Let $J'=J\otimes I_N$ and note that $\|J\|_{\mathrm{op}}=\|J'\|_{\mathrm{op}}$.
We view the spin system as a distribution on $(\mathbb S^{N-1})^n$, and let $(\nu_t)_{t\in [0,1]}$ be the stochastic localization process with the time-invariant driving matrix $C_t = (J \otimes I_N)^{1/2}$ initialized at $\nu_0\equiv \nu$. This yields
\[ \frac{d\nu_t}{d\mu}(x) \propto \exp\left(\sum_i \langle (h_t)_i, x_i \rangle + \frac{1 - t}{2}\sum_{i,j} J_{ij} \langle x_i, x_j \rangle\right) \]
for some adapted process $h_t = ((h_t)_i)_{i=1}^n$ in $\mathbb{R}^{Nn}$. 
Recall the trickle-down equation
\[ 	\Cov(\nu_t) = \E[\Cov(\nu_1) \mid \mathcal{F}_t] + \int_t^1 \E[\Cov(\nu_s) J' \Cov(\nu_s) \mid \mathcal{F}_t] ds 
\]
and applying the triangle inequality yields
\[ \|\Cov(\nu_t)\|_{\operatorname{op}} \le \left\|\E[\Cov(\nu_1) \mid \mathcal{F}_t]\right\|_{\operatorname{op}} + \int_t^1 \|\E[\Cov(\nu_s) J' \Cov(\nu_s) \mid \mathcal{F}_t]\|_{\operatorname{op}} ds. \]
Note that
\[
\nu_1=\bigotimes_{i=1}^n\lineartilt_{(h_1)_i}\mu^{(i)},
\qquad
\operatorname{Cov}(\nu_1)
=\bigoplus_{i=1}^n
\operatorname{Cov}\!\left(\lineartilt_{(h_1)_i}\mu^{(i)}\right)
\preceq\rho I.
\]
Hence
\[
\mathbb E[\operatorname{Cov}(\nu_1)\mid\mathcal F_t]\preceq\rho I.
\]



We therefore have the inequality
\[ \|\Cov(\nu_t)\|_{\operatorname{op}} \le \rho + \|J\|_{\operatorname{op}} \int_t^1 \E[\|\Cov(\nu_s)\|_{\operatorname{op}}^2 \mid \mathcal{F}_t] ds, \]
Applying the scalar
comparison as in \cite[Lemma 55]{AKV24} yields
$\|\Cov(\nu_t)\|_{\operatorname{op}} \leq \frac{\rho}{1- \rho \norm{J}_{\operatorname{op}} (1-t)}. $
The ATE for $\nu$ then follows from a similar argument as in the proof of \cite[Theorem 54]{AKV24}, using entropic stability and the supermartingale property.
For an arbitrary tilt $w,$ applying the same localization argument starting at $ \lineartilt_w\nu_t$ and running for time $1-t$ gives $ \|\Cov(\lineartilt_w\nu_t)\|_{\operatorname{op}} \leq \frac{\rho}{1- \rho \norm{J}_{\operatorname{op}} (1-t)}.$ Hence, $\nu_t$ is  $\frac{\rho \norm{J}_{\operatorname{op}}}{1- \rho \norm{J}_{\operatorname{op}} (1-t)} $-entropically stable with respect to the function $\psi(x,y) = \frac{1}{2} \norm{{J'}^{1/2} (x - y)}_2^2$. 
 Therefore, for any nonnegative function $f$,
\begin{align*}
\operatorname{Ent}_{\nu_0}(f)
&\leq
\exp\left(\int_0^1 \frac{\rho \norm{J}_{\operatorname{op}}}{1- \rho \norm{J}_{\operatorname{op}} (1-t)} \,dt\right)
\mathbb{E}\left[\operatorname{Ent}_{\nu_1}(f)\right] \\
&\leq
\exp\left(\int_0^1 \frac{\rho \norm{J}_{\operatorname{op}}}{1- \rho \norm{J}_{\operatorname{op}} (1-t)} \,dt\right)
\sum_i \mathbb{E}\left[
\mathbb{E}_{X_{\sim i}\sim\nu_1}\left[
\operatorname{Ent}_{\nu_1(\cdot\mid X_{\sim i})}(f)
\right]\right] \\
&\leq
\exp\left(\int_0^1 \frac{\rho \norm{J}_{\operatorname{op}}}{1- \rho \norm{J}_{\operatorname{op}} (1-t)} \,dt\right)
\sum_i \mathbb{E}_{X_{\sim i}\sim\nu_0}\left[
\operatorname{Ent}_{\nu_0(\cdot\mid X_{\sim i})}(f)
\right],
\end{align*}
where in the first inequality we used \cite[Proposition 35]{AKV24} (\cite[Proposition 39]{chen2022localization}), the second inequality is tensorization for the product measure $\nu_1$, and the last inequality is the supermartingale property (\cite[Lemma 39]{AKV24}).
Thus $ \nu$ satisfies ATE with constant at most $C = \exp\left(\int_0^1 \frac{\rho \norm{J}_{\operatorname{op}}}{1- \rho \norm{J}_{\operatorname{op}} (1-t)} \,dt\right) = \frac{1}{1- \rho \norm{J}_{\operatorname{op}}}.$ Hence, the Glauber dynamics on $\nu$ contracts KL-divergence by $\left(1-\frac1{Cn}\right).$ 
\end{proof}
\subsubsection{Gibbs state preparation}\label{sec:Gibbs state preparation}

We use the convention from the proof of \cref{prop:derive expression for observables}. We let $ \ket{+1} =\begin{bmatrix} 1 \\ 0\end{bmatrix}, \ket{-1}=\begin{bmatrix} 0 \\ 1\end{bmatrix} .$ 

For $\mathbf{s}=(s_1,\cdots,s_n) \in \set{\pm 1}^n,$ we let $\ket{\mathbf s } = \bigotimes_{i=1}^n\ket{s_i}$ be the corresponding $Z$-basis vector.

We need the following theorem, which follows from \cite[Corollary 3]{wong2026lee}.
\begin{theorem}\label{thm:marginal to state reduction}
Let $\epsilon_{\mathrm{TV}}\in(0,1)$ and let
$H,R_M\in\mathbb C^{2^n\times2^n}$ be Hermitian with
$\|R_M\|_{\mathrm{op}}\leq\epsilon_{\mathrm{TV}}/8$.
Set $\rho=e^H/\Tr(e^H)$ and suppose
$U_M=e^{(H+R_M)/2}$ is entrywise nonnegative in the
$Z$-basis,
$\bra{x}U_M\ket{y}\geq0$ for all
$x,y\in\{\pm1\}^n$.
For partial configurations $a\in\{0,\pm1\}^n$, where $0$
denotes an unspecified spin, define
\[
P_a:=\prod_{i:a_i\neq0}\frac{I+a_iZ_i}{2},
\qquad
q_M(\ell,r):=
\frac{\Tr(P_\ell U_M P_rU_M)}{\Tr(U_M^2)}.
\]
Suppose there exists a randomized classical algorithm that runs in time $T$
and returns nonnegative estimates
$\widehat q_M(\ell,r;\omega)$\footnote{Here $\omega$ denotes
the algorithm's random tape.}
such that for every $\ell,r\in\{0,\pm1\}^n$,
\[
\Pr_\omega\!\left[
\widehat q_M(\ell,r;\omega)
= e^{\pm \epsilon} q_M(\ell,r)
\right]\geq1-\delta,
\qquad
\epsilon=\frac{\epsilon_{\mathrm{TV}}^2}{256n},
\quad
\delta=\frac{\epsilon_{\mathrm{TV}}}{16\cdot4^n}.
\]
Then there exists a quantum algorithm that prepares a state
$\widehat\rho$
such that $\|\widehat\rho-\rho\|_1\leq\epsilon_{\mathrm{TV}}$
in time $\poly(n,T,\log(1/\epsilon_{\mathrm{TV}}))$.\footnote{We assume
that the classical algorithm admits a finite-precision implementation preserving the stated guarantees with
$\poly(n,T,\log(1/\epsilon_{\mathrm{TV}}))$ bit complexity.}
\end{theorem}

\begin{proof}
For two $n$-qubit registers $L,R$, define
\begin{equation}\label{eq:vectorization}
\operatorname{vec}(A)
:=\sum_{x,z\in\{\pm1\}^n}
\langle x|A|z\rangle\,|x\rangle_L\otimes|z\rangle_R,
\qquad
|\Psi_M\rangle
:=\frac{\operatorname{vec}(U_M)}{\sqrt{\Tr(U_M^2)}}.    
\end{equation}
Since $U_M=U_M^\dagger$, this is a normalized purification of
$\rho_M
:=\Tr_R\!\left(|\Psi_M\rangle\langle\Psi_M|\right)
=\frac{U_M^2}{\Tr(U_M^2)}
=\frac{e^{H+R_M}}{\Tr(e^{H+R_M})}.
$
By \cref{prop:compare tv distance between gibbs state}, we have:
\begin{equation}\label{eq:trotter error}
 \|\rho_M-\rho\|_1\leq \|R_M\|_{\mathrm{op}}   
\end{equation}

Order the qubits as $L_1,\ldots,L_n,R_1,\ldots,R_n$.
For a prefix $y=(y_1,\ldots,y_k)\in\{\pm1\}^k$,
where $k\in\{0,\ldots,2n\}$, define
\[
\begin{aligned}
(\ell(y),r(y))
&:=(y_1,\ldots,y_k,0,\ldots,0)\in\{0,\pm1\}^{2n},\\
q(y)&:=q_M(\ell(y),r(y)),\\
\widehat q_\omega(y)
&:=\widehat q_M(\ell(y),r(y);\omega).
\end{aligned}
\]
Thus
\[
q(y)=
\sum_{z\in\{\pm1\}^{2n-k}}
|\langle yz|\Psi_M\rangle|^2.
\]
Sample one random tape $\omega$ and use it for every evaluation.
There are $N=2^{2n+1}-2$ nonempty prefixes, so the event
\[
\mathcal G:=
\left\{
e^{-\epsilon}q(y)\leq\widehat q_\omega(y)
\leq e^\epsilon q(y)
\text{ for every nonempty prefix }y
\right\}
\]
satisfies
$\Pr_\omega(\mathcal G^c)\leq N\delta
\leq\epsilon_{\mathrm{TV}}/8$.

Starting from $|+1\rangle^{\otimes2n}$, successively apply
\[
|y\rangle |+1\rangle
\longmapsto
|y\rangle\sum_{b\in\{\pm1\}}
\sqrt{\widehat p_\omega(b\mid y)}\,|b\rangle,
\qquad
\widehat p_\omega(b\mid y)
:=\frac{\widehat q_\omega(yb)}
{\sum_{c\in\{\pm1\}}\widehat q_\omega(yc)},
\]
taking $\widehat p_\omega(b\mid y)=1/2$ if the denominator is zero.
Each step computes the two child estimates coherently, performs
the rotation, and uncomputes the estimates.
Let $|\widetilde\Psi_\omega\rangle$ be the state obtained after
these $2n$ exact rotations.
Writing $x_{<j}:=(x_1,\ldots,x_{j-1})$, we have
\[
|\widetilde\Psi_\omega\rangle
=\sum_{x\in\{\pm1\}^{2n}}
\left(\prod_{j=1}^{2n}
\widehat p_\omega(x_j\mid x_{<j})\right)^{1/2}|x\rangle.
\]
In $\poly(n,T,\log(1/\epsilon_{\mathrm{TV}}))$ time, we can implement an approximation of these rotations 
with operator-norm error at most \(\epsilon_{\mathrm{TV}}/(16n)\) for each controlled rotation, so that the resulting state $|\widehat\Psi_\omega\rangle$
satisfies:
\begin{equation}\label{eq:implementation error}
    \|\widehat\rho-\widetilde\rho\|_1
\leq
2\mathbb E_\omega
\bigl\|
|\widehat\Psi_\omega\rangle-
|\widetilde\Psi_\omega\rangle
\bigr\|_2
\leq\frac{\epsilon_{\mathrm{TV}}}{4},
\end{equation}
where $\widetilde\rho
:=\mathbb E_\omega\!
\left[\Tr_R
\left(|\widetilde\Psi_\omega\rangle
\langle\widetilde\Psi_\omega|\right)\right],
\widehat\rho
:=\mathbb E_\omega\!
\left[\Tr_R
\left(|\widehat\Psi_\omega\rangle
\langle\widehat\Psi_\omega|\right)\right].$

On $\mathcal G$, zero-mass prefixes are never reached, and
whenever $q(y)>0$,
$
\widehat p_\omega(b\mid y)
=e^{\pm2\epsilon}\frac{q(yb)}{q(y)}.
$

On $\mathcal G$, for every $x\in\{\pm1\}^{2n}$,
$
\langle x|\widetilde\Psi_\omega\rangle
\geq e^{-2n\epsilon}\sqrt{q(x)}
=e^{-2n\epsilon}\langle x|\Psi_M\rangle.
$
Thus
\begin{equation}\label{eq:bound on good event}
\langle\Psi_M|\widetilde\Psi_\omega\rangle
\geq e^{-2n\epsilon},
\qquad
\left\|
|\widetilde\Psi_\omega\rangle\langle\widetilde\Psi_\omega|
-|\Psi_M\rangle\langle\Psi_M|
\right\|_1
\leq4\sqrt{n\epsilon}
=\frac{\epsilon_{\mathrm{TV}}}{4}.    
\end{equation}

Recall that
$\rho_M=\Tr_R|\Psi_M\rangle\langle\Psi_M|$. Contractivity under partial trace gives:
\[
\begin{aligned}
\|\widetilde\rho-\rho_M\|_1
&\leq \left\|\mathbb E_\omega
|\widetilde\Psi_\omega\rangle\langle\widetilde\Psi_\omega|
-|\Psi_M\rangle\langle\Psi_M|
\right\|_1\leq 
\mathbb E_\omega
\left\|
|\widetilde\Psi_\omega\rangle\langle\widetilde\Psi_\omega|
-|\Psi_M\rangle\langle\Psi_M|
\right\|_1\\
&\leq \mathbb E_\omega\left[
\mathbbm{1}_{\mathcal G} \left\|
|\widetilde\Psi_\omega\rangle\langle\widetilde\Psi_\omega|
-|\Psi_M\rangle\langle\Psi_M|
\right\|_1\right]+ \mathbb E_\omega\left[
\mathbbm{1}_{\mathcal G^c} \left\|
|\widetilde\Psi_\omega\rangle\langle\widetilde\Psi_\omega|
-|\Psi_M\rangle\langle\Psi_M|
\right\|_1\right]\\
&\leq
4\sqrt{n\epsilon}
+2\Pr_\omega(\mathcal G^c),
\end{aligned}
\]
where we use \eqref{eq:bound on good event} on $ \mathcal G ,$ and the trivial bound  $\left\|
|\widetilde\Psi_\omega\rangle\langle\widetilde\Psi_\omega|
-|\Psi_M\rangle\langle\Psi_M|
\right\|_1\leq 2$ on $\mathcal G^c$.
This, combined with \cref{eq:trotter error}, \cref{eq:implementation error}, gives:
\[
\|\widehat\rho-\rho\|_1
\leq
\|\widehat\rho-\widetilde\rho\|_1
+\|\widetilde\rho-\rho_M\|_1 +\|\rho_M-\rho\|_1
\leq
\|R_M\|_{\mathrm{op}}
+4\sqrt{n\epsilon}
+\frac{\epsilon_{\mathrm{TV}}}{4}
+2\Pr_\omega(\mathcal G^c)
\leq\epsilon_{\mathrm{TV}}.
\]

\end{proof}

\begin{proof}[Proof of \cref{thm:Gibbs state preparation}]
Define $H_z,H_x$ as in \eqref{eq:Hz and Hx}.
Conjugating by $\prod_{i:\eta_i<0}Z_i$, we may assume
$\eta_i\geq0$; we undo this conjugation on the output state.
 Let $\epsilon,\delta\in (0,1)$ be as in
\cref{thm:marginal to state reduction}.
Let
\[ c  =  \norm{J}_{\infty} +\norm{\mathbf{h}^z}_\infty+ \norm{\boldsymbol{\eta}}_\infty+1,\quad L = n \lceil \norm{J}_\infty \rceil +1,\quad
    M := 2M' ,\quad M'=\lceil 500\cdot \max \left\{\frac{n c^2 }{\epsilon} , \left (\frac{nc}{\epsilon}\right)^{3/2}\right\}\rceil.\]
Let $U_M:=
\left(e^{H_z/(2M)}e^{H_x/(M)}e^{H_z/(2M)}\right)^{M'}
=e^{(H+R_M)/2}.$
By \cite[Lemma 3]{Bravyi2014MonteCS} and the choice of $M$,
$R_M$ is Hermitian and $\|R_M\|_{\mathrm{op}}
\leq\frac{\epsilon_{\mathrm{TV}}}{8}.$ We note also that $U_M$ is the product of matrices with nonnegative entries, and thus has nonnegative entries.
We apply \cref{thm:marginal to state reduction}.

For partial configurations $\ell,r\in \set{0, \pm 1}^n$, let $ P_\ell, P_r$ be as in \cref{thm:marginal to state reduction}, and let:
\[q_M(\ell,r)
=\frac{\Tr(P_\ell U_M P_rU_M)}{\Tr(U_M^2)} = \frac{\Tr(P_\ell \tilde{U}_M P_r \tilde{U}_M)}{\Tr(Q_M)},\text{ } \tilde{U}_M = \left(e^{H_z/M}e^{H_x/M}\right)^{M'}, \text{ } Q_M = \tilde{U}_M^2 = \left(e^{H_z/M}e^{H_x/M}\right)^M. \]
Here we use the fact that $ P_\ell, P_r$ are diagonal in the $Z$-basis and thus commute with $e^{H_z/M}$.

We use the setup in \cref{prop:derive expression for observables}
with $S=\varnothing$, and write
$\pi=\tilde\nu = \bigotimes_{i=1}^n\tilde{\nu}^{(i)}$ with 
\[\tilde{\nu}^{(i)} =\pi_{\mathbf{K}_{\eta_i, (0)}, w_i}^{\mathrm{1D Ising, CBC}}, w_i = (h_i^z/M) \cdot \mathbf{1}.
 \]
We express $q_M(\ell,r)$ in terms of $\pi$ and $W$ using 
the same argument as in the proof of \cref{prop:derive expression for observables}, and approximate $q_M(\ell,r)$ using \cref{thm:main mixing theorem}. 

For $s^{(0)}, s^{(M')} \in \set{\pm 1}^n, $ let \[ \mathbb{I}_{\ell,r,s^{(0)}, s^{(M')}} =\mathbbm{1} \{s_i^{(0)}=\ell_i\ \text{ if } \ell_i\neq0 \text{ and }
s_i^{(M')}=r_i\ \text{ if } r_i\neq0 \text{ for every } i \in [n]\}.\] 

Let \[A_i:=e^{h_i^z Z/M}e^{\eta_i X/M}, \quad
C:=\prod_{i=1}^n\Tr(A_i^M).
\]
We have: $\Tr(Q_M) = C \mathbb E_\pi[e^W]$ and
\begin{equation}\label{eq:express marginal}
\begin{aligned}
   \Tr(P_\ell \tilde{U}_M P_r \tilde{U}_M)&=\sum_{s^{(0)}, s^{(M')} \in \set{\pm 1}^n} \mathbb{I}_{\ell,r,s^{(0)}, s^{(M')}} \bra{\mathbf{s}^{(0)}}\tilde{U}_M \ket{\mathbf{s}^{(M')}}\bra{\mathbf{s}^{(M')}} \tilde{U}_M\ket{\mathbf{s}^{(0)}}\\
 &= \sum_{(\mathbf s^{(k)})_{k=0}^{M}} \mathbb{I}_{\ell,r,s^{(0)}, s^{(M')}}
 \langle \mathbf s^{(M)} | \mathbf s^{(0)}\rangle
 \prod_{k=0}^{M-1}
 \left(
 \bra{\mathbf s^{(k)}}e^{H_z/M}\ket{\mathbf s^{(k)}}
 \bra{\mathbf s^{(k)}}e^{H_x/M}\ket{\mathbf s^{(k+1)}}
 \right)\\
 &= C \sum_{(\mathbf s^{(k)})_{k=0}^{M-1}} \mathbb{I}_{\ell,r,s^{(0)}, s^{(M')}}  \exp\!\left(\frac{1}{2M}\sum_{k=0}^{M-1}\sum_{i,j}
                  J_{ij} \mathbf s^{(k)}_i s^{(k)}_j\right)
\prod_{i=1}^n \tilde{\nu}^{(i)}
\bigl((s_i^{(k)})_{k=0}^{M-1}\bigr)
\end{aligned}    
\end{equation}
Let 
$
E_{\ell,r}:=\left\{( s^{(k)})_{k=0}^{M-1} :
\mathbb{I}_{\ell,r,s^{(0)}, s^{(M')}} = 1
\right\}.
$
Note that 
\[
\pi(E_{\ell,r})
=\prod_{i=1}^n
\frac{\Tr(D_{\ell_i}A_i^{M'} D_{r_i}A_i^{M'})}
     {\Tr(A_i^{2M'})},
\quad D_0:=I,\quad D_{\pm1}:=\frac{I\pm Z}{2}.
\]
We compute $\pi(E_{\ell,r})$ in $O(nM)$ time. Note that if $\pi(E_{\ell,r}) =0,$ then $ q_M(\ell, r) =0,$ so we simply return $\widehat{q_M(\ell,r)}= 0.$  We assume $\pi(E_{\ell,r}) \neq 0,$ which implies the distribution $ \tilde{\rho} \equiv \pi(\cdot\mid E_{\ell,r})$ is well-defined, and
\[ \Tr(P_\ell \tilde{U}_M P_r \tilde{U}_M) = C\pi(E_{\ell,r}) \mathbb E_{\pi(\cdot\mid E_{\ell,r})}[e^W] \text{ thus } q_M(\ell,r) =\pi(E_{\ell,r})
\frac{\mathbb E_{\tilde{\rho}}[e^W]}
     {\mathbb E_\pi[e^W]}.\]
     
We can write $ \tilde{\rho} =\bigotimes_{i=1}^n\tilde{\rho}^{(i)}$ 
where 
\[\tilde{\rho}^{(i)} = \lineartilt_{\tilde{w}_i} \tilde{\nu}^{(i)} = \pi_{\mathbf{K}_{\eta_i, (0)}, \tilde{w}_i + w_i}^{\mathrm{1D Ising, CBC}}, \text{ with }
\tilde{w}_i = (\tilde{w}_i^{(k)})_{k=0}^{M-1}, \tilde{w}_i^{(k)} = \begin{cases} +\infty &\text{ if } \ell_i = +1, k=0 \\ -\infty &\text{ if } \ell_i = -1, k= 0 \\ +\infty &\text{ if } r_i = +1, k= M'\\ -\infty &\text{ if } r_i = -1, k = M' 
\\ 0   &\text{ else }
\end{cases} .\footnote{Recall from \cref{def:1D Ising model CBC} that conditioning is equivalent to taking external fields to $ \pm \infty.$}\]

 We use 
\cref{thm:main mixing theorem} with accuracy parameters $\epsilon/4, \delta/2$ to compute approximations $ \widehat{\mathbb E_{\tilde{\rho}}[e^W]}$ of $\mathbb E_{\tilde{\rho}}[e^W] $ and $\widehat{\mathbb E_\pi[e^W] }$ of $ \mathbb E_\pi[e^W] ,$ and obtain $ \widehat{q_M(\ell,r)  } =\pi(E_{\ell,r}) \frac{\widehat{\mathbb E_{\tilde{\rho}}[e^W]}}{\widehat{\mathbb E_\pi[e^W] }}$ which satisfies:
\[
\Pr\!\left[
\widehat q_M(\ell,r)=e^{\pm\epsilon}q_M(\ell,r)
\right]\geq1-\delta.
\]
We justify the runtime using the same argument as in the proof of \cref{thm:main theorem tractable} (see \eqref{eq:gap condition} and \eqref{eq:total runtime details}).
As in \eqref{eq:gap condition}, our choice of $M$ gives
\[
(1-1/L)\Delta(J)\alpha_{M}(\eta_{\min})
\leq1-\frac1{2L}.
\]
By \eqref{eq:total runtime details}, the estimator runs in time
\[
T=O\!\left(
n^2L^3M\epsilon^{-2}\log(L/\epsilon)\log\delta^{-1}
\right)
=\poly\!\left(
n,\max_{i,j}|J_{ij}|,\max_i|h_i^z|,
\max_i|\eta_i|,\epsilon_{\mathrm{TV}}^{-1}
\right).
\]
In \cref{sec:finite precision}, we show that this estimator admits
a finite-precision implementation that preserves its approximation
guarantees with only polynomial overhead.
\cref{thm:marginal to state reduction} then gives a quantum algorithm that runs in time $ \poly\!\left(
n,\max_{i,j}|J_{ij}|,\max_i|h_i^z|,
\max_i|\eta_i|,\epsilon_{\mathrm{TV}}^{-1}
\right)$\footnote{We assume rational input coefficients of polynomial bit length
and suppress the polynomial overhead of finite-precision arithmetic.} and prepares a quantum state $\widehat\rho$ such that
$\|\widehat\rho-\rho\|_1\leq\epsilon_{\mathrm{TV}}$.
\end{proof}
\subsection{From \cref{sec:lower bound}} \label{subsec:defer from lower bound}
\begin{proof}[Proof of \cref{prop:constant field deviation}]
Let \(\Lambda:=\Lambda_{\nu_\eta}\) and \(\Phi:=\Phi_{b,\eta}\).
\eqref{eq:derivation of observables continuous} gives
\begin{align*}
    \Lambda(r\mathbf{1}) = \log \frac{\Tr(\exp(\eta X + rZ)) }{\Tr(\exp(\eta X))} = \log \frac{\cosh R(r)}{\cosh \eta} 
\end{align*}
since the eigenvalues of $\eta X  + rZ = \begin{bmatrix} r & \eta \\ \eta & -r\end{bmatrix}$ are $ \pm R(r).$ 
Let $\phi(r) =\Phi(r\mathbf{1}) =  \log \frac{\cosh\left( R(r) \right)}{\cosh(\eta)} - \frac{r^2}{2b}.$ Then 
\[  \frac{d}{d r}\Phi(r\mathbf{1})  =  \phi'(r)= r(\frac{\tanh R(r)}{R(r)} - \frac{1}{b}). \]

    The function $g(R) = \frac{\tanh R}{R}$ is strictly decreasing on $(0,\infty)$. 
    
    Indeed, $g'(R) = \frac{R\sech^2(R) -\tanh R}{R^2}.$ Let  $ h(R) :=(R\sech^2(R) -\tanh R).$ Then $h$ vanishes at $0$ and has derivative $ h'(R) = -2R\tanh R\sech^2 R < 0$ for $ R >0$, so $g'(R) < 0$ for $R>0$. 
    Since $ g(\eta) > 1/b$ by \eqref{eq:assumption on b and eta} and $ \lim_{R \to +\infty} g(R)  = 0 < 1/b $, there exists a unique $R_* >\eta$ satisfying $g(R_*) =1/b.$ Let $ m_* = b^{-1}\sqrt{R_*^2 - \eta^2}.$ Then $\phi'(r)|_{r = \pm b m_*} =0,$ and $m_*\in (0,1),$ since $bm_* \leq R(b m_*)=R_* = b \tanh R_* < b.$ 
    


We let
    \[
\psi(R) := \log\frac{\cosh R}{\cosh\eta}
           - \frac{R^2-\eta^2}{2b},
\text{  }
\phi(r) = \psi(R(r)),
\qquad
\psi'(R) = G(R) = \tanh R- \frac{R}{b}.
\]
    Since $G''(R) = -2 \tanh(R) \sech^2(R),$ $G$ is concave on $(0,+\infty).$ Let $ c_0 =\frac{G(\eta)}{R_*-\eta} >0.$ Using concavity and $ G(R_*)=0,$ we have $|G(R)| \geq c_0 |R_*-R|$ for $ R\geq \eta.$
Thus, for $m\geq 0$:
\begin{equation}
  \phi( bm_*) - \phi(b m) = \int_{R(bm)}^{R_*} G(t) d t \geq \frac{c_0}{2} (R(bm)-R_*)^2 \geq c_1' (m-m_*)^2  
\end{equation}
for $c_1' = \frac{c_0 (b^2 m_*)^2}{8 R_*^2}>\frac{\eta}{8} (\frac{b\tanh (\eta)}{\eta} - 1)=c_1 >0$. Here, we use \[\left|R(bm)-R_*\right| = b|m-m_*|\cdot  \frac{bm+bm_*}{R(bm) + R_*}\geq b|m-m_*| \cdot \frac{b m_*}{2 R_*}.\]
Indeed, writing $x=bm$ and $u=bm_*$, and using $\sqrt{\eta^2+x^2}\le\eta+x$ and $u\le\eta+R_*,$ we have:
\[
\frac{x+u}{\sqrt{\eta^2+x^2}+R_*}
\ge \frac{x+u}{x+\eta+R_*}
\ge \frac{u}{\eta+R_*}
\ge \frac{u}{2R_*}.\qedhere\]
\end{proof}
\begin{proof}[Proof of \cref{prop:sign change}]
Let $\nu\equiv \nu_{\eta},$ and  let \[ A = \E_{\nu}[e^{\langle |f|, \sigma\rangle}] = \E_{\nu}[\cosh(\langle |f|, \sigma\rangle)], \quad B =  \E_{\nu}[e^{\langle f, \sigma\rangle}]=\E_{\nu}[\cosh(\langle f, \sigma\rangle)] ,\] where we use the fact that $d\nu (\sigma) = d\nu(-\sigma).$ We will show:
\begin{equation}\label{eq:difference between exponential}
   A- B = \E_{\nu}[\cosh(\langle |f| ,\sigma\rangle) - \cosh(\langle f, \sigma\rangle)] \geq  \sech(\eta) \cdot \left(\cosh(\norm{f}_1) -\cosh\left(\int f(t) dt\right)\right)  \geq 0   
\end{equation}
\eqref{eq:difference between exponential} together with $\log A - \log B = \int_{B}^A \frac{1}{t} dt\geq \frac{A-B}{A}$ and $e^{-a}(\cosh(a) -\cosh(b)) = \frac{(1-e^{b-a}) (1- e^{-b-a})}{2} $ gives:
\begin{align*}
    \Lambda (|f|) - \Lambda(f) &\geq \sech \eta
    \cdot \exp (- \norm{f}_1) \left(\cosh(\norm{f}_1) -\cosh\left(\int f(t) dt\right)\right) \\
    &= \frac{\sech(\eta)}{2}\cdot (1-e^{\left|\int_{\mS} f dt\right| - \norm{f}_1  } )( 1- e^{-\left|\int_{\mS} f dt\right| -\norm{f}_1} )
\end{align*}
Note that $ 1-e^{-\left|\int_{\mS} f dt\right| - \norm{f}_1  } \geq 1 - e^{-\norm{f}_1} \geq 1 - e^{-f_{\min}}\geq 0. $
    This combined with the inequality $ 1- e^{-x} = \int_0^x e^{-t} dt \geq x e^{-x}$, where $ x :=\norm{f}_1-\left|\int_{\mS} f dt\right| \in [0, f_{\max}] $, gives
    \[ (1-e^{-\left|\int_{\mS} f dt\right| - \norm{f}_1  }) (1-e^{\left|\int_{\mS} f dt\right| - \norm{f}_1  } ) \geq \left(\norm{f}_1-\left|\int_{\mS} f dt\right|\right) e^{-f_{\max}} (1 - e^{-f_{\min}}). \]
To finish the proof, we prove \eqref{eq:difference between exponential}. To do so, we show that for every even integer $d\geq 0$, and $\rho = \nu(\cdot | \, |\mathcal{J}(\sigma)| = d )$:
\begin{equation}\label{eq:helper sign change conditioned on number of jump}
   \Delta_d := \mathbb{E}_{\rho}\left[ \cosh(\langle |f|, \sigma \rangle) - \cosh(\langle f, \sigma \rangle) \right] \geq 0 
\end{equation}
then conclude that:
\begin{align*}
   A-B= \mathbb{E}_{\nu}\left[ \cosh(\langle |f|, \sigma \rangle) - \cosh(\langle f, \sigma \rangle) \right] 
    &\geq \Pr_{\nu}\big[ |\mathcal{J}(\sigma)| = 0 \big] \cdot \Delta_0 \\
    &= \sech(\eta) \cdot \left(\cosh(\norm{f}_1) -\cosh\left(\int f(t) dt\right)\right) 
\end{align*}
Using the Taylor expansion $\cosh x = \sum_{k\geq 0} \frac{x^{2k}}{(2k)!},$ we have:
\begin{align*}
    \Delta_d 
    = & \mathbb{E}_{\rho}\left[ \sum_{k\geq 0}\frac{1}{(2k)!} \cdot\left\{ \left(\int |f(t)| \sigma(t) dt \right)^{2k} -\left(\int f(t) \sigma(t) dt \right)^{2k}\right\}\right]\\
  = & \sum_{k\geq 0} \frac{1}{(2k)!}\cdot \int  \left(\prod_{i=1}^{2k}|f(t_i)| -\prod_{i=1}^{2k}f(t_i)\right) \E_{\rho}\left[\prod_{i=1}^{2k}\sigma(t_i)\right] dt_1\cdots d t_{2k}
\end{align*}
The exchanges of summation, expectation, and integration are justified by Fubini's theorem, since
\begin{align*}
    &\sum_{k\geq 0}   \frac{1}{(2k)!}\cdot \E_{\rho}\left[ \int\left| \left(\prod_{i=1}^{2k}|f(t_i)| -\prod_{i=1}^{2k}f(t_i)\right) \prod_{i=1}^{2k}\sigma(t_i) \right|dt_1\cdots d t_{2k}  \right] \\
    \leq&2 \sum_{k\geq 0}   \frac{1}{(2k)!}\cdot \int \prod_{i=1}^{2k}|f(t_i)| dt_1\cdots d t_{2k}
    \\
    = &
\sum_{k\geq 0}\frac{2\|f\|_1^{2k}}{(2k)!}
=2\cosh(\|f\|_1)\leq 2\cosh(\|f\|_2) < \infty.
\end{align*}
Hence, it suffices to show $ \E_{\rho} \left[ \prod_{i=1}^{2k}\sigma(t_i) \right]\geq 0. $ Wlog assume that $ t_1 \leq t_2 \leq \cdots \leq t_{2k}.$ Note that by the construction of $\sigma$
\[  \prod_{i=1}^{2k}\sigma(t_i) =\prod_{i=1}^k \left(\sigma(t_{2i-1})\sigma(t_{2i})\right) = \prod_{i=1}^k (-1)^{|\mathcal{J}(\sigma) \cap (t_{2i-1}, t_{2i}] |}  = (-1)^{ |\mathcal{J}(\sigma) \cap \bigcup_{i=1}^k(t_{2i-1}, t_{2i}]|}\]
Note that $ |\mathcal{J}(\sigma) \cap \bigcup_{i=1}^k(t_{2i-1}, t_{2i}]| \sim \mathrm{Bin}(d, p)$ with $p = \sum_{i=1}^k (t_{2i}-t_{2i-1}).$ Hence 
\[ \E_{\rho}\left[\prod_{i=1}^{2k}\sigma(t_i)\right] = \sum_{r=0}^d (-1)^r\binom{d}{r}  p^r (1-p)^{d-r} = (1-2p)^d \geq 0. \]
\end{proof}

  \begin{proof}[Proof of \eqref{eq:justify edge interaction}]
  Let
 \(a=bm_*>0\). 
By \eqref{eq:derivation of observables continuous}
\[
F(v)=(2\cosh\eta)^2 \mathbb E_{\sigma,\tau\sim\nu_\eta}\left[
\exp\!\left(a\langle\sigma, \mathbf{1}\rangle +a\langle\tau, \mathbf{1}\rangle 
+v\langle\sigma,\tau\rangle\right)\right]= (2\cosh\eta)^2
\mathbb E_{\sigma,\tau\sim\nu_\eta}\left[
e^{aA_\sigma+aA_\tau
+v B}\right],
\]
where \(A_\sigma= \langle\sigma, \mathbf{1}\rangle=\int_{\mathbb S}\sigma(t)\,dt\),  $B:=\langle\sigma,\tau\rangle
  =\int_{\mathbb S}\sigma(t)\tau(t)\,dt$, and
\(\sigma\) and \(\tau\) are independent. 

Fix $v >0.$
Subtracting the expressions for \(F(v)\) and \(F(-v)\), and using the Taylor series for $\sinh(x)$, we obtain
\begin{align}
\frac{F(v)-F(-v)}{(2\cosh\eta)^2}
&=\mathbb E_{\sigma,\tau}\left[
e^{aA_\sigma+aA_\tau}(e^{vB}-e^{-vB})
\right] =2\sum_{\substack{k\ge1\\k\text{ odd}}}
\frac{v^k}{k!}\,
\mathbb E_{\sigma,\tau}\left[
e^{aA_\sigma+aA_\tau}B^k
\right].
\label{eq:F-odd-expansion}
\end{align}
For each \(k\ge1\), we have:
\begin{align*}
B^k
&=\left(\int_{\mathbb S}\sigma(t)\tau(t)\,dt\right)^k=\int_{\mathbb S^k}
\prod_{j=1}^k \sigma(t_j)\tau(t_j)\,
dt_1\cdots dt_k.
\end{align*}
Consequently, by Fubini's theorem,
\begin{align*}
\mathbb E_{\sigma,\tau}\left[
e^{aA_\sigma+aA_\tau}B^k
\right]
&\quad=
\int_{\mathbb S^k}
\mathbb E_{\sigma,\tau}\left[
\left(e^{aA_\sigma}\prod_{j=1}^k\sigma(t_j)\right)
\left(e^{aA_\tau}\prod_{j=1}^k\tau(t_j)\right)
\right]dt_1\cdots dt_k\\
&\quad=
\int_{\mathbb S^k}
\mathbb E_{\nu_\eta}\left[
e^{aA_\sigma}\prod_{j=1}^k\sigma(t_j)
\right]
\mathbb E_{\nu_\eta}\left[
e^{aA_\tau}\prod_{j=1}^k\tau(t_j)
\right]dt_1\cdots dt_k\\
&\quad=
\int_{\mathbb S^k}
\left(
\mathbb E_{\nu_\eta}\left[
e^{aA_\sigma}\prod_{j=1}^k\sigma(t_j)
\right]\right)^2dt_1\cdots dt_k.
\end{align*}
The second equality uses independence, and the third uses the fact
that \(\sigma\) and \(\tau\) have the same law. Substituting into
\eqref{eq:F-odd-expansion} gives
\[
\frac{
F(v)-F(-v)}{
2(2\cosh\eta)^2} =
\sum_{\substack{k\ge1\\ k\text{ odd}}}\frac{v^k}{k!}
\int_{\mathbb S^k}
\left(
\mathbb E_{\nu_\eta}\!\left[
e^{aA_\sigma}\prod_{j=1}^k\sigma(t_j)
\right]\right)^2
dt_1\cdots dt_k 
\]
The exchanges of summation, expectation, and integration above are justified by Fubini's theorem and the fact that
\(|A_\sigma|,|A_\tau|,|B|\le1\). In particular, \[\int_{\mathbb S^k}\left|
\mathbb E_{\sigma,\tau}\left[
\left(e^{aA_\sigma}\prod_{j=1}^k\sigma(t_j)\right)
\left(e^{aA_\tau}\prod_{j=1}^k\tau(t_j)\right)
\right]\right| dt_1\cdots dt_k \leq e^{2a+v} \] and \[\sum_{\substack{k\ge1\\ k\text{ odd}}}\frac{v^k}{k!}
\int_{\mathbb S^k}
\left|
\mathbb E_{\nu_\eta}\!\left[
e^{aA_\sigma}\prod_{j=1}^k\sigma(t_j)
\right]\right|^2
dt_1\cdots dt_k\leq e^{2a+v}  .\] 

The summand corresponding to each $k$ is nonnegative. The \(k=1\) term is strictly positive, since by Holder's inequality $\int_{\mathbb S}\mathbb E_{\nu_\eta}
[e^{aA_\sigma}\sigma(t)]^2 \,dt \geq (\int_{\mathbb S}\mathbb E_{\nu_\eta}
[e^{aA_\sigma}\sigma(t)]\,dt)^2$ and
\[
\int_{\mathbb S}\mathbb E_{\nu_\eta}
[e^{aA_\sigma}\sigma(t)]\,dt
=\mathbb E_{\nu_\eta}[A_\sigma e^{aA_\sigma}]
=\mathbb E_{\nu_\eta}[A_\sigma\sinh(aA_\sigma)]   >0,
\]
where the second equality follows from spin-flip symmetry, i.e., $\nu_\eta(d\sigma)= \nu_\eta(d(-\sigma)),$ and the inequality is by
\[ \mathbb E_{\nu_\eta}[A_\sigma\sinh(aA_\sigma)] \geq \E_{\nu_\eta}[A_\sigma\sinh(aA_\sigma) \mid |\mathcal{J}(\sigma)| = 0]  \cdot \Pr_{\nu_\eta}\big[ |\mathcal{J}(\sigma)| = 0 \big] = \frac{\sinh (a)}{\cosh \eta } > 0\]
where we recall the definition of $ \mathcal{J}(\sigma)$ from \cref{def:general path definition}.
 Hence \(F(v)>F(-v)\) for every \(v>0\).
\end{proof}

\section{Finite precision}\label{sec:finite precision}
Assume the notation of \cref{thm:main theorem tractable}.
Suppose all model parameters $J_{ij}, h_i^z, \eta_i$ are rational numbers with bit length bounded by $b.$ 
Let $\Lambda = \max \{n, \max_{i,j} |J_{ij}|, \max_i |h_i^z|, \max_i |\eta_i| \}.$ 

We show how to implement the algorithms for approximating $ \Tr(X_S e^H)$ and $ \Tr(O \rho)$ from \cref{thm:main theorem tractable} and the subroutines to approximate $q_M(\ell, r) $ in \cref{thm:Gibbs state preparation} 
with $ \poly(b, \Lambda,\epsilon^{-1}) \cdot \log \delta^{-1} $ cost in the word-RAM model. This allows implementing the quantum state-preparation
algorithm in the quantum circuit model over a fixed finite
universal gate set with $ \poly(b, \Lambda,\epsilon_{TV}^{-1})$ cost. 

We use the same notation and setup as in the proof of \cref{thm:main theorem tractable}; in particular, we assume $\eta_i \geq 0\forall i$, and for the algorithms from \cref{thm:main theorem tractable} and \cref{thm:Gibbs state preparation},
choose $M, L$ as in the corresponding proofs.
These algorithms and subroutines call the sampler from \cref{thm:main mixing theorem} part 1, which in turn calls the sampler from \cref{prop:1D Ising sampling time} with input $ \tilde{\mathbf{K}} = \mathbf{K}_{\eta_i, (p_i)}$ from \cref{prop:approximation of path functional}, and $\tilde{\mathbf{w}}$ where \[\forall k: \text{ either } \tilde{w}^{(k)}\in \set{\pm \infty} \text{  or } |\tilde{w}^{(k)}| \leq O(\Lambda)\]
and $\forall k: |\tilde{K}^{(k)}| = K_M(\eta_i) .$ If $\eta_i =0$, $ K_M(\eta_i) =+\infty$, else $ K_M(\eta_i) =\frac{1}{2}\log \coth(\frac{\eta_i}{M})\leq \frac{1}{2}\log \frac{2M}{\eta_i} \leq O(b + \log M)$ by \cite[Proposition 18]{SV26}, and the facts that $ \left|\frac{\eta_i}{M}\right|\leq 1$ by the choice of $M$ and that $\eta_i\neq 0$ has bit length bounded by $b.$ 

For $\varepsilon_{TV}\in (0,1)$, we sample from $\hat{\pi}_{\tilde{\mathbf K},\tilde{\mathbf w}}^{\mathrm{1D Ising, CBC}} $ s.t. $ d_{TV} (\hat{\pi}_{\tilde{\mathbf K},\tilde{\mathbf w}}^{\mathrm{1D Ising, CBC}}, \pi_{\tilde{\mathbf K},\tilde{\mathbf w}}^{\mathrm{1D Ising, CBC}}) \leq \varepsilon_{TV}$ 
by rounding the arithmetic operations in the sampler from \cref{prop:1D Ising sampling time}\footnote{In the sampler from \cref{prop:1D Ising sampling time}, all arithmetic used to compute probabilities involves finite
real numbers: infinite external fields are handled exactly, and the
case $\eta_i=0$ is treated separately.}
using $O(B)$ bits of precision, with \[B = O\left( M (b +\log M + \Lambda)  + \log \frac{M}{\varepsilon_{TV}}\right).\]
Set $\varepsilon_{TV}=\frac{\epsilon_{TV}}{4N}$ with $N = O(nL \log (L/\epsilon_{TV})).$ Recall that $N$ gives an upper bound for the number of calls to the sampler from \cref{prop:1D Ising sampling time} by the (approximate) sampler for $\mu_\lambda$ with $\lambda \in \set{\frac{i}{L}\mid i\in [0,L-1]\cap \N}$ from \cref{thm:main mixing theorem} part 1. Hence, by a standard coupling argument, the above rounding procedure, plus other standard approximations, e.g., for sampling $i\sim \text{Uniform}([n])$, gives an algorithm that samples from $\hat{\mu}_\lambda$ with $d_{TV} (\hat{\mu}_\lambda, \mu_\lambda) \leq \epsilon_{TV}$ by operating on $O(B)$-bit numbers and having runtime bounded by
\[
O\!\left(
n^2ML\log\frac{2L}{\epsilon_{TV}}\,
\operatorname{poly}(B)
\right)
=\operatorname{poly}(b,\Lambda,M,L,\epsilon_{TV}^{-1}).
\]
Plugging this sampler into the algorithms in \cref{thm:main theorem tractable} and subroutines in \cref{thm:Gibbs state preparation}\footnote{We use a similar rounding procedure to approximate $D$, $C$, $\pi(E_{\ell,r}),$ and $e^{W/L}$ with relative accuracy; the case $\pi(E_{\ell,r})=0$ can be detected exactly.} gives the desired result.
\end{document}